\documentclass[11pt]{article}
\usepackage{vmargin}
\usepackage{latexsym}
\usepackage{amsmath,amssymb,amsthm,amsfonts}
\usepackage{times}
\usepackage{ifthen}
\usepackage{float}
\usepackage{hyperref}
\usepackage{xspace}

\usepackage{graphicx}
\usepackage{caption2}

\setpapersize[portrait]{USletter}
\setmarginsrb{1in}{1in}{1in}{1in}{0in}{0in}{3pt}{16pt}

\newlength\abovesectionskip
\newlength\belowsectionskip
\def\sectionfont{\normalfont\Large\bfseries}
\newlength\abovesubsectionskip
\newlength\belowsubsectionskip
\def\subsectionfont{\normalfont\large\bfseries}
\newlength\abovesubsubsectionskip
\newlength\belowsubsubsectionskip
\def\subsubsectionfont{\normalfont\normalsize\bfseries}
\newlength\aboveparagraphskip
\newlength\belowparagraphskip
\def\paragraphfont{\normalfont\normalsize\bfseries}
\makeatletter
\def\section{\@startsection{section}{1}{\z@}{-\abovesectionskip}%
               {\belowsectionskip}{\sectionfont}}
\def\subsection{\@startsection{subsection}{2}{\z@}{-\abovesubsectionskip}%
                  {\belowsubsectionskip}{\subsectionfont}}
\def\subsubsection{\@startsection{subsubsection}{3}{\z@}%
                     {-\abovesubsubsectionskip}{\belowsubsubsectionskip}%
                     {\subsubsectionfont}}
\def\paragraph{\@startsection{paragraph}{4}{\z@}{-\aboveparagraphskip}%
                 {-\belowparagraphskip}{\paragraphfont}}
\makeatother

\makeatletter
\renewenvironment{align*}{%
  \abovedisplayskip 3pt plus 1pt%
  \belowdisplayskip 3pt plus 1pt%
  \start@align\@ne\st@rredtrue\m@ne
}{%
  \endalign
}
\let\stdequation\equation
\renewcommand*\equation{%
  \abovedisplayskip 3pt plus 1pt%
  \belowdisplayskip 4pt plus 1pt%
  \stdequation}
\DeclareRobustCommand{\[}{
  \abovedisplayskip 3pt plus 1pt%
  \belowdisplayskip 4pt plus 1pt%
  \begin{equation*}
}
\makeatother

\renewenvironment{itemize}{
    \begin{list}{$\bullet$}{
        \setlength{\labelsep}{3mm}\setlength{\itemindent}{0mm}\setlength{\labelwidth}{3mm}
        \setlength{\leftmargin}{12mm}
        \setlength{\itemsep}{1pt}\setlength{\parsep}{0mm}
        \setlength{\topsep}{0mm}\setlength{\listparindent}{0pt}
    }
}
{
    \end{list}
}

\floatstyle{ruled}
\newfloat{algorithm}{t}{loa}
\floatname{algorithm}{Algorithm}
\newcommand{\mycomment}[1]{}
\newcommand{\comment}[1]{}

\numberwithin{equation}{section} 
\numberwithin{figure}{section} 

\newcommand{\thmabove}{3pt}
\newcommand{\thmbelow}{2pt}
\newtheoremstyle{mythmstyle}
  {\thmabove}   
  {\thmbelow}   
  {\itshape}    
  {}            
  {\bfseries}   
  {. }          
  {2.5pt}       
  {\thmname{#1}\thmnumber{ #2}\thmnote{ \normalfont (#3)}}   
\theoremstyle{mythmstyle}

\newcommand{\nc}{\newcommand}
\newcommand{\nt}{\newtheorem}
\nt{theorem}{Theorem}
\nt{definition}{Definition}
\nt{lemma}{Lemma}
\nt{proposition}{Proposition}
\nt{claim}{Claim}
\nt{fact}{Fact}
\nt{corollary}{Corollary}

\newcommand{\proofbelow}{4pt}
\renewenvironment{proof}{\noindent\textbf{Proof.}\,}{\afterproof}
\newenvironment{proofof}[1]{\vspace{\proofbelow}\noindent\textbf{Proof} \,(of #1).\,}{\afterproof}
\newenvironment{proofsketch}{\vspace{\proofbelow}\noindent\textbf{Proof Sketch.}\,}{\afterproof}
\newenvironment{proofsketchof}[1]{\vspace{\proofbelow}\noindent\textbf{Proof Sketch} \,(of #1).\,}{\afterproof}
\newenvironment{subproof}{\noindent\textit{Proof.}\,}{\aftersubproof}
\newcommand{\afterproof}{\hfill $\blacksquare$ \par \vspace{\proofbelow}}
\newcommand{\aftersubproof}{\hfill $\Box$ \par \vspace{\proofbelow}}

\newcommand{\repeatclaim}[2]{\vspace{6pt}\medskip\noindent\textbf{#1. }{\it #2} \medskip}

\nc{\ignore}[1]{}

\newcommand{\E}{{\bf E}}

\newcommand{\hE}{\hat{{\bf E}}}

\newcommand{\hclass}{C}
\newcommand{\vcdim}{D}
\newcommand{\hyp}{h}

\newcommand{\err}{\mathrm{err}}
\newcommand{\herr}{\widehat{\mathrm{err}}}

\newcommand{\hf}{\hat{f}}

\newcommand{\indices}{\cB}

\newcommand{\indicest}{J}

\newcommand{\reals}{\mathbb{R}}

\newcommand{\targetf}{{f^*}}
\newcommand{\ind}{\chi}
\newcommand{\trainS}{{\cal S}}
\newcommand{\NOR}{\textsc{Nor}\xspace}
\newcommand{\X}{{\cal X}}

\newcommand{\sgn}{{\operatorname{sgn}}}

\newcommand{\ground}{{[n]}}

\newcommand{\rankf}{\mathrm{rank}}

\newcommand{\fclass}{{\cal F}}

\newcommand{\rad}{R}

\newcommand{\hrad}{\hat{R}}

\newcommand{\Ind}{{\cal I}}

\usepackage{color}
\definecolor{darkblue}{rgb}{0,0.1,0.5}
\definecolor{violet}{rgb}{0.3,0.0,0.5}

\newcommand{\cSz}{{\mathcal{S}}_{0}}
\newcommand{\cSnz}{{\mathcal{S}}_{\neq 0}}

\newcommand{\Pone}{\textbf{P1}}
\newcommand{\Ptwo}{\textbf{P2}}
\newcommand{\Pthree}{\textbf{P3}}

\allowdisplaybreaks[3]      

\newcommand{\abs}[1]{\lvert #1 \rvert}

\newcommand{\card}[1]{\abs{#1}}

\newcommand{\intersect}{\cap}
\newcommand{\union}{\cup}
\newcommand{\Union}{\bigcup}
\newcommand{\set}[1]{\left \{ #1 \right \}}                     
\newcommand{\setst}[2]{\left\{\; #1 \,:\, #2 \;\right\}}        

\newcommand{\transpose}{^{\textsf{T}}}
\renewcommand{\And}{~\wedge~}
\newcommand{\defeq}{\,:=\,}                                     
\renewcommand{\th}{\ifmmode{^{\textrm{th}}}\else{\textsuperscript{th}\ }\fi}
\newcommand{\smallfrac}[2]{{\textstyle \frac{#1}{#2}}}
\newcommand{\norm}[1]{\left\lVert #1 \right\rVert}
\newcommand{\argmax}{\operatornamewithlimits{argmax}}

\newcommand{\bF}{\mathbb{F}}        
\newcommand{\bN}{\mathbb{N}}        
\newcommand{\bR}{\mathbb{R}}
\newcommand{\bZ}{\mathbb{Z}}

\newcommand{\cA}{\mathcal{A}}
\newcommand{\cB}{\mathcal{B}}
\newcommand{\cC}{\mathcal{C}}
\newcommand{\cE}{\mathcal{E}}
\newcommand{\cF}{\mathcal{F}}
\newcommand{\cI}{\mathcal{I}}

\newcommand{\cL}{\mathcal{L}}
\newcommand{\cM}{\mathcal{M}}
\newcommand{\cP}{\mathcal{P}}
\newcommand{\cS}{\mathcal{S}}
\newcommand{\cU}{\mathcal{U}}
\newcommand{\cY}{\mathcal{Y}}
\newcommand{\cZ}{\mathcal{Z}}

\newcommand{\tO}{\tilde{O}}
\newcommand{\tOmega}{\tilde{\Omega}}

\newcommand{\prob}[1]{\operatorname{Pr}\left[\,#1\,\right]}               
\newcommand{\probg}[2]{\operatorname{Pr}\left[\,#1 \:\mid\: #2\,\right]}  
\newcommand{\probover}[2]{\operatorname{Pr}_{#1}\left[\,#2 \,\right]}     
\newcommand{\expect}[1]{\operatorname{\bf E}\left[\,#1\,\right]}          
\newcommand{\expectover}[2]{\operatorname{\bf E}_{#1}\left[\,#2 \,\right]}     

\newcommand{\Algorithm}[1]{Algorithm~\ref{alg:#1}}
\newcommand{\AlgorithmName}[1]{\label{alg:#1}}
\newcommand{\Appendix}[1]{Appendix~\ref{app:#1}}
\newcommand{\AppendixName}[1]{\label{app:#1}}
\newcommand{\Claim}[1]{Claim~\ref{clm:#1}}
\newcommand{\ClaimName}[1]{\label{clm:#1}}
\newcommand{\Corollary}[1]{Corollary~\ref{cor:#1}}
\newcommand{\CorollaryName}[1]{\label{cor:#1}}

\newcommand{\DefinitionName}[1]{\label{def:#1}}
\newcommand{\Equation}[1]{Eq.\:\eqref{eq:#1}}
\newcommand{\EquationName}[1]{\label{eq:#1}}
\newcommand{\Lemma}[1]{Lemma~\ref{lem:#1}}
\newcommand{\LemmaName}[1]{\label{lem:#1}}
\newcommand{\Proposition}[1]{Proposition~\ref{prop:#1}}
\newcommand{\PropositionName}[1]{\label{prop:#1}}
\newcommand{\Section}[1]{Section~\ref{sec:#1}}
\newcommand{\SectionName}[1]{\label{sec:#1}}

\newcommand{\Theorem}[1]{Theorem~\ref{thm:#1}}

\newcommand{\neigh}{\Gamma}

\newcommand{\mat}{\mathbf{M}}
\newcommand{\poly}{\operatorname{poly}}

\newcommand{\newterm}[1]{\textit{#1}}

\usepackage{color}

\newcommand{\FigureName}[1]{\label{fig:#1}}

\newcommand{\smallsum}[2]{{\textstyle \sum_{#1}^{#2}}}
\newcommand{\TheoremName}[1]{\label{thm:#1}}
\newcommand{\ALG}{\mathcal{ALG}}

\newcommand{\demandSet}[1]{\mathcal{D}^{#1}}

\begin{document}

\title{Submodular Functions: Learnability, Structure, and Optimization\thanks{A preliminary version of this paper appeared in the 43rd ACM Symposium on Theory of Computing under the title ``Learning Submodular Functions''.}}
\author{
Maria-Florina Balcan\thanks{
    Georgia Institute of Technology, School of Computer Science.
    Email: \texttt{ninamf@cc.gatech.edu}.   }
\and
Nicholas J. A. Harvey\thanks{
    University of British Columbia.
    Email: \texttt{nickhar@cs.ubc.ca}.
    }
}

\date{} \maketitle


\begin{abstract}%
Submodular functions are discrete functions that model laws of diminishing returns and enjoy
numerous algorithmic applications. They have been used in many areas, including combinatorial optimization, machine learning, and economics.
In this work we study submodular functions from a learning theoretic angle.
We provide algorithms for learning submodular functions, as well as lower bounds on their
learnability.
In doing so, we uncover several novel structural results revealing
ways in which submodular functions can be both surprisingly structured
and surprisingly unstructured.
We provide several concrete implications of our work in other domains including algorithmic game
theory and combinatorial optimization.

At a technical level, this research combines ideas from many areas,
including learning theory (distributional learning and PAC-style analyses),
combinatorics and optimization (matroids and submodular functions),
and pseudorandomness (lossless expander graphs).
%

\end{abstract}

\maketitle

\section{Introduction}
\SectionName{intro}

Submodular functions are a discrete analog of convex functions that enjoy numerous applications and have
structural properties that can be exploited algorithmically. They arise naturally in the study of graphs, matroids,
covering problems, facility location problems, etc., and they have been extensively studied in operations
research and combinatorial optimization for many years ~\cite{EdmondsSubmodular}.
More recently, submodular functions have
become key concepts in other areas including machine learning,  algorithmic game theory, and social sciences.
For example, submodular functions have been used
to model bidders' valuation functions in combinatorial
auctions~\cite{HansonMartin,LLN06,DNS06,BBM08,V08}, and
for solving several machine learning problems, including feature selection problems in graphical models~\cite{KG05}
and various clustering problems~\cite{bilmes09}.

In this work we use a learning theory perspective to uncover new structural properties of submodular functions.
In addition to providing algorithms and lower bounds for learning submodular functions, we discuss
numerous implications of our work in algorithmic game theory, economics, matroid theory
and combinatorial optimization.

One of our foremost contributions is to provide the first known results about learnability of
submodular functions in a distributional (i.e., PAC-style) learning setting. Informally, such a
setting has a fixed but unknown submodular function $\targetf$
and a fixed but unknown distribution over the domain of $\targetf$.
The goal is to design an efficient algorithm which
provides a good approximation of $\targetf$ with respect to that
distribution, given only a small number of samples from the distribution.

Formally, let $[n] = \set{1,\ldots,n}$ denote a ground set of items and let $2^{[n]}$ be the power
set of $[n]$.
A function $f : 2^{[n]} \rightarrow \bR$ is submodular if it satisfies
$$
f(T \cup \set{i})-f(T) ~~\leq~~ f(S \cup \set{i})-f(S)
    \qquad\forall S \subseteq T \subseteq \ground ,\, i \in \ground.
$$
The goal is to output a
function $f$ that, with probability $1-\delta$ over the
samples,
is a good approximation of
$\targetf$ on most of the sets coming from the distribution.
Here ``most'' means a $1-\epsilon$ fraction and ``good approximation'' means that
$f(S) \leq \targetf(S) \leq \alpha \cdot f(S)$ for some approximation factor $\alpha$.
We prove nearly matching $\alpha=O(n^{1/2})$ upper and  $\alpha=\tOmega(n^{1/3})$ lower bounds on the approximation factor
achievable when the algorithm receives only $\poly(n,1/\epsilon,1/\delta)$ examples from an
arbitrary (fixed but unknown) distribution.
We additionally provide a learning algorithm with constant approximation factor for the case
that the underlying distribution is a product distribution. This is based on a new result
proving strong concentration properties of submodular functions.


To prove the $\tOmega(n^{1/3})$ lower bound for learning under  arbitrary distributions,
we construct a new family of matroids whose rank functions are fiendishly unstructured.
Since matroid rank functions are submodular, this  shows unexpected extremal properties of submodular functions and gives new insights into their complexity. This construction also provides a general tool for proving lower bounds
in several areas where submodular functions arise. We derive and discuss such implications in:

 \begin{itemize}
 \item Algorithmic Game Theory and Economics:  An important consequence of our construction is that matroid rank functions  do not have a ``sketch'',
i.e., a concise, approximate representation. As matroid rank functions are known to satisfy the
gross substitutes property~\cite{kazuo-book}, our work implies that  gross substitutes functions also do
not have a concise, approximate representation. This provides a surprising answer to an 
open question in algorithmic game theory and economics~\cite{BingLM04PresentationaSoSV}
\cite[Section 6.2.1]{liad-thesis} \cite[Section 2.2]{BN05}.

 \item Combinatorial Optimization:
Many optimization problems involving submodular functions, such as submodular function minimization,
are very well behaved and their optimal solutions have a rich structure.
In contrast, we show that, for several other submodular optimization problems which have been
considered recently in the literature, including submodular $s$-$t$ min cut and submodular vertex
cover, their optimal solutions are very unstructured, in the
sense that the optimal solutions do not have a succinct representation,
or even a succinct, approximate representation.

%
\end{itemize}

Although our new family of matroids proves that matroid rank functions (and more generally submodular functions)
are surprisingly unstructured, our concentration result for submodular functions shows that, in a different sense,
matroid rank functions (and other sufficiently ``smooth'' submodular functions) are surprisingly
structured.%

Submodularity has been an increasingly useful tool in machine learning in recent years.
For example, it has been used for feature selection problems in graphical models~\cite{KG05}
and various clustering problems~\cite{bilmes09}. In fact, submodularity 
has been the topic of several tutorials and workshops
at recent major conferences in machine learning~\cite{nips-2009-sub,icml-2008-sub,ijcai-09,nips-2011-sub}.
Nevertheless, our work is the first to use a  learning theory perspective to derive new structural
results for submodular functions and related structures (including matroids), thereby yielding
implications in many other areas.
Our work also potentially has useful applications --- our learning algorithms
can be employed in many areas where submodular functions arise (e.g., medical decision making and
economics). We discuss such applications in \Section{applications}.
Furthermore, our work defines a new learning model for approximate distributional learning that could
be useful for analyzing learnability of other interesting classes of real-valued functions. In fact,
this model has already been used to analyze the learnability of several classes of set functions
widely used in economics --- see Section~\ref{introlearn} and \Section{subsequent}.

\subsection{Our Results and Techniques}
\SectionName{overviewofresults}


The central topic of this paper is proving new structural results for submodular functions, motivated
by learnability considerations.
In the following we provide a more detailed description of our results.
For ease of exposition, we start by describing our new structural results,
then present our learning model and our learnability results within this model, and finally we
describe implications of our results in various areas.

\subsubsection{New Structural Results}

\paragraph{A new matroid construction}
The first result in this paper is the construction of a family of submodular
functions with interesting technical properties.
These functions are the key ingredient in our lower bounds for learning
submodular functions, inapproximability results for submodular optimization problems,
and the non-existence of succinct, approximate representations for gross substitutes functions.

Designing submodular functions directly is difficult because there is very little
tangible structure to work with.
It turns out to be more convenient to work with \newterm{matroids}\footnote{
    For the reader unfamiliar with matroids, a brief introduction to them is given in
    \Section{matroidbackgr}.
    For the present discussion, the only fact that we need about matroids is that
    the rank function of a matroid on $[n]$ is a submodular function on $2^{[n]}$.
},
because every matroid has an associated submodular function (its \newterm{rank function})
and because matroids are a very rich class of combinatorial objects
with numerous well-understood properties.

\newcommand{\rhigh}{\ensuremath{r_\mathrm{high}}\xspace}
\newcommand{\rlow}{\ensuremath{r_\mathrm{low}}\xspace}

\label{ourgoal}

Our goal is to find a collection of subsets of $[n]$
and two values \rhigh and \rlow such that,
for \emph{any} labeling of these subsets as either \textsc{High} or \textsc{Low},
we can construct a matroid for which each set labeled \textsc{High} has rank value \rhigh
and each set labeled \textsc{Low} has rank value \rlow.
We would like both the size of the collection and the
ratio $\rhigh / \rlow$ to be as large as possible.

Unfortunately existing matroid constructions can only achieve this goal
with very weak parameters; for further discussion of existing matroids, see \Section{related}.
Our new matroid construction, which involves numerous technical steps,
achieves this goal with the collection
of size super-polynomial in $n$ and the ratio
$\rhigh / \rlow = \tilde{\Omega}(n^{1/3})$.
This shows that matroid rank functions can be fiendishly unstructured
--- in our construction, knowing the value of the rank function on all-but-one of the sets in the
collection does not determine the rank value on the remaining set,
even to within a multiplicative factor $\tilde{\Omega}(n^{1/3})$.

More formally, let the collection of sets be $A_1,\ldots,A_k \subseteq [n]$
where each $\card{A_i} = \rhigh$.
For every set of indices $B \subseteq \set{1,\ldots,k}$ there is a
matroid $\mat_B$ whose associated rank function $r_B : 2^{[n]} \rightarrow \bR$ has the form
\begin{equation}
\EquationName{extremalmatroids}
r_B(S) ~=~ \max
\setst{ \card{I \cap S} }{
      \Big| I \intersect \bigcup_{j \in J} A_j \Big|
        \leq
        \rlow \cdot \card{J} - \sum_{j \in J} \card{A_j} +
        \Big| \bigcup_{j \in J} A_j \Big|
        ~~~\forall J \subseteq B ,\: \card{J} < \tau
}.
\end{equation}
We show that, if the sets $A_i$ satisfy a strong \newterm{expansion} property,
in the sense that any small collection of $A_i$ has small overlap,
and the parameters $\rhigh,\rlow,\tau$ are carefully chosen,
then this function satisfies $r_B(A_i) = \rlow$ whenever $i \in B$
and $r_B(A_i) = \rhigh$ whenever $i \not\in B$.

\paragraph{Concentration of submodular functions}
A major theme in probability theory is proving concentration bounds
for a function $f : 2^{[n]} \rightarrow \bR_{\geq 0}$ under product distributions\footnote{
    A random set $S \subseteq 2^{[n]}$ is said to have a \newterm{product distribution}
    if the events $i \in S$ and $j \in S$ are independent for every $i \neq j$.
}.
For example, when $f$ is linear, the Chernoff-Hoeffding bound is applicable.
For arbitrary $f$, the McDiarmid inequality is applicable.
The quality of these bounds also depends on the ``smoothness'' of $f$, which is quantified using the
Lipschitz constant $L := \max_{S,i} \abs{f(S \cup \set{i}) - f(S)}$.

We show that McDiarmid's tail bound can be strengthened under the
additional assumption that the function is monotone and submodular.
For a $1$-Lipschitz function (i.e., $L=1$),
McDiarmid's inequality gives concentration comparable to that of a
Gaussian random variable with standard deviation $\sqrt{n}$.
For example, the probability that the value of $f$ is $\sqrt{n}$ less than its expectation
is bounded above by a constant.
Such a bound is quite weak when the expectation of $f$ is significantly less than $\sqrt{n}$,
because it says that the probability of $f$ being negative is at most a constant,
even though that probability is actually zero.

Using Talagrand's inequality, we show that $1$-Lipschitz, monotone, submodular functions
are extremely tightly concentrated around their expected value.
The quality of concentration that we show is similar to
Chernoff-Hoeffding bounds --- importantly, it depends only on the expected value
of the function, and not on the dimension $n$.


\paragraph{Approximate characterization of matroids}
Our new matroid construction described above can be viewed at a high level as saying
that matroids can be surprisingly unstructured.
One can pick numerous large regions of the matroid (namely, the sets $A_i$)
and arbitrarily decide whether each region should have large rank or small rank.
Thus the matroid's structure is very unconstrained.

Our next result shows that, in a different sense, a matroid's structure is actually very constrained.
If one fixes any integer $k$ and looks at the rank values
amongst all sets of size $k$, then those values are extremely tightly concentrated
around their average --- almost all sets of size $k$ have nearly the same rank value.
Moreover, these averages are concave as a function of $k$.
That is, there exists a concave function $h : [0,n] \rightarrow \bR_{\geq 0}$
such that almost all sets $S$ have rank approximately $h(|S|)$.

This provides an interesting converse to the well-known fact that
the function $f : 2^{[n]} \rightarrow \bR$ defined by
$f(S)=h(\card{S})$ is a submodular function whenever $h : \bR \rightarrow \bR$ is concave.
Our proof uses our aforementioned result on concentration for submodular functions
under product distributions, and the multilinear extension \cite{CCPV}
of submodular functions, which has been of great value in recent work.

%

\subsubsection{Learning Submodular Functions}
\label{introlearn}
\paragraph{The learning model}
To study the learnability of submodular functions, we extend Valiant's classic PAC
model~\cite{Valiant:acm84}, which captures settings where the learning goal is to
predict the future based on past observations.
The abbreviation PAC stands for ``Probably Approximately Correct''.
The PAC model however is primarily designed for learning
\emph{Boolean-valued functions}, such as linear threshold functions, decision trees, and low-depth circuits~\cite{Valiant:acm84,KV:book94}.
For \emph{real-valued functions}, it is more meaningful to change the model by ignoring
small-magnitude errors in the predicted values.
Our results on learning submodular functions are presented in this new model,
which we call the \emph{PMAC model};
this abbreviation stands for ``Probably Mostly Approximately Correct''.

In this model, a learning algorithm is given a collection $\trainS=\set{S_1,S_2,\ldots}$
of polynomially many sets drawn i.i.d.~from
some fixed, but unknown, distribution $D$ over sets in $2^{[n]}$.
There is also a fixed but unknown function $\targetf : 2^{[n]} \rightarrow \bR_+$,
and the algorithm is given the value of $\targetf$ at each set in $\trainS$.
The goal is to design a polynomial-time algorithm that outputs a polynomial-time-evaluatable
function $f$ such that, with large probability over $\trainS$,
the set of sets for which $f$ is a good approximation for $\targetf$
has large measure with respect to $D$.
More formally,
$$
\operatorname{Pr}_{S_1,S_2,\ldots \sim D}\Big[~\:
    \probover{S \sim D}{
        f(S) \leq \targetf(S) \leq \alpha f(S)
    } \:\geq\: 1-\epsilon
~\:\Big]
\:~\geq~\: 1-\delta,
$$
where $f$ is the output of the learning algorithm when given inputs
$\set{\, (S_i,\targetf(S_i) ) \,}_{i=1,2,\ldots}$.
The approximation factor $\alpha \geq 1$ allows for multiplicative error in the function values.
Thus, whereas the PMAC model requires one to \emph{approximate} the value of a function on a set of
large measure and with high confidence,
the traditional PAC model requires one to predict the value \emph{exactly}
on a set of large measure and with high confidence.
The PAC model is the special case of our model with $\alpha=1$.

An alternative approach for dealing with real-valued functions in learning theory
 is to consider other loss functions such as the squared-loss or the L1-loss.
However, this approach does not distinguish between the case of having low error on most of the distribution and high error on just a few
points, versus moderately high error everywhere. In comparison, the PMAC model allows for more
fine-grained control with separate parameters for the amount and extent of errors, and in addition it
allows for consideration of multiplicative error which is often more natural in this context.
We discuss this further in \Section{related}.

Within the PMAC model we prove several algorithmic and hardness results for learning submodular
functions.  Specifically:

\paragraph{Algorithm for product distributions}
Our first learning result concerns product distributions.
This is a natural first step when studying
 learnability of various classes of functions, particularly when the class of functions has high complexity~\cite{KKMS,adam09,LMN,roccoproduct}.
By making use of our new concentration result
for monotone, submodular functions under product distributions,
we show that if the underlying distribution is a product distribution, then sufficiently ``smooth'' (formally, $1$-Lipschitz) submodular functions can be PMAC-learned with a
constant approximation factor $\alpha$  by a very simple algorithm.

\paragraph{Inapproximability for general distributions}
Although $1$-Lipschitz submodular functions can be PMAC-learned with a constant approximation factor
under product distributions, this result does not generalize to arbitrary distributions.
By making use of our new matroid construction, we show that every algorithm for PMAC-learning monotone, submodular functions
under arbitrary distributions must have approximation factor $\tOmega(n^{1/3})$ even for constant $\epsilon$ and $\delta$, and even if the
functions are matroid rank functions.
Moreover, this lower bound holds even if the algorithm knows the underlying distribution and
it can adaptively query the given function at points of its choice.

\paragraph{Algorithm for general distributions}
Our $\tOmega(n^{1/3})$ inapproximability result for general distributions
turns out to be close to optimal.
We give an algorithm to PMAC-learn an arbitrary non-negative, monotone, submodular function
with approximation factor $O(\sqrt{n})$  for any $\epsilon$ and $\delta$ by using a number of samples $\tilde{O}\left(n/\epsilon \log(1/\delta)\right)$.

This algorithm is based on a recent structural result which shows that any monotone, non-negative,
submodular function can be approximated within a factor of $\sqrt{n}$ on every point by the
square root of a linear function \cite{nick09}.
We leverage this result to reduce the problem of PMAC-learning a submodular function
to learning a linear separator in the usual PAC model.
We remark that an improved structural result for any subclass of submodular
functions would yield an improved analysis of our algorithm for that subclass.
Moreover, the  algorithmic approach we provide is quite robust and can be extended to handle more
general scenarios, including forms of noise.

\paragraph{The PMAC model}
Although this paper focuses only on learning submodular functions, the PMAC model that we introduce
is interesting in its own right, and can be used to study the learnability of other real-valued
functions.  Subsequent work by Badanidiyuru et al.\ \cite{BDFKNR} and Balcan et al.~\cite{BCIW} has
used this
model for studying the learnability of other classes of real-valued set functions that are widely
used in algorithmic game theory.  See \Section{related} for further discussion.

\subsubsection{Other Hardness Implications of Our Matroid Construction}

\paragraph{Algorithmic Game Theory and Economics}
An important consequence of our matroid construction is that matroid rank functions do not have a ``sketch'',
i.e., a concise, approximate representation.
Formally, there exist matroid rank functions on $2^{[n]}$ that do not have any
$\poly(n)$-space representation which approximates every value of the function to within a
$\tilde{o}(n^{1/3})$ factor.

In fact, as matroid rank functions are known to satisfy the gross substitute
property~\cite{kazuo-book}, our work implies that  gross substitutes do not have a concise,
approximate representation, or, in game theoretic terms, gross substitutes do not have a bidding
language. This provides a surprising answer to an 
open question in
economics~\cite{BingLM04PresentationaSoSV} \cite[Section 6.2.1]{liad-thesis}
\cite[Section 2.2]{BN05}.

\paragraph{Implications for submodular optimization}
Many optimization problems involving submodular functions, such as
linear optimization over a submodular base polytope,
submodular function minimization,
and submodular flow,
are very well behaved and their optimal solutions have a rich structure.
We consider several other submodular optimization problems which have been
considered recently in the literature, specifically
submodular function minimization under a cardinality constraint,
submodular $s$-$t$ min cut and submodular vertex cover.
These are difficult optimization problems, in the sense that the optimum value is
hard to compute. We show that they are also difficult in the sense that their optimal solutions
are very unstructured: the optimal solutions do not have a succinct
representation, or even a succinct, approximate representation.
%
%

Formally, the problem of submodular function minimization under a cardinality constraint is
$$
    \min \{\, f(A) \,:\, A \subseteq [n] ,\, |A| \geq d \,\}
$$
where $f$ is a monotone, submodular function.
We show that there there is no representation in $\poly(n)$ bits for the minimizers
of this problem, even allowing a factor $o(n^{1/3} / \log n)$ multiplicative error.
In contrast, a much simpler construction~\cite{bobby,SF,nick09} shows that no deterministic
algorithm performing $\poly(n)$ queries to $f$ can approximate the minimum value to within a
factor $o(n^{1/2} / \log n)$, but that construction implies nothing about
small-space representations of the minimizers.

For the submodular $s$-$t$ min cut problem, which is a generalization of the classic
$s$-$t$ min cut problem in network flow theory, we show that there is no
representation in $\poly(n)$ bits for the minimizers, even allowing
a factor $o(n^{1/3} / \log n)$ multiplicative error.
Similarly, for the submodular vertex cover problem, which is a generalization of the
classic vertex cover problem, we show that there is no representation in $\poly(n)$ bits for the
minimizers, even allowing a factor $4/3$ multiplicative error.

\subsection{Applications}
\SectionName{applications}

Algorithms for learning submodular functions could be
very useful in some of the applications where these functions arise.
For example, in the context of economics, our work provides useful tools for learning the valuation  functions of (typical) customers, with applications such as
bundle pricing, predicting demand, advertisement, etc.
%
%
Our algorithms are also useful in settings where one would like to predict
  the value of some function over objects described by features, where
  the features have positive but decreasing marginal impact on the
  function's value.  Examples include predicting the rate of growth of
  jobs in cities as a function of various amenities or enticements that the
  city offers, predicting the sales price of a house as a function of
  features (such as an updated kitchen, extra bedrooms,
  etc.) that it might have, and predicting the demand for a new
  laptop as a function of various add-ons that
  might be included. In all of these settings (and many
  others) it is natural to assume diminishing returns, making them
  well-suited to a formulation as a problem of learning a submodular function.

\subsection{Related Work}
\SectionName{related}

This section focuses primarily on prior work.
\Section{subsequent} discusses subsequent work that was directly motivated by this paper.

\paragraph{Submodular Optimization}
Optimization problems involving submodular functions have long played a central role
in combinatorial optimization. Recently there have been many applications of
these optimization problems in machine learning, algorithmic game theory and social networks.

The past decade has seen significant progress in algorithms for solving
submodular optimization problems.
There have been improvements in both the conceptual understanding
and the running time of algorithms for submodular function minimization \cite{IFF01,IO09,Lex00}.
There has also been much progress on approximation algorithms for various problems.
For example, there are now optimal approximation algorithms
for submodular maximization subject to a matroid constraint \cite{CCPV,Filmus,V08},
nearly-optimal algorithms for non-monotone submodular maximization \cite{FMV07,MNS,OGV},
and algorithms for submodular maximization subject to a wide variety of
constraints~\cite{CVZ10,CVZ11,FNS,KST,LMNS,LSV,OGV,V09}.

Approximation algorithms for submodular analogues of several other optimization problems
have been studied, including
load balancing~\cite{SF},
set cover~\cite{IN09,W82},
shortest path~\cite{GKTW09},
sparsest cut~\cite{SF},
$s$-$t$ min cut~\cite{JB},
vertex cover~\cite{GKTW09,IN09},
etc.
In this paper we provide several new results on the difficulty of such problems.
Most of these previous papers on submodular optimization prove inapproximability results
using matroids whose rank function has the same form as \Equation{extremalmatroids},
but only for the drastically simpler case of $k=1$.  Our construction is much more intricate since
we must handle the case $k = n^{\omega(1)}$.

Recent work of Dobzinski and Vondr\'ak \cite{DV} proves inapproximability
of welfare maximization in combinatorial auctions with submodular valuations.
Their proof is based on a collection of submodular functions that take high values on every set in a
certain exponential-sized family, and low values on sets that are far from that family.
Their proof is in the same spirit as our inapproximability result, although their construction
is technically very different than ours.
In particular, our result uses
a special family of submodular functions and family of sets for which the sets are \emph{local
minima} of the functions, whereas their result uses
a different family of submodular functions and family of sets for which the sets are \emph{local
maxima} of the functions.


\paragraph{Learning real-valued functions and the PMAC model}
In the machine learning literature~\cite{tishbirani,Vapnik:book98},
learning real-valued functions in a distributional  setting is often addressed by considering
 loss functions such as the $L_2$-loss or the $L_1$-loss, where the loss incurred by predicting according to hypothesis $f$ on a given example $x$
 is $l_f(x, \targetf)=(f(x)-{\targetf}(x))^2$ for $L_2$-loss and $l_f(x, \targetf)=|f(x)-{\targetf}(x)|$ for $L_1$-loss.
In this context,
one typically normalizes the function to be in $[0,1]$, and the aim is
to achieve low expected loss $\expectover{x}{ l_f(x, \targetf)}$.
However, lower bounds on expected loss do not distinguish between the case of achieving
low loss on most of the distribution and high loss on just a few
points, versus moderately high loss everywhere.

For example, consider a function $f$ with codomain $\set{0,1,\ldots,n}$.
Here we would normalize by a factor $1/n$, so a lower bound of
$\Omega(n^{1/3})$ on expected $L_1$-loss before normalizing
is equivalent to a lower bound of $\Omega(n^{-2/3})$ after normalizing.
But such a lower bound would not distinguish between the following two scenarios:
(1) 
one where any hypothesis $f$ produced by the algorithm has
$L_1$-loss of $\Omega(n^{-2/3})$ on a $1/2$ fraction of the points, and
(2) one where an algorithm can output a hypothesis $f$ that is exactly correct on a $1 \!-\!
O(n^{-2/3})$ fraction of the points, but has high loss on the rest.

In comparison, the PMAC model provides more fine-grained control,
with separate parameters for the amount $\epsilon$ and extent $\alpha$ of errors.
For instance, in scenario (1) if the normalized function has $f(x)=1/n$ on the points of
high $L_1$-loss,
then this would correspond to a lower bound of $\alpha=\Omega(n^{1/3})$ and $\epsilon=1/2$
in the PMAC model.
In contrast, scenario (2) would correspond to having an upper bound of $\alpha=1$ and $\epsilon=n^{-2/3}$.

Another advantage of the PMAC model is that, since it uses multiplicative error,
an algorithm in the PMAC model provides good approximations {\em uniformly at all
scales}.\footnote{For example, if $f$ is a good hypothesis in the PMAC model,
and one focuses on points $x$ such that $\targetf(x) \leq c$ and rescales, then 
the multiplicative approximation guarantee provided by $f$ remains true in the restricted
and rescaled domain, as long as this set has sufficiently large probability
mass.}
Existing work in the the learning theory literature has also considered
guarantees that combine both multiplicative and additive aspects,
in the context of sample complexity bounds~\cite{haussler,phil-relative}.
However, this work did not consider the development of efficient algorithms or learnability of
submodular functions.

We remark that our construction showing the $\tilde{\Omega}(n^{1/3})$
inapproximability in the PMAC model immediately implies a lower bound
of $\tilde{\Omega}(n^{-2/3})$ for the $L_1$-loss and
$\tilde{\Omega}(n^{-4/3})$ for the $L_2$-loss (after normalization).



\paragraph{Learning submodular functions}
To our knowledge, there is no prior work on learning submodular functions in a distributional,
PAC-style learning setting.
The most relevant work is a paper of Goemans et al.~\cite{nick09},
which considers the problem of ``approximating submodular functions everywhere''.
That paper considers the algorithmic problem of efficiently finding a function which approximates a
submodular function at \emph{every} set in its domain.
They give an algorithm which achieves an approximation factor $\tO(\sqrt{n})$,
and they also show $\tOmega(\sqrt{n})$ inapproximability.
Their algorithm adaptively queries the given function {\em on sets of its choice},
and their output function must approximate the given function on \emph{every} set.\footnote
{Technically speaking, their model can be viewed as ``approximate learning everywhere with value queries'',
which is less natural in certain machine learning  scenarios.
    In particular, in many
    applications arbitrary membership or value queries are undesirable because
    natural oracles, such as hired humans, have difficulty labeling synthetic examples~\cite{baum}.
    Also, negative results for approximate learning everywhere do not necessarily imply hardness for learning in more
    widely used learning models. We discuss this in more detail below.}
In contrast, our PMAC model falls into the more widely studied passive, supervised learning
setting~\cite{AB99,KV:book94,Valiant:acm84,Vapnik:book98},
which is more relevant for our motivating applications discussed in \Section{applications}.

Our algorithm for PMAC-learning under general distributions and the Goemans et al.~algorithm both
rely on the structural result (due to Goemans et al.) that monotone, submodular functions can be
approximated by the square root of a linear function to within a factor $\sqrt{n}$.
In both cases, the challenge is to find this linear function.
The Goemans et al.~algorithm is very sophisticated:
it gives an intricate combinatorial algorithm to approximately solve a certain convex program
which produces the desired function.
Their algorithm requires query access to the function and so it is not applicable in the PMAC model.
Our algorithm, on the other hand, is very simple:
given the structural result, we can reduce our problem to that of learning a linear separator,
which is easily solved by linear programming.
Moreover, our algorithm is noise-tolerant and more amenable to extensions; 
we elaborate on this in \Section{general-upper}. 

On the other hand, our lower bound is significantly more involved than the lower bound of Goemans et
al.~\cite{nick09} and the related lower bounds of Svitkina and~Fleischer~\cite{SF}.
Essentially, the previous results show only show \emph{worst-case} inapproximability,
whereas we need to show \emph{average-case} inapproximability.
A similar situation occurs with Boolean functions,
where lower bounds for distributional learning are typically much harder to show
than lower bounds for exact learning (i.e., learning everywhere).
For instance, even conjunctions are hard to learn in the exact learning model (from random examples or via membership queries),
and yet they are trivial to PAC-learn.
Proving a lower bound for PAC-learning requires exhibiting
some fundamental complexity in the class of functions.
It is precisely this phenomenon which makes our lower bound challenging to prove.

\paragraph{Learning valuation functions and other economic solutions concepts}
As discussed in \Section{applications}, one important application of our results on learning is
for learning valuation functions.  G. Kalai~\cite{kalai03} considered the problem of learning
\newterm{rational choice functions} from random examples.  Here, the learning algorithm
observes sets $S \subseteq [n]$ drawn from some distribution $D$,
along with a choice $c(S) \in [n]$ for each $S$.  The goal is then to
learn a good approximation to $c$ under various natural assumptions
on $c$.  For the assumptions considered in \cite{kalai03}, the choice function $c$
has a simple description as a linear ordering.  In contrast, in our
work we consider valuation functions that may be much more complex and for which the PAC model
would not be sufficient to capture the inherent easiness or difficulty of the problem.
Kalai briefly considers utility functions over bundles and
remarks that ``the PAC-learnability of preference relations and choice
functions on commodity bundles ... deserves further study''~\cite{kalai01}.

\subsection{Structure of the paper}
We begin with background about matroids and submodular functions in \Section{backgr}.
In Section~\ref{section-structure} we present our new structural results: a new extremal family of
matroids and new concentration results for submodular functions.  We present our new framework for
learning real-valued functions as well as our results for learning submodular functions within this
framework in~\Section{learning}. We further present implications of our matroid construction in
optimization and algorithmic game theory in \Section{implications} and Section~\ref{sec:econ}.

\section{Preliminaries: Submodular Functions and Matroids}
\SectionName{backgr}

\subsection{Notation}
Let $[n]$ denote the set $\set{1,2,\ldots,n}$.
This will typically be used as the ground set for the matroids and submodular functions
that we discuss.
For any set $S \subseteq [n]$ and element $x \in [n]$,
we let $S+x$ denote $S \union \set{x}$.
The indicator vector of a set $S \subseteq [n]$ is $\ind(S) \in \set{0,1}^{n}$,
where $\ind(S)_i$ is $1$ if $i$ is in $S$ and $0$ otherwise.
We frequently use this natural isomorphism between $\set{0,1}^n$ and $2^{[n]}$.

\subsection{Submodular Functions and Matroids}
\SectionName{matroidbackgr}

In this section we give a brief introduction to matroids and submodular functions
and discuss some standard facts that will be used throughout the paper.
A more detailed discussion can be found in standard references \cite{Frank,Fuji,L83,Oxley,Schrijver}.
The reader familiar with matroids and submodular functions
  may wish to skip to Section~\ref{section-structure}.

Let $V=\set{v_1,\ldots,v_n}$ be a collection of vectors in some vector space $\bF^m$.
Roughly one century ago, several researchers observed
that the linearly independent subsets of $V$ satisfy some interesting combinatorial properties.
For example, if $B \subseteq V$ is a basis of $\bF^m$ and $I \subseteq V$ is linearly independent
but not a basis, then there is always a vector $v \in B$ which is not in the span of $I$,
implying that $I + v$ is also linearly independent.

These combinatorial properties are quite interesting to study in their own right,
as there are a wide variety of objects which satisfy these properties but
(at least superficially) have no connection to vector spaces.
A \newterm{matroid} is defined to be any collection of elements that satisfies these same
combinatorial properties, without referring to any underlying vector space.
Formally, a pair $\mat = (\ground, \Ind)$ is called a matroid if
$\Ind \subseteq 2^\ground$ is a non-empty family such that
\begin{itemize}
\item if $J \subseteq I$ and $I \in \Ind$, then $J \in \Ind$, and
\item if $I, J \in \Ind$ and $\card{J} < \card{I}$, then there exists an $i \in I \setminus J$
such that $J+i \in \Ind$.
\end{itemize}
The sets in $\Ind$ are called \newterm{independent}.

Let us illustrate this definition with two examples.
\begin{description}
\item[Partition matroid]
Let $V_1 \union \cdots \union V_k$ be a partition of $\ground$,
i.e., $\bigcup_i V_i = [n]$ and $V_i \intersect V_j = \emptyset$ whenever $i \neq j$.
Define $\Ind \subseteq 2^{\ground}$ be the family of partial transversals of $\ground$,
i.e., $I \in \Ind$ if and only if $\card{I \intersect V_i} \leq 1$ for all $i \in [k]$.
It is easy to verify that the pair $(\ground,\Ind)$ satisfies the definition of a matroid.
This is called a \newterm{partition matroid}.

This definition can be generalized slightly.
Let $I \in \Ind$ if and only if $\card{I \intersect V_i} \leq b_i$ for all $i=1,\ldots,k$,
where the $b_i$ values are arbitrary.
The resulting pair $(\ground,\Ind)$ is a (generalized) partition matroid.

\item[Graphic matroid]
Let $G$ be a graph with edge set $E$.
Define $\Ind \subseteq 2^{E}$ to be the collection of all acyclic sets of edges.
One can verify that the pair $(\ground,\Ind)$ satisfies the definition of a matroid.
This is called a \newterm{graphic matroid}.
\end{description}

One might wonder: given an arbitrary matroid $(\ground,\Ind)$,
do there necessarily exist vectors $V = \set{v_1,\ldots,v_n}$ in some vector space
for which the independent subsets of $V$ correspond to $\Ind$?
Although this is true for partition matroids and graphic matroids,
in general the answer is no.
So matroids do not capture all properties of vector spaces.
Nevertheless, many concepts from vector spaces do generalize to matroids.

For example, given vectors $V \subset \bF^m$, all maximal linearly independent subsets of $V$
have the same cardinality, which is the dimension of the span of $V$.
Similarly, given a matroid $(\ground,\Ind)$, all maximal sets in $\Ind$
have the same cardinality, which is called the \newterm{rank} of the matroid.

\newcommand{\Span}{\operatorname{span}}

More generally, for any subset $V' \subseteq V$, we can define its rank
to be the dimension of the span of $V'$;
equivalently, this is the maximum size of any linearly independent subset of $V'$.
This notion generalizes easily to matroids.
The \newterm{rank function} of the matroid $(\ground,\Ind)$ is the function
$\rankf_\mat : 2^\ground \rightarrow \bN$ defined by
\[
    \rankf_\mat(S) \defeq \max \setst{ |I| }{ I\subseteq S ,\, I\in \Ind }.
\]

Rank functions also turn out to have numerous interesting properties,
the most interesting of which is the \newterm{submodularity} property.
Let us now illustrate this via an example.
Let $V'' \subset V' \subset V$ be collections of vectors in some vector space.
Suppose that $v \in V$ is a vector which does not lie in $\Span(V')$.
Then it is clear that $v$ does not lie in $\Span(V'')$ either.
Consequently,
$$
\rankf(V'+v)-\rankf(V')=1 \qquad\implies\qquad \rankf(V''+v)-\rankf(V'')=1.
$$
The submodularity property is closely related: it states that
$$
\rankf_\mat(T+i)-\rankf_\mat(T) ~~\leq~~ \rankf_\mat(S+i)-\rankf_\mat(S)
    \qquad\forall S \subseteq T \subseteq \ground ,\, i \in \ground.
$$

The following properties of real-valued set functions play an important role in this paper.
A function $f : 2^{[n]} \rightarrow \bR$ is
\begin{itemize}
\item \newterm{Normalized} if $f(\emptyset) = 0$.
\item \newterm{Non-negative} if $f(S) \geq 0$ for all $S$.
\item \newterm{Monotone} (or \newterm{non-decreasing}) if $f(S) \leq f(T)$ for all $S \subseteq T$.
\item \newterm{Submodular} if it satisfies
\begin{equation}
\EquationName{submoddef1}
f(T+i)-f(T) ~~\leq~~ f(S+i)-f(S)
    \qquad\forall S \subseteq T \subseteq \ground ,\, i \in \ground.
\end{equation}
An equivalent definition is
\begin{equation}
\EquationName{submoddef2}
f(A)+f(B) ~~\geq~~ f(A \union B)+f(A \intersect B)
    \qquad\forall A \subseteq B \subseteq \ground.
\end{equation}
\item \newterm{$L$-Lipschitz} if $\abs{f(S+i)-f(S)} \leq L$ for all $S \subseteq [n]$ and $i \in [n]$.
\end{itemize}
Matroid rank functions are integer-valued,
normalized, non-negative, monotone, submodular and $1$-Lipschitz.
The converse is also true: any function satisfying those properties is a matroid rank function.

The most interesting of these properties is submodularity.
It turns out that there are a wide variety of set functions which satisfy
the submodularity property but do not come from matroids.
Let us mention two examples.
\begin{description}
\item[Coverage function]
Let $S_1,\ldots,S_n$ be a subsets of a ground set $[m]$.
Define the function $f : 2^{\ground} \rightarrow \bN$ by
$$
f(I) ~=~ \Big\lvert \bigcup_{i \in I} S_i \Big\rvert.
$$
This is called a \newterm{coverage function}.
It is integer-valued, normalized, non-negative, monotone and submodular,
but it is not $1$-Lipschitz.

\item[Cut function]
Let $G=(\ground,E)$ be a graph.
Define the function $f : 2^{\ground} \rightarrow \bN$ by
$$
f(U) ~=~ \card{ \delta(U) }
$$
where $\delta(U)$ is the set of all edges that have exactly one endpoint in $U$.
This is called a \newterm{cut function}.
It is integer-valued, normalized, non-negative and submodular,
but it is not monotone or $1$-Lipschitz.
\end{description}

%
%
%

\section{New Structural Results About Matroids and Submodular Functions}
\label{section-structure}

\subsection{A New Family of Extremal Matroids}
\label{section-new-extremal-matroid}

In this section we
present a new family of matroids whose rank functions take wildly varying
values on many sets.
The formal statement of this result is as follows.

\begin{figure}[t]
\begin{center}
\ifpdf
    \includegraphics[scale=.8]{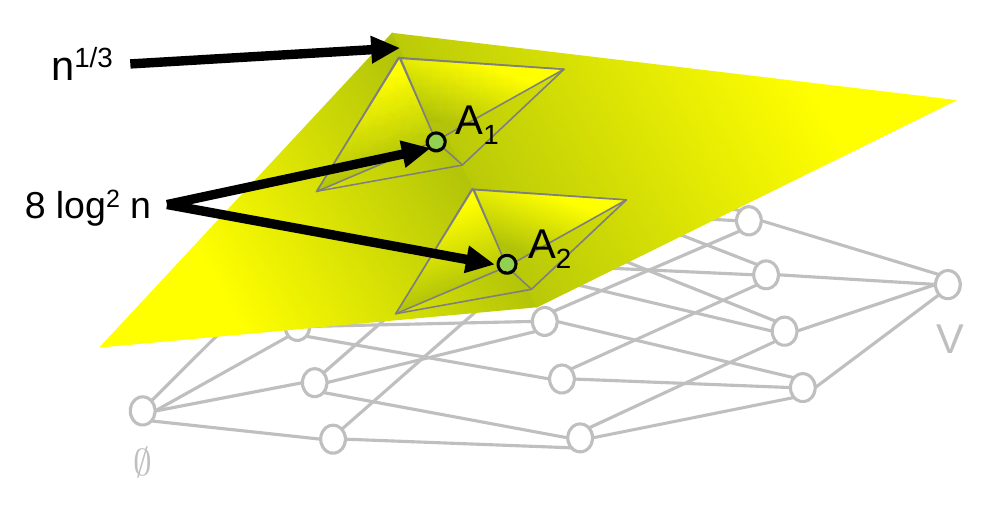}
\else
    \includegraphics[scale=.8]{Fig1.eps}
\fi
\caption{This figure aims to illustrate a function $\rankf_{\mat_\cB}$ that is
constructed by \Theorem{manymatroids}.
This is a real-valued function whose domain is the lattice of subsets of $V$.
The family $\cB$ contains the sets $A_1$ and $A_2$, both of which have size $n^{1/3}$.
Whereas $\rankf_{\mat_\cB}(S)$ is large (close to $n^{1/3}$) for most sets $S$
of size $n^{1/3}$, we have $\rankf_{\mat_\cB}(A_1) = \rankf_{\mat_\cB}(A_2) = 8 \log^2 n$.
In order to ensure submodularity, sets near $A_1$ or $A_2$ also have low values.
}
\FigureName{fig1}
\end{center}
\end{figure}

\newcommand{\manymatroids}
{
For any $k \geq 8$ with $k = 2^{o(n^{1/3})}$,
there exists a family of sets $\cA \subseteq 2^{[n]}$
and a family of matroids $\cM = \setst{ \mat_\cB }{ \cB \subseteq \cA }$ with the following properties.
\begin{itemize}
\item $\card{\cA} = k$ and $\card{A} = n^{1/3}$ for every $A \in \cA$.
\item For every $\cB \subseteq \cA$ and every $A \in \cA$, we have
\[
\rankf_{\mat_\cB}(A) ~=~
    \begin{cases}
    8 \log k   &\qquad\text{(if $A \in \cB$)} \\
    \card{A}    &\qquad\text{(if $A \in \cA \setminus \cB$)}.
    \end{cases}
\]
\end{itemize}
}

\begin{theorem}
\TheoremName{manymatroids}
\manymatroids
\end{theorem}

\Theorem{manymatroids} implies that there exists a super-polynomial-sized collection of subsets
of $[n]$ such that, for \emph{any} labeling of those
sets as \textsc{High} or \textsc{Low}, we can construct a matroid where the sets
in \textsc{High} have rank $r_\mathrm{high}$ and the sets in \textsc{Low} have rank
$r_\mathrm{low}$, and the ratio
$r_\mathrm{high}/r_\mathrm{low} = \tilde{\Omega}(n^{1/3})$.
For example, by picking $k = n^{\log n}$, in the matroid $\mat_\cB$,
a set $A$ has rank only $O(\log^2 n)$ if $A \in \cB$,
but has rank $n^{1/3}$ if $A \in \cA \setminus \cB$.
In other words, as $\cB$ varies, the rank of a set $A \in \cA$ varies wildly,
depending on whether $A \in \cB$ or not.

Later sections of the paper use \Theorem{manymatroids} to prove various negative results.
In \Section{general-lower} we use the theorem to prove our inapproximability result for PMAC-learning
submodular functions under arbitrary distributions.
In \Section{implications} we use the theorem to prove results on the difficulty of several submodular
optimization problems.

In the remainder of Section \ref{section-new-extremal-matroid} we discuss~\Theorem{manymatroids} and
give a detailed proof.

\subsubsection{Discussion of \Theorem{manymatroids} and Sketch of the Construction}

We begin by discussing some set systems which give intuition on how \Theorem{manymatroids}
is proven. Let $\cA = \set{ A_1,\ldots,A_k }$ be a collection of subsets of $[n]$
and consider the set system
$$
    \cI ~=~ \setst{ I }{ \card{I} \leq r ~~\wedge~~ \card{I \intersect A_j} \leq b_j ~\:\forall
        j \in [k]}.
$$
If $\cI$ is the family of independent sets of a matroid $\mat$,
and if $\rankf_\mat(A_j) = b_j$ for each $j$, then perhaps such a construction
can be used to prove \Theorem{manymatroids}.

Even in the case $k=2$, understanding $\cI$ is quite interesting.
First of all, $\cI$ typically is not a matroid.
Consider taking $n=5$, $r=4$, $A_1 = \set{1,2,3}$, $A_2 = \set{3,4,5}$ and $b_1=b_2=2$.
Then both $\set{1,2,4,5}$ and $\set{2,3,4}$ are maximal sets in $\cI$
but their cardinalities are unequal, which violates a basic matroid property.
However, one can verify that $\cI$ is a matroid
if we additionally require that $r \leq b_1+b_2-\card{A_1 \intersect A_2}$.
In fact, we could place a constraint on $\card{I \intersect (A_1 \union A_2)}$
rather than on $\card{I}$, obtaining
$$
    \setst{ I }{ \card{I \intersect A_1} \leq b_1
        ~~\wedge~~ \card{I \intersect A_2} \leq b_2
        ~~\wedge~~ \card{I \intersect (A_1 \union A_2)} \leq b_1 + b_2 - \card{A_1 \intersect A_2} },
$$
which is the family of independent sets of a matroid.
In the case that $A_1$ and $A_2$ are disjoint, the third constraint becomes
$\card{I \intersect (A_1 \union A_2)} \leq b_1 + b_2$, which is redundant because it is
implied by the first two constraints.
In the case that $A_1$ and $A_2$ are ``nearly disjoint'', this third constraint
becomes necessary and it incorporates an ``error term'' of $-\card{A_1 \intersect A_2}$.

To generalize to $k > 2$, we impose similar constraints for every subcollection of $\cA$,
and we must include additional ``error terms'' that are small when the $A_j$'s are nearly disjoint.
\Theorem{mainthm} proves that
\begin{equation}
    \EquationName{Idef}
    \cI ~=~ \setst{ I }{ \card{ I \intersect A(J) } \:\leq\: g(J) ~\:\forall J \subseteq [k] }.
\end{equation}
is a matroid, where the function $g : 2^{[k]} \rightarrow \bZ$ is defined by
\begin{equation}
\EquationName{fdef}
g(J) \defeq
\sum_{j \in J} b_j \:-\:
        \Big( \sum_{j \in J} \card{A_j} - \card{A(J)} \Big),
~~\qquad\text{where}\qquad~~
A(J) \defeq \Union_{j \in J} A_j.
\end{equation}
In the definition of $g(J)$,
we should think of $-\big( \sum_{j \in J} \card{A_j} - \card{A(J)}\big)$ as an ``error term'',
since it is non-positive, and it captures the ``overlap'' of the sets $\setst{ A_j }{ j \in J }$.
In particular, in the case $J=\set{1,2}$, this error term is $-\card{A_1 \intersect A_2}$,
as it was in our discussion of the case $k=2$.

Let us now consider a special case of this construction.
If the $A_j$'s are all disjoint then the error terms are all $0$,
so the family $\cI$ reduces to
$$
    \setst{ I }{  \card{ I \intersect A_j } \:\leq\: b_j ~\:\forall j \in [k] },
$$
which is a (generalized) partition matroid, regardless of the $b_j$ values.
Unfortunately these matroids cannot achieve our goal of having superpolynomially many sets
labeled \textsc{High} or \textsc{Low}.
The reason is that, since the $A_j$'s must be disjoint, there can be at most $n$ of them.

In fact, it turns out that any matroid of the form \eqref{eq:Idef}
can have at most $n$ sets in the collection $\cA$.
To obtain a super-polynomially large $\cA$ we must modify this construction slightly.
\Theorem{modified} shows that, under certain conditions, the family
\[
\bar{\cI}
    ~=~ \Big\{~ I ~:~
        \card{I} \leq d
        ~\And~
        \card{I \intersect A(J)} \leq g(J)
            ~\: \forall J \subseteq [k],\, \card{J} < \tau
        ~\Big\}
\]
is also the family of independent sets of a matroid.
Introducing the crucial parameter $\tau$ allows us to have
obtain a super-polynomially large $\cA$.

There is an important special case of this construction.
Suppose that $\card{A_j} = d$ and $b_j = d-1$ for every $j$,
and that $\card{ A_i \intersect A_j } \leq 2$ for all $i \neq j$.
The resulting matroid is called a \newterm{paving matroid}, a well-known type of matroid.
These matroids are quite relevant to our goals of having super-polynomially
many sets labeled \textsc{High} and \textsc{Low}.
The reason is that the conditions on the $A_j$'s are equivalent to $\cA$ being a constant-weight
error-correcting code of distance $4$, and it is well-known that such codes can have
super-polynomial size.
Unfortunately this construction has $r_\mathrm{low} = d-1$ and $r_\mathrm{high}=d$;
this small, additive gap is much too weak for our purposes.

The high-level plan underlying \Theorem{manymatroids} is to find a new class of matroids
that somehow combines the positive attributes of both partition and paving matroids.
From paving matroids we will inherit the large size of the collection $\cA$,
and from partition matroids we will inherit a large ratio $r_\mathrm{high} / r_\mathrm{low}$.

One of our key observations
is that there is a commonality between partition and paving matroids:
the collection $\cA$ must satisfy an ``expansion'' property,
which roughly means that the $A_j$'s cannot overlap too much.
With partition matroids the $A_j$'s must be disjoint, which amounts to having ``perfect'' expansion.
With paving matroids the $A_j$'s must have small pairwise intersections,
which is a fairly weak sort of expansion.

It turns out that the ``perfect'' expansion required by partition matroids is too strong
for $\cA$ to have super-polynomial size, and the ``pairwise'' expansion required by paving
matroids is too weak to allow a large ratio $r_\mathrm{high} / r_\mathrm{low}$.
Fortunately, weakening the expansion from ``perfect'' to ``nearly-perfect''
is enough to obtain a collection $\cA$ of super-polynomial size.
With several additional technical ideas, we show that these
nearly-perfect expansion properties can be leveraged to achieve
our desired ratio $r_\mathrm{high} / r_\mathrm{low} = \tilde{\Omega}(n^{1/3})$.
These ideas lead to a proof of \Theorem{manymatroids}.

\subsubsection{Our New Matroid Constructions}
\SectionName{bumpy}


Our first matroid construction is given by the following theorem,
which is proven in \Section{matroidconstructions}.

\begin{theorem}
\TheoremName{mainthm}
The family $\cI$ given in \Equation{Idef}
is the family of independent sets of a matroid, if it is non-empty.
\end{theorem}

As mentioned above, \Theorem{mainthm} does not suffice to prove \Theorem{manymatroids}.
To see why, suppose that $\card{\cA} = k > n$ and that $b_i < \card{A_i}$ for every $i$.
Then $g([k]) \leq n-k < 0$, and therefore $\cI$ is empty.
So the construction of \Theorem{mainthm} is only applicable when $k \leq n$,
which is insufficient for proving \Theorem{manymatroids}.

We now modify the preceding construction by introducing
a sort of ``truncation'' operation which allows us to take $k \gg n$.
We emphasize that this truncation is \emph{not} ordinary matroid truncation.
The ordinary truncation operation \emph{decreases} the rank of the matroid,
whereas we want to \emph{increase} the rank by throwing away constraints in the definition of $\cI$.
We will introduce an additional parameter $\tau$,
and only keep constraints for $\card{J} < \tau$.
So long as $g$ is large enough for a certain interval,
then we can truncate $g$ and still get a matroid.

\begin{definition}
Let $d$ and $\tau$ be non-negative integers.
A function $g : 2^{[k]} \rightarrow \bR$ is called $(d,\tau)$-large if
\begin{equation}
\EquationName{large}
g(J) ~\geq~ \begin{cases}
0  &\quad\forall J \subseteq [k],~ \card{J} < \tau \\
d  &\quad\forall J \subseteq [k],~ \tau \leq \card{J} \leq 2\tau-2.
\end{cases}
\end{equation}
The truncated function $\bar{g} : 2^{[k]} \rightarrow \bZ$ is defined by
\[
\bar{g}(J) \defeq
    \begin{cases}
    g(J) &\quad\text{\rm (if $\card{J} < \tau$)} \\
    d &\quad\text{\rm (otherwise).}
    \end{cases}
\]
\end{definition}


\newcommand{\thmmodified}{
    Suppose that the function $g$ defined in \Equation{fdef} is $(d,\tau)$-large.
    Then the family
    \[
    \bar{\cI}
        ~=~ \setst{ I }{ \card{I \intersect A(J)} \leq \bar{g}(J) ~~\forall J \subseteq [k] }
    \]
    is the family of independent sets of a matroid.
}
\begin{theorem}
\TheoremName{modified} \thmmodified
\end{theorem}

Consequently, we claim that the family
\[
\bar{\cI}
    ~=~ \Big\{~ I ~:~
        \card{I} \leq d
        ~\And~
        \card{I \intersect A(J)} \leq g(J)
            ~\: \forall J \subseteq [k],\, \card{J} < \tau
        ~\Big\}
\]
is also the family of independent sets of a matroid.
This claim follows immediately if the $A_i$'s cover the ground set (i.e., $A([k])=[n]$),
because the matroid definition in \Theorem{modified} includes the constraint
$\card{I} = \card{I \intersect A([k])} \leq \bar{g}([k]) = d$.
Alternatively, if $A([k]) \neq [n]$, we may we apply the well-known matroid truncation operation
which constructs a new matroid simply by removing all independent sets of size greater than $d$.

This construction yields quite a broad family of matroids.
We list several interesting special cases in \Appendix{specialcases}.
In particular, partition matroids and paving matroids are both special cases.
Thus, our construction can produce ``non-linear'' matroids
(i.e., matroids that do not correspond to vectors over any field),
as the V\'amos matroid is a paving matroid that is non-linear \cite{Oxley}.

\subsubsection{Proofs of \Theorem{mainthm} and \Theorem{modified}}
\SectionName{matroidconstructions}

In this section, we will prove \Theorem{mainthm} and \Theorem{modified}.
We start with a simple but useful lemma which describes
a general set of conditions that suffice to obtain a matroid.

Let $\cC \subseteq 2^{[n]}$ be an arbitrary family of sets
and let $g : \cC \rightarrow \bZ$ be a function.
Consider the family
\begin{equation}
\EquationName{independentsets}
\cI ~=~ \setst{ I }{ \card{I \intersect C} \leq g(C) ~~\forall C \in \cC }.
\end{equation}
For any $I \in \cI$, define $T(I) = \setst{ C \in \cC }{ \card{I \intersect C} = g(C) }$
to be the set of constraints that are ``tight'' for the set $I$.
Suppose that $g$ has the following property:
\begin{equation}
\EquationName{uncross2}
\forall I \in \cI, \quad C_1, C_2 \in T(I)
    \quad\implies\quad
    (C_1 \union C_2 \in T(I)) ~\vee~ (C_1 \intersect C_2 = \emptyset).
\end{equation}
Properties of this sort are commonly called ``uncrossing'' properties.
Note that we do not require that $C_1 \intersect C_2 \in \cC$.
We show in the following lemma that this uncrossing property is sufficient\footnote{
    There are general matroid constructions in the literature
    which are similar in spirit to \Lemma{generalindep},
    e.g., the construction of Edmonds~\cite[Theorem 15]{EdmondsSubmodular}
    and the construction of Frank and Tardos~\cite[Corollary 49.7a]{Schrijver}.
    However, we were unable to use those existing constructions to prove
    \Theorem{mainthm} or \Theorem{modified}.
}
to obtain a matroid.

\begin{lemma}
\LemmaName{generalindep}
Assume that \Equation{uncross2} holds.
Then $\cI$ is the family of independent sets of a matroid, if it is non-empty.
\end{lemma}

\begin{proof}
We will show that $\cI$ satisfies the required axioms of an independent set family.
If $I \subseteq I' \in \cI$ then clearly $I \in \cI$ also.
So suppose that $I \in \cI$, $I' \in \cI$ and $\card{I} < \card{I'}$.
Let $C_1, \ldots, C_m$ be the maximal sets in $T(I)$ and let $C^* = \union_i \, C_i$.
Note that these maximal sets are disjoint, otherwise we could replace
any intersecting sets with their union.
In other words,
$C_i \intersect C_j = \emptyset$ for $i \!\neq\! j$, otherwise \Equation{uncross2} implies that
$C_i \union C_j \in T(I)$, contradicting maximality.
So
$$
\card{I' \intersect C^*}
    ~=~ \sum_{i=1}^m \card{I' \intersect C_i}
    ~\leq~ \sum_{i=1}^m g(C_i)
    ~=~ \sum_{i=1}^m \card{I \intersect C_i}
    ~=~ \card{I \intersect C^*}.
$$
Since $\card{I'} > \card{I}$ but
$\card{I' \intersect C^*} \leq \card{I \intersect C^*}$,
we must have that $\card{I' \setminus C^*} > \card{I \setminus C^*}$.
The key consequence is that some element $x \in I' \setminus I$ is not contained in any tight set,
i.e., there exists $x \in I' \setminus \big( C^* \union I \big)$.
Then $I+x \in \cI$ because for every $C \in \cC$ with $x \in C$
we have $\card{I \intersect C} \leq g(C)-1$.
\end{proof}


We now use \Lemma{generalindep} to prove \Theorem{mainthm}, restated here.

\repeatclaim{\Theorem{mainthm}}{
    The family $\cI$ defined in \Equation{Idef}, namely
    $$
        \cI ~=~ \setst{ I }{  \card{ I \intersect A(J) } \:\leq\: g(J) ~\:\forall J \subseteq [k] },
    $$
    where
    $$
    g(J) \defeq
    \sum_{j \in J} b_j \:-\:
            \Big( \sum_{j \in J} \card{A_j} - \card{A(J)} \Big)
    ~~\qquad\text{and}\qquad~~
    A(J) \defeq \Union_{j \in J} A_j,
    $$
    is the family of independent sets of a matroid, if it is non-empty.
}

This theorem is proven by showing that the constraints defining $\cI$
can be ``uncrossed'' (in the sense that they satisfy \eqref{eq:uncross2}),
then applying \Lemma{generalindep}.
It is not a~priori obvious that these constraints can be uncrossed:
in typical uses of uncrossing, the right-hand side $g(J)$ should be a submodular function of $J$
and the left-hand side $\card{I \intersect A(J)}$ should be a \textit{super}modular function of $J$.
In our case both $g(J)$ and $\card{I \intersect A(J)}$ are submodular
functions of $J$.

\begin{proofof}{\Theorem{mainthm}}
The proof applies \Lemma{generalindep}
to the family $\cC = \setst{ A(J) }{ J \subseteq [k]}$.
We must also define a function $g' : \cC \rightarrow \bZ$.
However there is a small issue: it is possible that there exist $J \neq J'$
with $A(J)=A(J')$ but $g(J) \neq g(J')$, so we cannot simply define $g'(A(J)) = g(J)$.
Instead, we define the value of $g'(A(J))$ according the tightest constraint
on $\card{I \intersect A(J)}$, i.e.,
$$
    g'(C) \defeq \min \setst{ g(J) }{ A(J)=C }
    \qquad\forall C \in \cC.
$$

Now fix $I \in \cI$ and suppose that $C_1$ and $C_2$ are tight,
i.e., $\card{I \intersect C_i} = g'(C_i)$.
Define $h_I : 2^{[k]} \rightarrow \bZ$ by
\[
h_I(J) ~:=~ g(J) - \card{I \intersect A(J)}
     ~=~ \card{A(J) \setminus I} - \sum_{j \in J} (\card{A_j} - b_j).
\]
We claim that $h_I$ is a submodular function of $J$.
This follows because $J \mapsto \card{A(J) \setminus I}$ is a submodular function of $J$
(cf.~\Theorem{coverage-subm} in \Appendix{facts-submodular}),
and $J \mapsto \sum_{j \in J} (\card{A_j} - b_j)$ is a modular function of $J$.

Now choose $J_i$ satisfying $C_i = A(J_i)$ and $g'(C_i) = g(J_i)$, for both $i \in \set{1,2}$.
Then
\[
    h_I(J_i) = g(J_i) - \card{I \intersect A(J_i)} = g'(C_i) - \card{I \intersect C_i} = 0,
\]
for both $i \in \set{1,2}$.
However $h_I \geq 0$, since we assume $I \in \cI$ and therefore
$\card{I \intersect A(J)} \leq g(J)$ for all $J$.
So we have shown that $J_1$ and $J_2$ are both minimizers of $h_I$.
It is well-known that the minimizers of any
submodular function are closed under union and intersection. (See Lemma~\ref{minim-subm}
in \Appendix{facts-submodular}.) So $J_1 \union J_2$ and $J_1 \intersect J_2$ are
also minimizers, implying that $A(J_1 \union J_2)=A(J_1) \union A(J_2) = C_1 \union C_2$
is also tight.

This shows that \Equation{uncross2} holds, so the theorem follows from \Lemma{generalindep}.
\end{proofof}

A similar approach is used for our second construction.


\begin{proofof}{\Theorem{modified}}
Fix $I \in \bar{\cI}$.
Let $J_1$ and $J_2$ satisfy $\card{I \intersect A(J_i)} = \bar{g}(J_i)$.
By considering two cases, we will show that
$$\card{I \intersect A(J_1 \union J_2)} ~\geq~ \bar{g}(J_1 \union J_2),$$
so the desired result follows from \Lemma{generalindep}.

\noindent\textit{Case 1:}
$\max \set{\card{J_1}, \card{J_2}} \geq \tau$.
Without loss of generality, $\card{J_1} \geq \card{J_2}$.
Then $$\bar{g}(J_1 \union J_2) = d = \bar{g}(J_1) = \card{ I \intersect A(J_1) }
\leq \card{ I \intersect A(J_1 \union J_2) }.$$

\noindent\textit{Case 2:}
$\max \set{\card{J_1}, \card{J_2}} \leq \tau-1$.
So $\card{J_1 \union J_2} \leq 2\tau-2$.
We have $\card{I \intersect A(J_i)} = \bar{g}(J_i) = g(J_i)$ for both $i$.
As argued in the proof of \Theorem{mainthm},
we also have $\card{I \intersect A(J_1 \union J_2)} = g(J_1 \union J_2)$.
But $g(J_1 \union J_2) \geq \bar{g}(J_1 \union J_2)$ since $g$ is $(d,\tau)$-large,
so $\card{I \intersect A(J_1 \union J_2)} \geq \bar{g}(J_1 \union J_2)$, as desired.
\end{proofof}

\subsubsection{Putting it all together: Proof of \Theorem{manymatroids}}
\SectionName{lb}

In this section we
use the construction in \Theorem{modified}
to prove \Theorem{manymatroids}, which is restated here.

\repeatclaim{\Theorem{manymatroids}}{\manymatroids}

To prove this theorem, we must construct a family of sets $\cA = \set{A_1,\ldots,A_k}$
where each $\card{A_i} = n^{1/3}$,
and for every $\cB \subseteq \cA$ we must construct a matroid $\mat_\cB$
with the desired properties.
It will be convenient to let $d=n^{1/3}$ denote the size of the $A_i$'s,
to let the index set of $\cA$ be denoted by $U := [k]$,
and to let the index set for $\cB$ be denoted by $U_\cB := \setst{ i \in U }{ A_i \in \cB }$.
Each matroid $\mat_\cB$ is constructed by applying \Theorem{modified}
with the set family $\cB$ instead of $\cA$, so its independent sets are
\[
    \cI_\cB
    ~:=~ \Big\{~ I ~:~
        \card{I} \leq d
        ~\And~
        \card{I \intersect A(J)} \leq g_\cB(J)
            ~\: \forall J \subseteq U_\cB,\, \card{J} < \tau
        ~\Big\}.
\]
where the function $g_{\indices} : 2^{U_\indices} \rightarrow \reals$ is defined as in
\Equation{fdef}, taking all $b_i$'s to be equal to a common value $b$:
\begin{align*}
g_{\indices}(\indicest)
    ~:=~ \sum_{j \in J} b \:-\: \Big( \sum_{j \in J} \card{A_j} - \card{A(J)} \Big)
    ~=~ (b-d)\card{J} + \card{A(\indicest)}
            \qquad\forall\indicest \subseteq U_\cB.
\end{align*}

Several steps remain. We must choose the set family $\cA$, then choose parameters
carefully such that, for every $\cB \subseteq \cA$, we have
\begin{itemize}
\item \Pone: $\mat_\cB$ is indeed a matroid,
\item \Ptwo: $\rankf_{\mat_\cB}(A_i) = 8 \log k$ for all $A_i \in \cB$, and
\item \Pthree: $\rankf_{\mat_\cB}(A_i) = \card{A}$ for all $A_i \in \cA \setminus \cB$.
\end{itemize}

Let us start with \Ptwo. Suppose $A_i \in \cB$.
The definition of $\cI_\cB$ includes the constraint $\card{I \intersect A_i} \leq g_\cB(\set{i})$,
which implies that $\rankf_{\mat_\cB}(A_i) \leq g_\cB(\set{i}) = b$.
This suggests that choosing $b := 8 \log k$ may be a good choice to satisfy \Ptwo.

On the other hand, if $A_i \not \in \cB$ then \Pthree\ requires that $A_i$ is independent in $\mat_\cB$.
To achieve this, we need the constraints $\card{I \intersect A(J)} \leq g_\cB(J)$
to be as loose as possible, i.e., $g_\cB(J)$ should be as large as possible.
Notice that $g_\cB(J)$ has two terms, $\sum_{j \in J} b$, which grows as a function of $J$,
and $- \big( \sum_{j \in J} \card{A_j} - \card{A(J)} \big)$, which is non-positive.
So we desire that $\card{A(J)}$ should be as close as possible to $\sum_{j \in J} \card{A_j}$,
for all $J$ with $\card{J} < \tau$.
Set systems with this property are equivalent to expander graphs.

\begin{definition}
\DefinitionName{nice}
Let $G=(U \cup V, E)$ be a bipartite graph.
For $J \subseteq U$, define
$$
\neigh(J) \defeq \setst{ v }{ \exists u \in J \text{ such that } \set{u,v} \in E }.
$$
The graph $G$ is called a $(d,L,\epsilon)$-\newterm{expander} if
$$ \setlength{\arraycolsep}{1pt}
\begin{array}{rll}
|\neigh(\set{u})| &~=~ d \qquad&\forall u \in U \\
|\neigh(J)| &~\geq~ (1-\epsilon) \cdot d \cdot |J| \qquad&\forall J \subseteq U,\, |J| \leq L.
\end{array}
$$
Additionally, $G$ is called a lossless expander if $\epsilon < 1/2$.
\end{definition}

Given such a graph $G$, we construct the set family
$\cA = \set{A_1,\ldots,A_k} \subseteq 2^\ground$
by identifying $U = [k]$, $V = \ground$,
and for each vertex $i \in U$ defining $A_i := \neigh(\set{i})$.
The resulting sets satisfy:
\begin{equation}
\EquationName{nice11}
\setlength{\arraycolsep}{1pt}
\begin{array}{rll}
\card{A_i} &~=~ d
    &\qquad\forall i \in U \\
\card{A(J)} &~\geq~ (1-\epsilon) \cdot d \cdot |J|
    &\qquad\forall J \subseteq U,\, |J| \leq L \\
\implies\qquad
\sum_{j \in J} \card{A_j} - \card{A(J)} &~\leq~ \epsilon \cdot d \cdot |J|
    &\qquad\forall J \subseteq U,\, |J| \leq L.
\end{array}
\end{equation}
This last inequality will allow us to show that $g_\cB(J)$ is sufficiently large.

To make things concrete, let us now state the expander construction
that we will use.
Lossless expanders are well-studied~\cite{guv,HLW},
and several probabilistic constructions are known, both in folklore
and in the literature \cite[Lemma 3.10]{BMRV},
\cite[\S 1.2]{HLW}, \cite[Theorem 26]{SS}, \cite[Theorem 4.4]{Vadhan}.
The following construction of Buhrman et al.\ \cite[Lemma 3.10]{BMRV} has parameters that
match our requirements.

\begin{theorem}
\TheoremName{expanders}
    Suppose $k \geq 8$, $n \geq 25 L \log(k)/\epsilon^2$, and $d \geq \log(k)/ 2\epsilon$.
    Then there exists a graph $G=(U \cup V, E)$
    with $\card{U}=k$ and $\card{V}=n$ that is a $(d,L,\epsilon)$-lossless expander.
\end{theorem}

The next theorem states another (folklore) probabilistic construction
that also matches our requirements.
We include a proof in \Appendix{expander} for the sake of completeness,
and because we will require a slight variant in \Section{implications}.

\newcommand{\expthm}
{
    Let $G=(U \union V, E)$ be a random multigraph where
    $\card{U}=k$, $\card{V}=n$, and every $u \in U$
    has exactly $d$ incident edges, each of which has an endpoint chosen
    uniformly and independently from all nodes in $V$.
    Suppose that $k \geq 4$, $d \geq \log(k)/\epsilon$ and $n \geq 16 L d/\epsilon$.
    Then, with probability at least $1-2/k$,
    $$
    |\neigh(J)| ~\geq~ (1-\epsilon) \cdot d \cdot |J| \qquad\forall J \subseteq U,\, |J| \leq L.
    $$
    If it is desired that $|\neigh(\set{u})| = d$ for all $u \in U$
    then this can be achieved by replacing any parallel edges incident on $u$
    by new edges with distinct endpoints.
    This cannot decrease $|\neigh(J)|$ for any $J$.
}

\begin{theorem}
\TheoremName{ourexpanders}
\expthm
\end{theorem}


%


We require an expander with the following parameters.
Recall that $n$ is arbitrary and $k = 2^{o(n^{1/3})}$.
$$
    d     ~:=~ n^{1/3}\qquad\quad
    L       ~:=~ \frac{ n^{1/3} }{ 2 \log k }\qquad\quad
    \epsilon~:=~ \frac{2 \log k}{n^{1/3}}
$$
These satisfy the hypotheses of \Theorem{expanders} (and \Theorem{ourexpanders}),
so a $(d,L,\epsilon)$-expander exists,
and a set family $\cA$ satisfying \Equation{nice11} exists.
Next we use these properties of $\cA$ to show
that \Pone, \Ptwo\ and \Pthree\ hold.

The fact that \Pone\ holds follows from \Theorem{modified} and the following claim.
Recall that $b = 8 \log k$.

\begin{claim}
    \ClaimName{mtlarge}
    Set $\tau = n^{1/3} / 4 \log k $.
    Then $g_\indices$ is $(d,\tau)$-large,
    as defined in \eqref{eq:large}.
\end{claim}

\begin{proof}
Consider any $J \subseteq U_\cB$ with $\card{J} \leq 2\tau-2$.
Then
\begin{align}
\nonumber
g_{\indices}(\indicest)
    &~=~ (b - d)\card{\indicest} + \card{A(\indicest)}
        \\\nonumber
    &~\geq~ b \card{\indicest} - \epsilon d \card{\indicest}
        \qquad\text{(by \Equation{nice11}, since $\card{\indicest} \leq 2 \tau-2 \leq L$)}\\
        \EquationName{fbb}
    &~=~ \frac{3b}{4} \card{\indicest}
        \qquad\text{(since $\epsilon = b/4d$)}.
\end{align}
This shows $g_\indices(\indicest) \geq 0$.
If additionally $\card{\indicest} \geq \tau$ then
$g_\indices(\indicest) \geq (3/4) b \tau > d$.
\end{proof}

The following claim implies that \Ptwo\ holds.

\begin{claim}
    \ClaimName{smallrank}
    For all $\cB \subseteq \cA$ and all $A_i \in \cB$ we have
    $\rankf_{\mat_\cB}(A_i) = b$.
\end{claim}

\begin{proof}
The definition of $\cI_\cB$ includes the constraint
$\card{I \intersect A_i} \leq g_\cB(\set{i}) = b$.
This immediately implies $\rankf_{\mat_\cB}(A_i) \leq b$.
To prove that equality holds, it suffices to prove that
$g_\cB(J) \geq b$ whenever $\card{J} \geq 1$, since
this implies that every constraint in the definition of $\cI_\cB$
has right-hand side at least $b$
(except for the constraint corresponding to $J=\emptyset$, which is vacuous).
For $\card{J}=1$ this is immediate, and for $\card{J} \geq 2$ we
use \eqref{eq:fbb} to obtain $ g_\cB(J) = 3b\card{J}/4 > b $.
\end{proof}

Finally, the following claim implies that \Pthree\ holds.

\begin{claim}
    \ClaimName{largerank}
    For all $\cB \subseteq \cA$ and all $A_i \in \cA \setminus \cB$ we have
    $\rankf_{\mat_\cB}(A_i) = d$.
\end{claim}

\begin{proof}
Since $d = \card{A_i}$, the condition $\rankf_{\mat_\cB}(A_i) = d$ holds
iff $A_i \in \cI_\cB$.
So it suffices to prove that $A_i$ satisfies all constraints
in the definition of $\cI_\cB$.
%

The constraint $\card{A_i} \leq d$ is trivially satisfied.
So it remains to show that for every $J \subseteq U_\cB$ with $\card{J} < \tau$, we have
\begin{equation}
\EquationName{Aifeas}
\card{ A_i \intersect A(J) }  ~\leq~ g_\cB(J).
\end{equation}
This is trivial if $J = \emptyset$, so assume $\card{J} \geq 1$.
We have
\begin{align*}
 \card{A_i \intersect A(J)}
    &~=~ \card{A_i} + \card{A(J)} - \card{A(J+i)}  \\
    &~\leq~ d + d \card{J} - (1-\epsilon) d \card{J+i}
        \qquad\text{(by \Equation{nice11})}\\
    &~=~ \frac{b \, \card{J+i}}{4}
        \qquad\text{(since $\epsilon = b/4d$)}
        \\
    &~\leq~ \frac{b \, \card{J}}{2} \\
    &~\leq~ g_\cB(J) \qquad\text{(by \Equation{fbb})}.
\end{align*}
This proves \Equation{Aifeas}, so $A_i \in \cI_\cB$, as desired.
\end{proof}



\subsection{Concentration Properties of Submodular Functions}
\SectionName{subsec-concentr}

In this section we provide a strong concentration
bound for submodular functions.

\newcommand{\OURTALAGRAND}{
    Let $f : 2^\ground \rightarrow \bR_+$ be a non-negative, monotone, submodular, $1$-Lipschitz function.
    Let the random variable $X \subseteq \ground$ have a product distribution.
    For any $b, t \geq 0$,
    $$
    \prob{ f(X) \leq b - t \sqrt{b} } \cdot \prob{ f(X) \geq b } ~\leq~ \exp( - t^2 / 4 ).
    $$
}

\begin{theorem}
\TheoremName{OURTALAGRAND}
\OURTALAGRAND
\end{theorem}

To understand \Theorem{OURTALAGRAND}, it is instructive to compare it with known results.
For example, the Chernoff bound is precisely a concentration
bound for \emph{linear}, Lipschitz functions.
On the other hand, if $f$ is an arbitrary 1-Lipschitz function then McDiarmid's inequality
implies concentration, although of a much weaker form, with standard deviation roughly $\sqrt{n}$.
If $f$ is additionally known to be submodular, then we can apply
\Theorem{OURTALAGRAND} with $b$ equal to a median, which can be much smaller than $n$.
So \Theorem{OURTALAGRAND} can be viewed as saying that McDiarmid's inequality can be
significantly strengthened when the given function is known to be submodular.

Our proof of \Theorem{OURTALAGRAND} is based on the Talagrand inequality
\cite{Talagrand,AlonSpencer,MolloyReed,Janson}.
Independently, Chekuri et al.~\cite{CVZ10} proved a similar result using the FKG
inequality.
Concentration results of this flavor can also be proven using the framework of self-bounding
functions~\cite{BLM}, as observed in an earlier paper by Hajiaghayi et al.~\cite{Hajiaghayi} (for a specific class of submodular functions);
see also the survey by Vondr\'ak~\cite{V10}.

\Theorem{OURTALAGRAND} most naturally implies concentration around a median of $f(X)$.
As shown in the following corollary, this also implies concentration around the expected value.
This corollary, with better constants, also follows from the results of Chekuri et al.~\cite{CVZ10}
and Vondr\'ak~\cite{V10}

\newcommand{\rankconcentr}{
    Let $f : 2^\ground \rightarrow \bR_+$ be a non-negative, monotone, submodular, $1$-Lipschitz function.
    Let the random variable $X \subseteq \ground$ have a product distribution.
    For any $0 \leq \alpha \leq 1$,
    $$
    \prob{ |f(X) - \expect{f(X)}| > \alpha \expect{f(X)} }
      ~\leq~ 4 \exp\big( - \alpha^2 \expect{f(X)} / 422 \big).
    $$
}

\begin{corollary}
\CorollaryName{rank-concentr}
\rankconcentr
\end{corollary}

%

As an interesting application of \Corollary{rank-concentr}, let us consider the case where
$f$ is the rank function of a linear matroid. 
Formally, fix a matrix $A$ over any field.
Construct a random submatrix by selecting the $i\th$ column of $A$ with probability $p_i$,
where these selections are made independently.
Then \Corollary{rank-concentr} implies that the rank of the resulting submatrix is highly
concentrated around its expectation, in a way that does not depend on the number of rows of $A$.
%

The proofs of this section are technical applications of Talagrand's inequality and are provided
in~\Appendix{appendix-special}.
Later sections of the paper use \Theorem{OURTALAGRAND} and \Corollary{rank-concentr}
to prove various results.
In \Section{special} we use these theorems
to analyze our algorithm for PMAC-learning submodular functions under product distributions.
In \Section{approxconv} we use these theorems to give an approximate characterization
of matroid rank functions.

\section{Learning Submodular Functions}
\SectionName{learning}

\subsection{A New Learning Model: The PMAC Model}
In this section we introduce a new learning model for learning real-valued functions in the passive, supervised learning paradigm, which we call the PMAC model.  In this model, a learning algorithm is given a collection $\trainS=\set{x_1,x_2,\ldots, x_\ell}$
of polynomially many sets drawn i.i.d.~from
some fixed, but unknown, distribution $D$ over an instance space $\X$.  There is also a fixed but unknown function $\targetf :\X \rightarrow \bR_+$,
and the algorithm is given the value of $\targetf$ at each set in $\trainS$.
The algorithm may perform an arbitrary polynomial time computation on the  examples $\set{ (x_i,\targetf(x_i)) }_{1 \leq i \leq \ell}$,
then must output another function $f : \X \rightarrow \reals_+$.
This function is called a ``hypothesis function''.
The goal is that, with high probability, $f$ is a good approximation of $\targetf$
for most points in $D$.
Formally:

\begin{definition}
Let $\cF$ be a family of non-negative, real-valued functions with domain $\X$.
We say that an algorithm $\cA$ \,\newterm{PMAC-learns} $\cF$ with
approximation factor $\alpha$ if, for any distribution $D$ over $\X$,
for any target function $\targetf \in \cF$,
and for $\epsilon \geq 0 $ and $\delta \geq 0$ sufficiently small:
\begin{itemize}
\item The input to $\cA$ is a sequence of pairs
$\set{ (x_i,\targetf(x_i)) }_{1 \leq i \leq \ell}$
where each $x_i$ is i.i.d.~from $D$.

\item The number of inputs $\ell$ provided to $\cA$ and the running time of $\cA$
are both at most $\poly(n, 1/\epsilon, 1/\delta)$.

\item The output of $\cA$ is a function $f : \X \rightarrow \bR$
that can be evaluated in time $\poly(n, 1/\epsilon, 1/\delta)$ and that satisfies
\[
\operatorname{Pr}_{x_1,\ldots,x_\ell \sim D}
\Big[~~
\probover{x \sim D}{f(x) \leq \targetf(x) \leq \alpha \cdot f(x)} \geq 1-\epsilon
~~\Big]
~\geq~ 1-\delta.
\]
\end{itemize}
\end{definition}

The name PMAC stands for ``Probably Mostly Approximately Correct''.
It is an extension of the PAC model to learning non-negative, real-valued functions,
allowing multiplicative error $\alpha$.
The PAC model for learning boolean functions is precisely the special case when $\alpha=1$.


%

In this paper we focus on the PMAC-learnability of submodular functions. In this case
$\X=\set{0,1}^n$ and $\cF$ is the family of all non-negative, monotone, submodular functions.
We note that it is quite easy to PAC-learn the class of \emph{boolean} submodular functions.
Details are given in \Appendix{boolean}.
The rest of this section considers the much more challenging task of
PMAC-learning the general class of real-valued, submodular functions.


\subsection{Product Distributions}
\SectionName{special}
A first natural and common step in studying learning problems is  
 to study learnability of functions
when the examples are distributed according to the uniform distribution or a product distribution
\cite{KKMS,KDS,LMN}.
In this section we consider learnability of submodular functions when the underlying distribution is a product distribution.
Building on our concentration results in \Section{subsec-concentr} we
provide an algorithm that PMAC learns the class of Lipschitz submodular   functions  with a constant approximation factor.

We will let $L < M < H$ and $K$ be universal constants,
whose values we can take to be $L = 10550$, $M = 11250$, $H=12500$, and $K=26000$.
We begin with the following technical lemma which states some useful concentration bounds. 

\newcommand{\usefullemma}
{Let $f : 2^{[n]} \rightarrow \bR$ be a non-negative, monotone, submodular, $1$-Lipschitz function.
Suppose that $S_1, \ldots, S_l$ are drawn from a product distribution $D$ over $2^{[n]}$.
Let $\mu$ the empirical average
$\mu = \sum_{i=1}^\ell \targetf(S_i) / \ell$, which is our estimate for $\expectover{S \sim D}{\targetf(S)}$.
Let $\epsilon,\delta \leq 1/5$.
We have:
\begin{enumerate}
\item[(1)] If $\expect{\targetf(S)} > H \log(1/\epsilon)$ and $\ell \geq 16 \log(1/\delta)$ then
$$\prob{ \mu \geq M \log(1/\epsilon) } ~\geq~ 1-\delta/4.$$

\item[(2)] If $\expect{\targetf(S)} > L \log(1/\epsilon)$ and $\ell \geq 16 \log(1/\delta)$ then
$$\prob{ \smallfrac{5}{6} \expect{\targetf(S)} \leq \mu \leq \smallfrac{4}{3} \expect{\targetf(S)}}
    ~\geq~ 1 - \delta/4.$$

\item[(3)] If $\expect{\targetf(S)} \leq H \log(1/\epsilon)$ then
$$
\prob{ \targetf(S) < K \log(1/\epsilon) } ~\geq~ 1 - \epsilon.
$$

\item[(4)] If $\expect{\targetf(S)} < L \log(1/\epsilon)$ and $\ell \geq 16 \log(1/\delta)$ then
$$
\prob{ \mu < M \log(1/\epsilon) } ~\geq~ 1-\delta/4.
$$
\end{enumerate}
}

\begin{lemma}
\LemmaName{useful}
\usefullemma
\end{lemma}

The proof of~\Lemma{useful},
which is provided in Appendix~\ref{appendix-learn-product},
follows easily from~\Theorem{OURTALAGRAND} and \Corollary{rank-concentr}.
We now present our main result in this section.

\newcommand{\thmprodalg}{
Let $\cF$ be the class of non-negative, monotone, 1-Lipschitz, submodular functions
with ground set $\ground$ and minimum non-zero value $1$.
Let $D$ be a product distribution on $\set{0,1}^n$.
For any sufficiently small $\epsilon > 0$ and $\delta > 0$,
\Algorithm{talagrand} PMAC-learns $\cF$
with approximation factor $\alpha=K \log(1/\epsilon)$.
The number of training examples used is
$\ell= n \log(n/\delta) / \epsilon + 16 \log(1/\delta)$.

If it is known a priori that $\expect{f^*(S)} \geq L \log(1/\epsilon)$ then the approximation
factor improves to $8$, and the number of examples can be reduced to
$\ell = 16 \log(1/\delta)$, which is independent of $n$ and $\epsilon$.
}
\begin{theorem}
\TheoremName{productalg}
\thmprodalg
\end{theorem}

\begin{algorithm}
\begin{itemize}
\item Let $\mu = \sum_{i=1}^\ell \targetf(S_i) / \ell$.
\item \textit{Case 1:} If $\mu \geq M \log(1/\epsilon)$,
    then return the constant function $f = \mu/4$.
\item \textit{Case 2:} If $\mu < M \log(1/\epsilon)$, then compute the set
$ U = \Union_{i \::\: \targetf(S_i)=0} ~ S_i $.
Return the function $f$ where $f(A)=0$ if $A \subseteq U$ and $f(A)=1$ otherwise.
\end{itemize}
\caption{\: An algorithm for PMAC-learning a non-negative, monotone, $1$-Lipschitz, submodular function
$\targetf$ with minimum non-zero value $1$,
when the examples come from a product distribution.
Its input is a sequence of labeled training examples
$(S_1,\targetf(S_1)), \ldots, (S_\ell,\targetf(S_\ell))$, parameters $\epsilon$ and $\ell$.
}
\AlgorithmName{talagrand}
\end{algorithm}

%

\begin{proof}
We begin with an overview of the proof.
Consider the expected value of $\targetf(S)$ when $S$ is drawn from distribution $D$.
When this expected value of $\targetf$ is large compared to $\log(1/\epsilon)$,
we simply output a constant function given by the empirical average $\mu$ estimated by the algorithm.
Our concentration bounds for submodular functions (\Theorem{OURTALAGRAND} and \Corollary{rank-concentr})
allow us to show that this constant function provides a good estimate.
However, when the expected value of $\targetf$ is small,
we must carefully handle the zeros of $\targetf$,
since they may have large measure under distribution $D$.
The key idea here is to use the fact that the zeros of a non-negative, monotone, submodular function have special structure:
they are both union-closed and downward-closed, so it is sufficient to \emph{PAC-learn} the Boolean \NOR function which indicates the
zeros of $\targetf$.

We now present the proof formally.
%
By \Lemma{useful}, with probability at least $1-\delta$ over the choice of examples,
we may assume that the following implications hold.
\begin{equation}
\EquationName{assumedimplication}
\begin{split}
\mu \geq M \log(1/\epsilon) &\quad\implies\quad
    \expect{\targetf(S)} \geq L \log(1/\epsilon) \quad\text{and}\quad
    \smallfrac{5}{6} \expect{\targetf(S)} \leq \mu \leq \smallfrac{4}{3}
    \expect{\targetf(S)}
    \\
\mu < M \log(1/\epsilon) &\quad\implies\quad
    \expect{\targetf(S)} \leq H \log(1/\epsilon).
\end{split}
\end{equation}
Now we show that the function $f$ output by the algorithm approximates $\targetf$
to within a factor $K \log(1/\epsilon)$.

\vspace{3pt}
\noindent
\textbf{Case 1:} $\mu \geq M \log(1/\epsilon)$.
Since we assume that \eqref{eq:assumedimplication} holds, we have
$\smallfrac{5}{6} \expect{\targetf(S)} \leq \mu \leq \smallfrac{4}{3} \expect{\targetf(S)}$
and $\expect{\targetf(S)} \geq L \log(1/\epsilon)$.
Using these together with~\Corollary{rank-concentr} we obtain:
\begin{equation}
\begin{split}
\EquationName{talagLargeRank}
\prob{ \mu/4 \leq \targetf(S) \leq 2\mu }
    &~~\geq~~ \prob{ \smallfrac{1}{3} \expect{\targetf(S)} \leq \targetf(S) \leq \smallfrac{5}{3}
    \expect{\targetf(S)} } \\
    &~~\geq~~ 1-\prob{ \abs{\targetf(S) - \expect{\targetf(S)}} \geq (2/3) \expect{\targetf(S)} } \\
    &~~\geq~~ 1-4 \exp\big(-\expect{\targetf(S)}/950 \big)
    ~~\geq~~ 1-\epsilon,
\end{split}
\end{equation}
since $L \geq 4000$ and $\epsilon \leq 1/2$.
Therefore, with confidence at least $1-\delta$,
the constant function $f$ output by the algorithm
approximates $\targetf$ to within a factor $8$
on all but an $\epsilon$ fraction of the distribution.

\vspace{3pt}
\noindent
\textbf{Case 2:} $\mu < M \log(1/\epsilon)$.
As mentioned above, we must separately handle the zeros and the non-zeros of $\targetf$.
To that end, define
$$
\cP = \setst{ S }{ \targetf(S) > 0 } \qquad\text{and}\qquad
\cZ = \setst{ S }{ \targetf(S) = 0 }.
$$
Recall that the algorithm sets $U = \Union_{\targetf(S
_i)=0} S_i$.
Monotonicity and submodularity imply that $\targetf(U) = 0$.
Furthermore, setting $\cL = \setst{ T }{ T \subseteq U }$,
monotonicity implies that
\begin{equation}
\EquationName{zeros}
\targetf(T)~=~0 \qquad\forall T \in \cL.
\end{equation}

We wish to analyze the measure of the points for which the function $f$ output by the algorithm
fails to provide a good estimate of $\targetf$.
So let $S$ be a new sample from $D$ and let $\cE$ be the event that $S$ violates the inequality
$$
    f(S) ~\leq~ \targetf(S) ~\leq~ \big(K \log(1/\epsilon)\big) \cdot f(S).
$$
Our goal is to show that, with probability $1-\delta$ over the training examples,
we have $\prob{ \cE } \leq \epsilon$.
Clearly
$$
    \prob{ \cE } \:=\: \prob{\: \cE \:\wedge\: S \!\in\! \cP \:}
                 \:+\: \prob{\: \cE \:\wedge\: S \!\in\! \cZ \:}.
$$
We will separately analyze these two probabilities.

First we analyze the non-zeros of $\targetf$.
So assume that $S \in \cP$, which implies that $\targetf(S) \geq 1$ by our hypothesis.
Then $S \not\subseteq U$ (by \Equation{zeros}), and hence $f(S)=1$
by the definition of $f$.
Therefore the event $\cE \,\wedge\, S \!\in\! \cP$ can only occur when
$\targetf(S) > K \log(1/\epsilon)$.
Since we assume that \eqref{eq:assumedimplication} holds,
we have $\expect{\targetf(S)} \leq H \log(1/\epsilon)$,
so we can apply \Lemma{useful}, statement (3).
This shows that
$$
\prob{\cE \,\wedge\, S \!\in\! \cP }
~\leq~ \prob{ \targetf(S) > K \log(1/\epsilon) }
~\leq~ \epsilon.
$$


It remains to analyze the zeros of $\targetf$.
Assume that $S \in \cZ$, i.e., $\targetf(S)=0$.
Since our hypothesis has $f(S)=0$ for all $S \in \cL$,
the event $\cE \,\wedge\, S \!\in\! \cZ$ holds only if $S \in \cZ \setminus \cL$.
The proof now follows from \Claim{handlezeros}.
\end{proof}

\begin{claim}
\ClaimName{handlezeros}
With probability at least $1-\delta$, the set $\cZ \setminus \cL$ has measure at most $\epsilon$.
\end{claim}
\begin{proof}
The idea of the proof is as follows.
At any stage of the algorithm, we can compute the set $U$ and the subcube
$\cL = \setst{ T }{ T \subseteq U }$.
We refer to $\cL$ as the algorithm's \newterm{null subcube}.
Suppose that there is at least an $\epsilon$ chance that a new example is a zero of $\targetf$,
but does not lie in the null subcube.
Then such a example should be seen in the next sequence of $\log(1/\delta)/\epsilon$ examples,
with probability at least $1-\delta$.
This new example increases the dimension of the null subcube by at least one,
and therefore this can happen at most $n$ times.

Formally, for $k \leq \ell$, define
$$
    U_k ~=~ \Union_{\substack{i \leq k \\ \targetf(S_i)=0}} \!\!\!\! S_i
        \qquad\text{ and }\qquad
    \cL_k ~=~ \setst{ S }{ S \subseteq \cU_k }.
$$
As argued above, we have $\cL_k \subseteq \cZ$ for any $k$.
Suppose that, for some $k$, the set $\cZ \setminus \cL_k$ has measure at least $\epsilon$.
Define $k' = k+\log(n/\delta)/\epsilon$. Then amongst the subsequent examples
$S_{k+1},\ldots,S_{k'}$, the probability that none of them lie in $\cZ \setminus \cL_k$ is at most
$$(1-\epsilon)^{\log(n/\delta)/\epsilon} \leq \delta/n.$$ On the other hand, if one of them does lie
in $\cZ \setminus \cL_k$, then $\card{\cU_{k'}} > \card{\cU_k}$. But $\card{\cU_k} \leq n$ for all
$k$, so this can happen at most $n$ times. Since $\ell \geq n \log(n/\delta) / \epsilon$, with
probability at least $\delta$ the final set $\cZ \setminus \cL_\ell$ has measure at most $\epsilon$.
\end{proof}

The class $\cF$ defined in \Theorem{productalg} contains the class of matroid rank functions.
We remark that \Theorem{productalg} can be easily modified to handle
the case where the minimum non-zero value for functions in $\cF$ is $\eta<1$.
To do this, we simply modify Step $2$ of the algorithm to output $f(A)=\eta$
for all $A \not \subseteq U$.
The same proof shows that this modified algorithm has an approximation factor of
$K \log(1/\epsilon)/\eta$.


\subsection{Inapproximability under Arbitrary Distributions}
\SectionName{general-lower}

The simplicity of \Algorithm{talagrand} might raise one's hopes that
a constant-factor approximation is possible under arbitrary distributions.
However, we show in this section that no such approximation is possible.
In particular, by making use of the new family of matroids we presented in Section~\ref{section-new-extremal-matroid}, we show that no algorithm can PMAC-learn the class of non-negative, monotone,
submodular functions with approximation factor $o({n^{1/3}}/{\log n})$. Formally:

\newcommand{\thmmainlb}{
Let $\ALG$ be an arbitrary learning algorithm that uses only a polynomial number of training
examples drawn i.i.d.\ from the underlying distribution.
There exists a distribution $D$ and a submodular target function $\targetf$ such that, with probability at least $1/8$
(over the draw of the training samples),
the hypothesis function $f$ output by $\ALG$ does not approximate $\targetf$ within a
$o({n^{1/3}}/{\log n})$ factor on at least a $1/4$ fraction of the
examples under $D$.
    This holds even for the subclass of matroid rank functions.
}

\begin{theorem}
\TheoremName{mainlb}
\thmmainlb
\end{theorem}

\begin{proof}
To show the lower bound, we use the family of matroids from~\Theorem{manymatroids} in \Section{lb},
whose rank functions take wildly varying values on large set of points.
The high level idea is to show that for a super-polynomial sized
set of $k$ points in $\set{0,1}^{n}$, and for \emph{any} partition of those
points into \textsc{High} and \textsc{Low}, we can construct a matroid where the points
in \textsc{High} have rank $r_\mathrm{high}$ and the points in \textsc{Low} have rank
$r_\mathrm{low}$, and the ratio
$r_\mathrm{high}/r_\mathrm{low} = \tilde{\Omega}(n^{1/3})$.
This then implies hardness for learning over the uniform
distribution on these $k$ points from any polynomial-sized sample,
even with value queries.

To make the proof formal, we  use the probabilistic method. Assume that $\ALG$ uses $\ell \leq n^c$ training examples for
some constant $c$.
To construct a hard family of submodular functions, we will apply \Theorem{manymatroids} with $k = 2^t$ where $t = c \log(n) + 3$.
Let $\cA$ and $\cM$ be the families that are guaranteed to exist by \Theorem{manymatroids}.
Let the underlying distribution $D$ on $2^\ground$ be the uniform distribution on $\cA$.
(We note that $D$ is {\em not} a product distribution.)
Choose a matroid $\mat_\cB \in \cM$ uniformly at random and let the target function be
$\targetf = \rankf_{\mat_\cB}$. Clearly $\ALG$ does not know $\cB$.

Assume that $\ALG$ uses a set $\cS$ of $\ell$ training examples.
For any $A \in \cA$ that is not a training example,
the algorithm $\ALG$ has \emph{no information} about $\targetf(A)$; in particular, the conditional distribution of its value, given
  $\cS$, remains uniform in $\{8t,|A|\}$.
So $\ALG$ cannot determine its value better than
randomly guessing between the two possible values $8 t$ and $\card{A}$.
The set of non-training examples has measure $1-2^{-t + \log \ell}$.
 Thus
$$
\E_{\targetf, \cS}\Bigg[~
    \Pr_{A \sim D} \Big[\:
        \targetf(A) \not \in \big[f(A), \smallfrac{ n^{1/3}}{16 t} f(A)\big]
    \:\Big]
~\Bigg]
~\geq~ \frac{1-2^{-t + \log \ell}}{2} ~\geq~ 7/16.
$$
Therefore, there exists $\targetf$ such that
$$
\Pr_{\cS}\Bigg[~
    \Pr_{A \sim D} \Big[\:
        \targetf(A) \not \in \big[f(A), \smallfrac{ n^{1/3}}{16 t} f(A)\big]
    \:\Big]
    \geq 1/4
~\Bigg]
~\geq~ 1/8.
$$
That is there exists $\targetf$ such that with probability at least $1/8$
(over the draw of the training samples) we have that
the hypothesis function  $f$ output by $\ALG$ does not approximate $\targetf$ within a $o({n^{1/3}}/{\log n})$ factor on at least $1/4$ fraction of the
examples  under $D$.
\end{proof}

We can further show that the lower bound in \Theorem{mainlb} holds even if the algorithm is told the underlying distribution,
even if the algorithm can query the function on inputs of its choice,
and even if the queries are adaptive. In other words, this inapproximability still holds in the PMAC model augmented with value queries.
Specifically:

\newcommand{\thmmainlbadaptive}{
Let $\ALG$ be an arbitrary learning algorithm that uses only a polynomial number of training
examples, which can be either drawn i.i.d.\ from the underlying distribution or value queries.
There exists a distribution $D$ and a submodular target function $\targetf$ such that, with probability at least $1/4$ (over the draw of the training samples),
the hypothesis function output by $\ALG$ does not approximate $\targetf$ within a $o({n^{1/3}}/{\log
n})$ factor on at least a $1/4$ fraction of the
examples under $D$.
    This holds even for the subclass of matroid rank functions.
}

\begin{theorem}
\TheoremName{mainlbadaptive} \thmmainlbadaptive
\end{theorem}

\Theorem{mainlb} is an information-theoretic hardness result.
A slight modification yields \Corollary{crypto}, which is a complexity-theoretic hardness result.

\newcommand{\cryptocor}{
    Suppose one-way functions exist.
    For any constant $\epsilon>0$,
    no algorithm can
    PMAC-learn the class of non-negative,
    monotone, submodular functions with approximation factor $O({n^{1/3-\epsilon}})$,
    even if the functions are given by polynomial-time algorithms
    computing their value on the support of the distribution.
}
\begin{corollary}
\CorollaryName{crypto}
\cryptocor
\end{corollary}


The proofs of \Theorem{mainlbadaptive} and \Corollary{crypto} are given in~\Appendix{mainlb}.
The lower bound in~\Corollary{crypto} gives a family of
submodular functions that are hard to learn, even though the functions
can be evaluated by polynomial-time algorithms on the \emph{support of the distribution}.
However we do not prove that the functions can be evaluated by polynomial-time algorithms
at \emph{arbitrary points}, and we leave it as an open question whether such a construction is
possible.

\subsection{An $O(\sqrt{\lowercase{n}})$-approximation Algorithm}
\SectionName{general-upper}

In this section we discuss  our most general upper bound for efficiently PMAC-learning the class of non-negative, monotone, submodular functions with with an approximation factor of $O(\sqrt{n})$.

We start with a useful structural lemma concerning submodular functions.
\begin{lemma}[Goemans et al.~\cite{nick09}]
\LemmaName{structure-monotone}
Let $f: 2^{[n]} \rightarrow \reals_{+}$ be a normalized, non-negative, monotone, submodular function.
Then there exists a function $\hf$ of the form $\hf(S) = \sqrt{ w \transpose \ind(S) }$
where $w \in \bR^n_+$
such that for all $S \subseteq [n]$ we have $$\hf(S) ~\leq~ f(S) ~\leq~ \sqrt{n} \hf(S).$$
\end{lemma}

This result, proven by Goemans et al.~\cite{nick09},
follows from properties of submodular polyhedra and John's theorem on approximating centrally-symmetric
convex bodies by ellipsoids~\cite{John}.
We now use it in proving our main algorithmic result.
%
%

\begin{algorithm}
{\bf Input:} A sequence of labeled training examples $\cS= \left\{(S_1,\targetf(S_1)),
(S_2,\targetf(S_2)), \ldots (S_{\ell},\targetf(S_{\ell})) \right\}$.
\begin{itemize}
\item Let $\cSnz=\left\{(A_1,\targetf(A_1)), \ldots, (A_a,\targetf(A_a)) \right\}$ be the
subsequence of $\cS$ with $\targetf(A_i) \neq 0 ~~\forall i$.
Let $\cSz=\cS \setminus \cSnz$.
Let $\cU_{0}$ be the set of indices defined as
$$
\cU_{0} ~=~ \Union_{\substack{i \leq \ell \\ \targetf(S_i)=0}} \!\!\!\! {S_i}.
$$


\item For each $i \in [a]$,
let $y_i$ be the outcome of flipping a fair $\set{+1,-1}$-valued coin,
each coin flip independent of the others.
Let $x_i \in \bR^{n+1}$ be the point defined by
$$
x_i ~=~ \begin{cases}
\big(\: \ind(A_i) ,\, \targetf^2(A_i) \:\big)              &\quad\text{(if $y_i=+1$)} \\
\big(\: \ind(A_i) ,\, (n+1) \cdot \targetf^2(A_i) \:\big)  &\quad\text{(if $y_i=-1$)}.
\end{cases}
$$

\item  Find a linear separator $u = (w,-z) \in \bR^{n+1}$, where $w \in \bR^n$ and $z \in \bR$,
such that $u$ is consistent with the labeled examples $(x_i,y_i) \:~\forall i \in [a]$,
and with the additional constraint that $w_j=0 \:~\forall j \in \cU_{0}$.
\end{itemize}
{\bf Output:} The function $f$ defined as
$f(S) = \left(\frac{w \transpose \ind(S)}{(n+1) z} \right)^{1/2}$.

\caption{\: Algorithm for PMAC-learning the class of non-negative, monotone, submodular functions.
} \AlgorithmName{algsubmodular}
\end{algorithm}

\newcommand{\thmnmonotone}{
    Let $\cF$ be the class of non-negative, monotone, submodular
    functions over $\X=2^{[n]}$. There is an algorithm that PMAC-learns $\cF$ with
    approximation factor $\sqrt{n+1}$. That is, for any distribution $D$  over $\X$, for any
    $\epsilon$, $\delta$ sufficiently small, with probability $1-\delta$, the algorithm
    produces a function $f$ that approximates $\targetf$ within a multiplicative factor of
    $\sqrt{n+1}$ on a set of measure $1-\epsilon$ with respect to $D$. The algorithm uses $
    \ell = \frac{48n}{\epsilon} \log \left(\frac{9n}{\delta \epsilon} \right) $ training
    examples and runs in time $\poly(n,1/\epsilon,1/\delta)$.
}
\begin{theorem}
\TheoremName{n-monotone}
\thmnmonotone
\end{theorem}

\begin{proof}
As in \Theorem{productalg},
because of the multiplicative error allowed by the PMAC-learning model, we will
separately analyze the subset of the instance space where $\targetf$ is zero and the
subset of the instance space where $\targetf$ is non-zero. For convenience, let us
define:
$$
\cP = \setst{ S }{ \targetf(S) \neq 0 } \qquad\text{and}\qquad \cZ = \setst{ S }{
\targetf(S) = 0 }.
$$

The main idea of our algorithm is to reduce our learning problem to the standard problem
of learning a binary classifier (in fact, a linear separator) from i.i.d.\ samples in the
passive, supervised learning setting~\cite{KV:book94,Vapnik:book98} with a slight twist
in order to handle the points in $\cZ$. The problem of learning a linear separator in the
passive supervised learning setting is one where the instance space is $\bR^m$, the
samples are independently drawn from some fixed and unknown distribution $D'$ on $\bR^m$,
and there is a fixed but unknown target function $c^* : \bR^m \rightarrow \set{-1,+1}$
defined by $c^*(x) = \sgn(u \transpose x)$ for some vector $u \in \bR^m$.
The examples induced by $D'$ and $c^*$ are called \emph{linearly separable}.

The linear separator learning problem we reduce to is defined as follows. The instance
space is $\bR^m$ where $m=n+1$ and the distribution $D'$ is defined by the following
procedure for generating a sample from it. Repeatedly draw a sample $S \subseteq [n]$
from the distribution $D$ until $\targetf(S) \neq 0$. Next, flip a fair coin. The sample
from $D'$ is
\begin{equation}
\EquationName{coin}
\begin{array}{ll}
\big(\: \ind(S) ,\, \targetf(S)^2 \:\big) &\qquad\text{(if the coin is heads)} \\[2pt]
\big(\: \ind(S) ,\, (n+1) \cdot \targetf(S)^2 \:\big) &\qquad\text{(if the coin is tails).}
\end{array}
\end{equation}
The function $c^*$ defining the labels is as follows:
samples for which the coin was heads are labeled $+1$, and the others are labeled $-1$.

We claim that the distribution over labeled examples induced by $D'$ and $c^*$ is
linearly separable in $\bR^m$.
To prove this we use \Lemma{structure-monotone} which says that there exists a linear
function $\hat{f}(S) = w \transpose \ind(S)$ such that
\begin{equation}
\EquationName{LinearNApprox}
    \hat{f}(S) ~\leq~ \targetf(S)^2 ~\leq~ n \cdot \hat{f}(S)
        \qquad\forall S \subseteq [n].
\end{equation}
Let $u = \big(\, (n+1/2)\cdot w ,\, -1 \,\big) \in \bR^m$.
For any point $x$ in the support of $D'$ we have
\begin{align*}
x \:=\: \big(\, \ind(S) ,\, \targetf(S)^2 \,\big)
&\quad\implies\quad
u \transpose x \:=\: (n+1/2) \cdot \hat{f}(S) - \targetf(S)^2 \:>\: 0
\\
x \:=\: \big(\, \ind(S) ,\, (n+1)\cdot\targetf(S)^2 \,\big)
&\quad\implies\quad
u \transpose x \:=\: (n+1/2) \cdot \hat{f}(S) - (n+1)\cdot\targetf(S)^2 \:<\: 0.
\end{align*}
This proves the claim.

Moreover, due to \eqref{eq:LinearNApprox},
the linear function $\hat{f}$ also satisfies $\hat{f}(S)=0$ for every $S \in \cZ$.
In particular, every training example $S_i$ satisfies $\hat{f}(S_i)=0$ whenever $S_i \in \cZ$,
and moreover
$$
\hat{f}(\set{j})=w_j=0 \quad\forall j \in \cU_{D}
\qquad\text{ where }\qquad
\cU_D ~=~ \Union_{\substack{S_i \in \cZ }}{S_i}.
$$

Our algorithm is now as follows. It first partitions the training set $\cS=
\left\{(S_1,\targetf(S_1)), \ldots, (S_\ell,\targetf(S_\ell))\right\}$ into two sets
$\cSz$ and $\cSnz$, where $\cSz$ is the subsequence of $\cS$ with $\targetf(S_i) = 0$,
and $\cSnz=\cS \setminus \cSz$. For convenience, let us denote the sequence  $\cSnz$ as
$$
    \cSnz ~=~ \Big(\: \big(A_1,\targetf(A_1)\big), \ldots, \big(A_a, \targetf(A_a)\big) \:\Big).
$$
Note that $a$ is a random variable and we can think of the sets the $A_i$ as drawn
independently from $D$, conditioned on belonging to $\cP$. Let
$$
    \cU_{0} ~=~ \Union_{S_i \::\: \targetf(S_i)=0} S_i
    \qquad\text{ and }\qquad
    \cL_0 ~=~ \setst{ S }{ S \subseteq \cU_0 }.
$$

 Using $\cSnz$, the algorithm then constructs a sequence
$ \cSnz' = \big( (x_1,y_1),\ldots,(x_a,y_a) \big) $ of training examples for the binary
classification problem. For each $i \in [a]$, let $y_i$ be $-1$ or $1$, each with
probability $1/2$. Define $x_i$ as in \eqref{eq:coin}:
$$
x_i ~=~
\begin{cases}
\big(\: \ind(A_i) ,\, \targetf(A_i)^2 \:\big) &\qquad\text{(if $y_i = +1$)} \\
\big(\: \ind(A_i) ,\, (n+1) \cdot \targetf(A_i)^2 \:\big) &\qquad\text{(if $y_i = -1$)}.
\end{cases}
$$
The last step of our algorithm is to solve a
linear program in order to find a linear separator $u = (w,-z)$ where
$w \in \bR^n$, $z \in \bR$, and
\begin{itemize}
\item $u$ is consistent with the labeled examples $(x_i,y_i)$ for all $i=1,\ldots,a$, and
\item $w_j=0$ for all $j \in \cU_{0}$.
\end{itemize}
The output hypothesis is $f(S) = \left(\frac{w \transpose \ind(S)}{(n+1) z} \right)^{1/2}$.

To prove correctness, note first that the linear program is feasible; this follows from
our earlier discussion using
the facts that (1) $\cSnz'$ is a set of labeled examples drawn
from $D'$ and labeled by $c^*$, and (2) $\cU_{0} \subseteq \cU_D$. It remains to show
that $f$ approximates the target on most of the points. Let $\cY$ denote the set of
points $S \in \cP$ such that both of the points $(\ind(S), \targetf^2(S))$ and $(\ind(S),
(n+1)\cdot\targetf^2(S))$ are correctly labeled by $\sgn(u \transpose x)$, the linear
separator found by our algorithm. It is easy to see that the function $f$
approximates $\targetf$ to within a factor
$\sqrt{n+1}$ on all the points in the set $\cY$:
for any point $S \in \cY$, we have
\begin{gather*}
w \transpose \ind(S) - z \targetf(S)^2 > 0 \qquad\text{and}\qquad w \transpose \ind(S) - z
(n+1) \targetf(S)^2 < 0
\\
\implies\quad \left(\frac{w \transpose \ind(S)}{(n+1) z} \right)^{1/2}
    ~<~ \targetf(S)
    ~<~ \sqrt{n+1} \left(\frac{w \transpose \ind(S)}{(n+1) z} \right)^{1/2}.
\end{gather*}
So, for any point in $S \in \cY$, the function
$f(S) =\left(\frac{w \transpose \ind(S)}{(n+1) z} \right)^{1/2}$
approximates $\targetf$ to within a factor $\sqrt{n+1}$.

Moreover,  by design the function $f$ correctly labels as $0$ all the examples in
$\cL_0$. To finish the proof, we now note two important facts: for our choice of $\ell
= \frac{16n}{\epsilon} \log \left(\frac{n}{\delta \epsilon} \right)$, with high
probability both $\cP \setminus \cY$  and $\cZ \setminus \cL_0$ have small measure.
The fact that $\cZ \setminus \cL_0$ has small measure follows from an argument similar to the
one in \Claim{handlezeros}. We now prove:

\begin{claim}
If $\ell = \frac{16n}{\epsilon} \log \left(\frac{n}{\delta \epsilon} \right)$, then
 with
probability at least $1-2\delta$, the set $\cP \setminus \cY$ has measure at most
$2\epsilon$ under $D$.
\end{claim}
\begin{subproof}
Let $q = 1-p = \probover{S \sim D}{S \in \cP}$. If $q < \epsilon$ then the claim is
immediate, since $\cP$ has measure at most $\epsilon$. So assume that $q \geq \epsilon$.
Let $\mu = \expect{a} = q \ell$. By assumption $ \mu > 16n \log(n/\delta\epsilon)
\frac{q}{\epsilon}$. Then Chernoff bounds give that
\begin{align*}
\prob{ a < 8n \log (n/\delta\epsilon) \frac{q}{\epsilon} }
    ~<~ \exp(- n \log(n/\delta)q/\epsilon) ~<~ \delta.
\end{align*}
So with probability at least $1-\delta$, we have
    $a \geq 8n \log (qn/\delta\epsilon) \frac{q}{\epsilon}$.
By a standard sample complexity argument~\cite{Vapnik:book98}
(which we reproduce in \Theorem{VCbound} in Appendix \ref{useful-lemmas}),
with probability at least $1-\delta$, any linear separator consistent with $\cS'$ will
be inconsistent with the labels on a set of measure at most $\epsilon/q$ under $D'$.
In particular, this property holds for the linear separator computed by the linear program.
So for any set $S$,
the conditional probability that either
$(\ind(S), \targetf(S)^2)$ or $(\ind(S), (n+1)\cdot\targetf(S)^2)$ is incorrectly labeled,
given that $S \in \cP$, is at most $2 \epsilon / q$.
Thus
$$
\prob{ S \in \cP \And S \not \in \cY }
~=~ \prob{ S \in \cP } \cdot \probg{ S \not \in \cY }{ S \in \cP }
~\leq~ q \cdot (2\epsilon/q),
$$
as required.
\end{subproof}

In summary,
our algorithm produces a hypothesis $f$ that approximates
$\targetf$ to within a factor $n+1$ on the set $\cY \union \cL_\ell$.
The complement of this set is $(\cZ \setminus \cL_\ell) \union (\cP \setminus \cY)$,
which has measure at most $3 \epsilon$, with probability at least $1-3\delta$.
\end{proof}



\paragraph{Remark}
Our algorithm proving \Theorem{n-monotone} is significantly
simpler than the algorithm of Goemans et al.~\cite{nick09} which achieves
a slightly worse approximation factor in the model of
approximately learning everywhere with value queries.

\subsubsection{Extensions}
Our algorithm for learning submodular functions
is quite robust and can be extended to handle more general
scenarios, including forms of noise. In this section we discuss  several such extensions.

It is clear from the proofs of \Theorem{n-monotone} that
any improvements in the approximation factor for approximating submodular functions
by linear functions (i.e., \Lemma{structure-monotone})
for specific subclasses of submodular functions
yield PMAC-learning algorithms with improved approximation factors.

Next, let us consider the
more general case where we do not even assume that the target function is  submodular, but that it is within a factor $\alpha$ of a submodular
function on every point in the instance space. Under this relaxed assumption we are able
to achieve the approximation factor $\alpha \sqrt{n+1}$.
Specifically:

\begin{theorem}
Let $\cF$ be the class of non-negative, monotone,
submodular functions over $\X=2^{[n]}$ and let
$$
\cF' ~=~
\setst{ f }
{ \exists g \in \cF ,\: g(S) \leq f(S) \leq \alpha \cdot g(S)~~\text{for all}~S \subseteq [n]},
$$
for some known $\alpha>1$. There is an algorithm that PMAC-learns $\cF'$ with
approximation factor $\alpha  \sqrt{n+1}$.
The algorithm uses $ \ell ~=~ \frac{48n}{\epsilon} \log \left(\frac{9n}{\delta \epsilon}
\right) $ training examples and runs in time $\poly(n,1/\epsilon,1/\delta)$.
\end{theorem}

\begin{proof}
By assumption, there exists $g \in \cF$ such that $g(S) \leq \targetf(S) \leq
\alpha \cdot g(S)$. Combining this with~\Lemma{structure-monotone}, we get that there
exists $\hat{f}(S) = w \transpose \ind(S)$ such that
$$
w \transpose \ind(S) ~\leq~ \targetf^2(S) ~\leq~ n \cdot \alpha^2 \cdot w \transpose \ind(S)
~~~~~\text{for all}~S \subseteq [n].
$$
We then apply
the algorithm described in \Theorem{n-monotone}  with the following modifications:
$(1)$ in the
second step if $y_i=+1$ we set $x_i = (\ind(S), \targetf^2(S))$ and if $y_i=-1$  we set
$x_i = (\ind(S), \alpha^2 (n+1)  \cdot \targetf(S))$; $(2)$ we output the function $f(S) =
\left(\frac{1}{\alpha^2 (n+1) z} w \transpose \ind(S)\right)^{1/2}$.
It is then easy to show that the distribution over labeled examples induced
by $D'$ and $c^*$ is linearly separable in $\bR^{n+1}$; in particular, $u = (
{\alpha^2}(n+1/2)\cdot w, -1) \in \bR^{n+1}$ defines a good linear separator. The proof then
proceeds as in~\Theorem{n-monotone}.
\end{proof}

We can also extend the result in \Theorem{n-monotone} to the
{\em agnostic case} where we assume that there exists a submodular
function that agrees with the target on all but an $\eta$ fraction of the points; note
that on the $\eta$ fraction of the points the target can be arbitrarily far from a
submodular function.
In this case we can still PMAC-learn with a polynomial
number of samples $O(\frac{n}{\epsilon^2} \log \left(\frac{n}{\delta \epsilon} \right))$,
but using a potentially computationally inefficient procedure.

\begin{theorem}
Let $\cF$ be the class of non-negative, monotone,
submodular functions over $\X=2^{[n]}$. Let
$$\cF' ~=~ \setst{f }{ \exists g \in \cF \text{ s.t. }  f(S) = g(S)
~~\text{on more than $1-\eta$ fraction of the points} }.
$$
There is an algorithm that PMAC-learns $\cF'$ with approximation factor $ \sqrt{n+1}$. That is,
for any distribution $D$  over $\X$, for any $\epsilon$, $\delta$ sufficiently small, with
probability $1-\delta$, the algorithm produces a function $f$ that approximates
$\targetf$ within a multiplicative factor of $\sqrt{n+1}$ on a set of measure
 $1-\epsilon-\eta$ with respect to $D$. The algorithm uses
 $O(\frac{n}{\epsilon^2} \log \left(\frac{n}{\delta \epsilon} \right))$ training examples.
 \end{theorem}

\begin{proofsketch} The proof proceeds as in~\Theorem{n-monotone}.
The main difference is that in the new feature space $\bR^m$, the best linear separator has error
(fraction of mistakes) $\eta$. It is well known that even in the agnostic case the number
of samples needed to learn a separator of error at most $\eta+\epsilon$ is
$O(\frac{n}{\epsilon^2} \log \left(\frac{n}{\delta \epsilon} \right))$
(see~\Theorem{vc-normal} in Appendix~\ref{useful-lemmas}). However, it is NP-hard to
minimize the number of mistakes, even approximately \cite{RV06}, so the resulting
procedure uses a polynomial number of samples, but it is computationally inefficient.
\end{proofsketch}


\section{An Approximate Characterization of Matroid Rank Functions}
\SectionName{approxconv}

We now present an interesting structural result that is an application of the ideas in
\Section{special}.
The statement is quite surprising:
matroid rank functions are very well approximated by \emph{univariate}, concave functions.
The proof is also based on \Theorem{OURTALAGRAND}.
To motivate the result, consider the following easy construction of submodular functions,
which can be found in Lov\'asz's survey~\cite[pp.~251]{L83}

\begin{proposition}
Let $h : \bR \rightarrow \bR$ be concave.
Then $f : 2^\ground \rightarrow \bR$ defined by $f(S) = h(\card{S})$
is submodular.
\end{proposition}

Surprisingly, we now show that a partial converse is true.

\newcommand{\charthmst}{
    There is an absolute constant $c > 1$ such that the following is true.
    Let $f : 2^\ground \rightarrow \bZ_+$ be the rank function of a matroid with no loops,
    i.e., $f(S) \geq 1$ whenever $S \neq \emptyset$.
    Fix any $\epsilon > 0$, sufficiently small.
    There exists a concave function $h : [0,n] \rightarrow \bR$ such that,
    for \textbf{every} $k \in [n]$,
    and for a $1-\epsilon$ fraction of the sets $S \in \binom{\ground}{k}$,
    \[
     h(k) / (c \log(1/\epsilon)) ~\leq~ f(S) ~\leq c  \log(1/\epsilon)  h(k).
    \]
}
\begin{theorem}
\TheoremName{character}
\charthmst
\end{theorem}

The idea behind this theorem is as follows.
For $x \in [0,n]$, we define $h(x)$ to be the expected value of $f$
under the product distribution which samples each element independently with probability $x/n$.
The value of $f$ under this distribution is tightly concentrated around $h(x)$,
by the results of \Section{subsec-concentr} and \Section{special}.
For any $k \in [n]$, the distribution defining $h(k)$ is very similar to
the uniform distribution on sets of size $k$,
so $f$ is also tightly concentrated under the latter distribution.
So the value of $f$ for most sets of size $k$ is roughly $h(k)$.
The concavity of this function $h$ is a consequence of submodularity of $f$.


\comment{
\textbf{Note:} This says that we get an $O(\log^2(1/\epsilon))$-approximation.
The analysis here is very sloppy and it can probably be improved to
a $O(\log(1/\epsilon))$-approximation.
As in \Theorem{productalg}, if $h(k)$ is sufficiently large we get a constant-factor approximation.
}

Henceforth, we will use the following notation.
For $p \in [0,1]$, let $R(p) \subseteq \ground$ denote the random variable obtained
by choosing each element of $\ground$ independently with probability $p$.
For $k \in [n]$, let $S(k) \subseteq \ground$ denote a set of cardinality $k$
chosen uniformly at random.
Define the function $h' : [0,1] \rightarrow \bR$ by
$$
h'(p) ~=~ \expect{ f(R(p)) }.
$$
For any $\tau \in \bR$, define the functions
$g_\tau : [0,1] \rightarrow \bR$ and $g'_\tau : [n] \rightarrow \bR$ by
\begin{align*}
g_\tau(p)  &~=~ \prob{ f(R(p)) > \tau } \\
g'_\tau(k) &~=~ \prob{ f(S(k)) > \tau }.
\end{align*}
Finally, let us introduce the notation $X \cong Y$ to denote that random variables
$X$ and $Y$ are identically distributed.

\begin{lemma}
$h'$ is concave.
\end{lemma}
\begin{proof}
One way to prove this is by appealing to the \newterm{multilinear extension} of $f$,
which has been of great value in recent work \cite{CCPV}.
This is the function $F : [0,1]^{[n]} \rightarrow \bR$
defined by $F(y) = \expect{ f(\hat{y}) }$,
where $\hat{y} \in \set{0,1}^{[n]}$ is a random variable obtained by
independently setting $\hat{y}_i = 1$ with probability $y_i$, and $\hat{y}_i = 0$ otherwise.
Then $h'(p) = F(p,\ldots,p)$.
It is known \cite{CCPV} that
$\frac{\partial^2 F}{\partial y_i \partial y_j} \leq 0$ for all $i,j$.
By basic calculus, this implies that the second derivative of $h'$ is non-positive,
and hence $h'$ is concave.
\end{proof}

\comment{
\begin{proof}
Fix scalars
$$
0 \leq \delta \leq \alpha \leq \beta \leq \gamma \leq 1
\qquad\text{s.t.}\qquad
\alpha+\beta-\delta=\gamma.
$$
We will show that
\begin{equation}
\EquationName{concave}
h'(\alpha) + h'(\beta) ~\geq~ h'(\gamma) + h'(\delta),
\end{equation}
which implies that $h'$ is concave.
(The implication would hold even if we required that $\alpha=\beta=(\gamma+\delta)/2$.)

To this end, consider the following process for generating sets.
Pick each element of $V$ independently with probability $\alpha$
and call the resulting set $A$.
Now construct a new set by removing each element of $A$ independently with probability $\delta/\alpha$
and adding each element not in $A$ independently with probability $(\beta-\delta)/(1-\alpha)$.
Call the resulting set $B$.
One may easily verify that
$$
A \cong R(\alpha) \qquad\qquad
B \cong R(\beta) \qquad\qquad
A \intersect B \cong R(\delta) \qquad\qquad
A \union B \cong R(\gamma).
$$
Submodularity implies
\begin{align*}
f(A) + f(B) &~\geq~ f(A\union B) + f(A \intersect B) \\
\implies
\expect{ f(A) } + \expect{ f(B) }
    &~\geq~ \expect{ f(A\union B) } + \expect{ f(A \intersect B) } \\
\implies
\expect{ f(R(\alpha)) } + \expect{ f(R(\beta)) }
    &~\geq~ \expect{ f(R(\gamma)) } + \expect{ f(R(\delta)) },
\end{align*}
which implies \Equation{concave}.
\end{proof}
}

\begin{lemma}
\LemmaName{gpmonotone}
$g'_\tau$ is a monotone function.
\end{lemma}
\begin{proof}
Fix $k \in [n-1]$ arbitrarily.
Pick a set $S = S(k)$.
Construct a new set $T$ by adding to $S$ a uniformly chosen element of $V \setminus S$.
By monotonicity of $f$ we have $f(S) > \tau \implies f(T) > \tau$.
Thus $\prob{ f(S) > \tau } \leq \prob{ f(T) > \tau }$.
Since $T \cong S(k+1)$, this implies that $g_\tau(k) \leq g_\tau(k+1)$, as required.
\end{proof}

\begin{lemma}
\LemmaName{double}
$g'_\tau(k) \leq 2 \cdot g_\tau(k/n)$,
for all $\tau \in \bR$ and $k \in [n]$.
\end{lemma}
\begin{proof}
This lemma is reminiscent of a well-known property of the Poisson approximation
\cite[Theorem 5.10]{MitzenmacherUpfal}, and the proof is also similar.
Let $p=k/n$. Then
\begin{align*}
g_\tau(p)
    &~=~ \prob{ f(R(p)) > \tau } \\
    &~=~ \sum_{i=0}^n \probg{ f(R(p)) > \tau }{ \card{R(p)}=i } \cdot \prob{ \card{R(p)}=i } \\
    &~=~ \sum_{i=0}^n g'_\tau(i) \cdot \prob{ \card{R(p)}=i } \\
    &~\geq~ \sum_{i=k}^n g'_\tau(k) \cdot \prob{ \card{R(p)}=i }
        \qquad\text{(by \Lemma{gpmonotone})}\\
    &~=~ g'_\tau(k) \cdot \prob{ \card{R(p)} \geq k } \\
    &~\geq~ g'_\tau(k) / 2,
\end{align*}
since the mean $k$ of the binomial distribution $B(n,k/n)$ is also a median.
\end{proof}

\begin{proofof}{\Theorem{character}}
For $x \in [0,n]$, define $h(x) = h'(x/n) = \expect{ f(R(x/n)) }$.
Fix $k \in [n]$ arbitrarily.
We use the same constants $L < M < H$ and $K$ as in \Section{special}.

\vspace{12pt}
\noindent\textit{Case 1.}
Suppose that $h(k) \geq L \log(1/\epsilon)$.
As argued in \Equation{talagLargeRank}, since $L \geq 4000$ and $\epsilon \leq 1/2$ we have
$$
    \prob{ f(R(k/n)) < \frac{1}{3} h(k) } ~\leq~ \epsilon
    \qquad\text{and}\qquad
    \prob{ f(R(k/n)) > \frac{5}{3} h(k)} ~\leq~ \epsilon.
$$
By \Lemma{double}, $\prob{ f(S(k)) > \frac{5}{3} h(k)} \leq 2\epsilon$.
By a symmetric argument, which we omit, one can show that
$\prob{ f(S(k)) < \frac{1}{3} h(k) } \leq 2\epsilon$.
Thus,
$$
\prob{~ \smallfrac{1}{3} h(k) \:\leq\: f(S(k)) \:\leq\: \smallfrac{5}{3} h(k) ~} ~\geq~ 1 - 4 \epsilon.
$$
This completes the proof of Case 1.

\vspace{12pt}
\noindent\textit{Case 2.}
Suppose that $h(k) < L \log(1/\epsilon)$.
This immediately implies that
\begin{equation}
\EquationName{fnotsmall}
    \prob{ f(S(k)) < \frac{ h(k) }{ L \log(1/\epsilon) } }
    ~\leq~
    \prob{ f(S(k)) < 1 }
    ~=~ 0,
\end{equation}
since $k \geq 1$, and since we assume that $f(S) \geq 1$ whenever $S \neq \emptyset$.
These same assumptions lead to the following lower bound on $h$:
\begin{equation}
\EquationName{hnotsmall}
h(k) ~\geq~ \prob{ f(R(k/n)) \geq 1 } ~=~ \prob{ R(k/n) \neq \emptyset }
~=~ 1-(1-k/n)^n ~\geq~ 1-1/e.
\end{equation}
Thus
\begin{align*}
           &\prob{ f(S(k)) > \big(2K \log(1/\epsilon) \big) \cdot h(k)} \\
    ~\leq~ &2 \cdot \prob{ f(R(k/n)) > \big(2K \log(1/\epsilon)\big) \cdot h(k) }
        \qquad\text{(by \Lemma{double})}\\
    ~\leq~ &2 \cdot \prob{ f(R(k/n)) > K \log(1/\epsilon) }
        \qquad\text{(by \Equation{hnotsmall})}\\
    ~\leq~ &2 \cdot \epsilon,
\end{align*}
by \Lemma{useful}, statement (3),
since $\expect{f(R(k/n))} = h(k) < L \log(1/\epsilon)$.
Thus,
$$
\prob{~ \frac{ h(k) }{ L \log(1/\epsilon) }
 \:\leq\: f(S(k)) \:\leq\: \big( 2K \log(1/\epsilon) \big) h(k) ~}
  ~\geq~ 1 - 2 \epsilon,
$$
completing the proof of Case 2.
\end{proofof}

\section{Implications of our Matroid Construction for Submodular Optimization}
\SectionName{implications}

The original motivation of our matroid construction in \Section{general-lower}
is to show hardness of learning in the PMAC model.
In this section we show that this construction has implications beyond learning theory;
it reveals interesting structure of matroids and submodular functions.
We illustrate this interesting structure by using it to show strong inapproximability results
for several submodular optimization problems.

\subsection{Submodular Minimization under a Cardinality Constraint}
\label{sec:opt}

Minimizing a submodular function is a fundamental problem in combinatorial optimization.
Formally, the problem is
\begin{equation}
\EquationName{SFM}
    \min \setst{ f(S) }{ S \subseteq [n] }.
\end{equation}
There exist efficient algorithms to solve this problem exactly \cite{GLS,IFF01,Lex00}.

\begin{theorem}
Let $f : 2^{[n]} \rightarrow \bR$ be any submodular function.
\begin{description}
\item[(a)] There is an algorithm with running time $\poly(n)$
that computes the minimum value of \eqref{eq:SFM}.
\item[(b)] There is an algorithm with running time $\poly(n)$
that constructs a lattice which represents all minimizers of \eqref{eq:SFM}.
This lattice can be represented in space $\poly(n)$.
\end{description}
\end{theorem}

The survey of McCormick \cite[Section 5.1]{McCormick} contains
further discussion about algorithms to construct the lattice of minimizers.
This lattice efficiently encodes a lot of information about the minimizers.
For example, given any set $S \subseteq [n]$, one can use the lattice
to efficiently determine whether $S$ is a minimizer of \eqref{eq:SFM}.
Also, the lattice can be used to efficiently find the inclusionwise-minimal and
inclusionwise-maximal minimizer of \eqref{eq:SFM}.
In summary, submodular function minimization is a very tractable optimization problem,
and its minimizers have a rich combinatorial structure.

The submodular function minimization problem becomes
much harder when we impose some simple constraints.
In this section we consider submodular function minimization under a cardinality constraint:
\begin{equation}
\EquationName{SFMCC}
    \min \setst{ f(S) }{ S \subseteq [n], |S| \geq d }.
\end{equation}
This problem, which was considered in previous work \cite{SF},
is a minimization variant of
submodular function maximization under a cardinality constraint \cite{NWF78},
and is a submodular analog of the minimum coverage problem \cite{Vinterbo}.
Unfortunately, \eqref{eq:SFMCC} is not a tractable optimization problem.
We show that, in a strong sense, its minimizers are very unstructured.


%

The main result of this section is that the minimizers of \eqref{eq:SFMCC} do not
have a succinct, approximate representation.

\begin{theorem}
\TheoremName{SFMCC}
There exists a randomly chosen non-negative, monotone, submodular function $f : 2^{[n]} \rightarrow
\bR$ such that,
for any algorithm that performs any number of queries to $f$
and outputs a data structure of size $\poly(n)$,
that data structure cannot represent the minimizers of \eqref{eq:SFMCC}
to within an approximation factor $o(n^{1/3}/\log n)$.
Moreover, any algorithm that performs $\poly(n)$ queries to $f$
cannot compute the minimum value of \eqref{eq:SFMCC} to within a $o(n^{1/3}/\log n)$ factor.
\end{theorem}

Here, a ``data structure representing the minimizers to within a factor $\alpha$''
is a program of size $\poly(n)$ that, given a set $S$, returns ``yes'' if $S$ is a
minimizer, returns ``no'' if $f(S)$ is at least $\alpha$ times larger than the minimum,
and otherwise can return anything.

Previous work \cite{bobby,SF,nick09} showed that
there exists a randomly chosen non-negative, monotone, submodular function $f : 2^{[n]} \rightarrow
\bR$ such that any algorithm that performs $\poly(n)$ queries to $f$
cannot approximate the minimum value of \eqref{eq:SFMCC} to within a $o(n^{1/2}/\log n)$ factor.
Also, implicit in the work of Jensen and Korte~\cite[pp.~186]{JK}
is the fact that no data structure of size $\poly(n)$ can
\textit{exactly} represent the minimizers of \eqref{eq:SFMCC}.
In contrast, \Theorem{SFMCC} is much stronger because it implies
that no data structure of size $\poly(n)$ can
even \textit{approximately} represent the minimizers of \eqref{eq:SFMCC}.


%




To prove \Theorem{SFMCC} we require the matroid construction of \Section{lb},
which we restate as follows.

\begin{theorem}
\TheoremName{newmatroids}
Let $n$ be a sufficiently large integer and let $h(n)$ be any slowly divergent function.
Define $k = n^{h(n)}+1$, $d = n^{1/3}$, $b = 8 \log k$ and $\tau = d/4 \log k$.

Set $U=\set{u_1,\ldots,u_k}$ and $V = \set{v_1,\ldots,v_n}$.
Suppose that $H=(U \union V,E)$ is a $(d,L,\epsilon)$-lossless expander.
We construct a family $\cA = \set{A_1,\ldots,A_k}$ of subsets of $[n]$, each of size $d$,
by setting
\begin{equation}
\EquationName{Aidef}
A_i ~=~ \setst{ j \in [n] }{ v_j \in \Gamma(\set{u_i}) } \qquad\forall i=1,\ldots,k.
\end{equation}
%
As before, $\Gamma(J)$ denotes the neighbors of the vertex set $J \subseteq U$.

For every $\cB \subseteq U$ there is a matroid
$\mat_\cB = ([n], \cI)$ whose rank function satisfies
\[
\rankf_{\mat_\cB}(A_i) ~=~
    \begin{cases}
    b          &\qquad\text{(if $u_i \in \cB$)} \\
    d          &\qquad\text{(if $u_i \in U \setminus \cB$)}.
    \end{cases}
\]
Furthermore, every set $S \subseteq [n]$ with $|S| \geq b$ has $\rankf_{\mat_\cB}(S) \geq b$.
\end{theorem}


\begin{proofof}{\Theorem{SFMCC}}
Pick a subset $\cB \subseteq U \setminus \set{u_k}$ randomly.
We now define a submodular function on the ground set $[n]$.
Set $L = d / 2 \log k$ and $\epsilon = 1/L$.
We apply \Theorem{ourexpanders} to obtain a random bipartite multigraph $H$.
With probability at least $1-2/k$, the resulting graph $H$ is a $(d,L,\epsilon)$-lossless expander
(after eliminating parallel edges).
In this case, we can apply \Theorem{newmatroids} to obtain the matroid $\mat_\cB$,
which we emphasize does not depend on $\Gamma(\set{u_k})$.
Define $A_i$ as in \eqref{eq:Aidef} for $i=1,\ldots,k-1$.

Now consider an algorithm $\ALG$ which performs any number of queries to $f$
and attempts to represent $\cB$ in $\poly(n)$ bits.
Since $\cB$ is a random subset of $U \setminus \set{u_k}$, which has cardinality $n^{h(n)}$,
the probability that $\cB$ can be represented in $\poly(n)$ bits is $o(1)$.
If $\cB$ cannot be exactly represented by $\ALG$ then,
with probability $1/2$, there is some set $A_i$ whose value is not correctly represented.
The multiplicative error in the value of $A_i$ is $d/b = o(n^{1/3}/\log n)$.

Next we will argue that any algorithm $\ALG$ performing $m = \poly(n)$ queries to
$f = \rankf_{\mat_\cB}$ has low probability of determining whether $\cB = \emptyset$.
If $\cB = \emptyset$ then the minimum value of \eqref{eq:SFMCC} is $d=n^{1/3}$,
whereas if $\cB \neq \emptyset$ then the minimum value of \eqref{eq:SFMCC} is $b=O(h(n) \log n)$.
Therefore this will establish the second part of the theorem.

Suppose the algorithm $\ALG$ queries
the value of $f$ on the sets $S_1,\ldots,S_m \subseteq [n]$.
Consider the $i\th$ query and suppose inductively that
$\rankf_{\mat_\cB}(S_j) = \rankf_{\mat_\emptyset}(S_j)$ for all $j<i$.
Thus $\ALG$ has not yet distinguished between the cases
$f=\rankf_{\mat_\cB}$ and $f=\rankf_{\mat_\emptyset}$.
Consequently the set $S_i$ used in the $i\th$ query is independent of $A_1,\ldots,A_{k-1}$.

Let $S_i'$ be a set of size $|S_i'|=d$ obtained from $S_i$ by
either adding (if $|S_i|<d$) or removing (if $|S_i| > d$) arbitrary elements of $[n]$,
or setting $S_i'=S_i$ if $|S_i|=d$.
We will apply \Theorem{ourexpanders} again, but this time we make an additional observation.
Since the definition of expansion does not depend on the labeling of the ground set,
one may assume in \Theorem{ourexpanders} that one vertex in $U$, say $u_k$, chooses its neighbors
deterministically and that all remaining vertices in $U$ choose their neighbors at random.
Specifically, we will set
$$
    \Gamma(\set{u_k}) ~=~ \setst{ v_j }{ j \in S_i' }.
$$
The neighbors $\Gamma(\set{u_i})$ for $i<j$ are not randomly rechosen;
they are chosen to be the same as they were in the first invocation of \Theorem{ourexpanders}.
With probability at least $1-2/k$ we again obtain a $(d,L,\epsilon)$-lossless expander,
in which case \Theorem{newmatroids} shows that $\rankf_{\mat_\cB}(S_i') = d = |S_i'|$.
That event implies
$$
    \rankf_{\mat_\cB}(S_i) ~=~ \begin{cases}
        |S_i|=\rankf_{\mat_\emptyset}(S_i)  &\quad\text{(if $|S_i|<d$)} \\
        d=\rankf_{\mat_\emptyset}(S_i)      &\quad\text{(if $|S_i| \geq d$)},
    \end{cases}
$$
and hence the inductive hypothesis holds for $i$ as well.

By a union bound over all $m$ queries, the probability of distinguishing whether
$B = \emptyset$ is at most $2m/k = o(1)$.
\end{proofof}

\subsection{Submodular $s$-$t$ Min Cut}

Let $G$ be an undirected graph with edge set $E$ and $n = \card{E}$.
Let $s$ and $t$ be distinct vertices of $G$.
A set $C \subseteq E$ is called an $s$-$t$ cut if every $s$-$t$ path intersects $C$.
Let $\cC \subset 2^E$ be the collection of all $s$-$t$ cuts.
The \newterm{submodular $s$-$t$ min cut} problem \cite{JB} is
\begin{equation}
\EquationName{SSTMC}
\min \setst{ f(C) }{ C \in \cC },
\end{equation}
where $f : 2^E \rightarrow \bR$ is a non-negative, monotone, submodular function.

\begin{theorem}[Jegelka and Bilmes~\cite{JB}]
Any algorithm for the submodular $s$-$t$ min cut problem
with approximation ratio $o(n^{1/3})$ must perform exponentially many queries to $f$.
\end{theorem}

Modifying their result to incorporate our matroid construction in \Section{general-lower},
we obtain the following theorem.

\begin{theorem}
\TheoremName{stmincut}
Let $d = n^{1/3}$.
Let $G$ be a graph with edge set $E$ consisting of $d$ internally-vertex-disjoint $s$-$t$ paths,
each of length exactly $n/d$.
Assume that $f : 2^E \rightarrow \bR$ is a non-negative, monotone, submodular function.
For any algorithm that performs any number of queries to $f$
and outputs a data structure of size $\poly(n)$,
that data structure cannot represent the minimizers of \eqref{eq:SSTMC}
to within an approximation factor $o(n^{1/3}/\log n)$.
Moreover, any algorithm that performs $\poly(n)$ queries to $f$
cannot compute the minimum value of \eqref{eq:SSTMC} to within a $o(n^{1/3}/\log n)$ factor.
\end{theorem}

The proof of this theorem is almost identical to the proof of \Theorem{SFMCC}.
All that we require is a slightly different expander construction.

\begin{theorem}
\TheoremName{modifiedexpanders}
Let $U=\set{u_1,\ldots,u_k}$ and $V$ be disjoint vertex sets, where $|V|=n$
and $n$ is a multiple of $d$.
Write $V$ as the disjoint union $V = V_1 \cup \cdots \cup V_{d}$ where each $|V_i| = n/d$.

Generate a random bipartite multigraph $H$ with left-vertices $U$ and right-vertices $V$ as follows.
The vertex $u_k$ has exactly $d$ neighbors in $V$, chosen deterministically and arbitrarily.
For each vertex $u_\ell$ with $\ell \leq k-1$, pick exactly one neighbor from each $V_i$,
uniformly and independently at random.
So each vertex in $U$ has degree exactly $d$.

Suppose that $k \geq 4$, $L \geq d$, $d \geq \log(k)/\epsilon$ and $n \geq 22 L d/\epsilon$.
Then, with probability at least $1-2/k$, the multigraph $H$ has no parallel edges and satisfies
\begin{align*}
|\neigh(\set{u})| &~=~ d \qquad\forall u \in U \\
|\neigh(J)| &~\geq~ (1-\epsilon) \cdot d \cdot |J| \qquad\forall J \subseteq U,\, |J| \leq L.
\end{align*}
\end{theorem}

\begin{proof}
The proof is nearly identical to the proof of \Theorem{ourexpanders} in \Appendix{expander}.
The only difference is in analyzing the probability of a repeat when sampling the neighbors
of a set $J \subseteq U$ with $|J|=j$.
First consider the case that $u_k \in J$.
When sampling the neighbors $\Gamma(J)$, an element $v_i$ is considered a repeat
if $v_i \in \set{ v_1, \ldots, v_{i-1} }$ or if $v_i \in \Gamma(\set{u_k})$.
Conditioned on $v_1, \ldots, v_{i-1}$, the probability of a repeat is
at most $\frac{j+d}{n/d}$.
If $u_k \not\in J$ then this probability is at most $jd/n$.
Consequently, the probability of having more than $\epsilon j d$ repeats is at most
$$
\binom{j d}{\epsilon j d} \Big( \frac{(j+d) d}{n} \Big)^{\epsilon j d}
    ~\leq~ \Big( \frac{e}{\epsilon} \Big)^{\epsilon j d}
           \Big( \frac{(j+d) d}{n} \Big)^{\epsilon j d}
    ~\leq~ (1/4)^{\epsilon j d}.
$$
The last inequality follows from $j+d \leq 2L$ and our hypothesis $n \geq 22 L d / \epsilon$.
The remainder of the proof is identical to the proof of \Theorem{ourexpanders}.
\end{proof}

\begin{proofsketchof}{\Theorem{stmincut}}
Let $V_i$ be the edges of the $i\th$ $s$-$t$ path.
The minimal $s$-$t$ cuts are those which choose exactly one edge from each $s$-$t$ path;
in other words, they are the transversals of the $V_i$'s.
Let $V = V_1 \cup \cdots \cup V_d$; this is also the edge set of the graph $G$.

As in \Theorem{SFMCC} we apply \Theorem{modifiedexpanders} and \Theorem{newmatroids}
to obtain a matroid $\mat_\cB$.
Because the expander construction of \Theorem{modifiedexpanders} ensures that each vertex $u_\ell$
has exactly one neighbor in each $V_i$, the corresponding set $A_\ell$ is a minimal $s$-$t$ cut.

Suppose $\ALG$ performs any number of queries to $f = \rankf_{\mat_\cB}$. The set $\cB$ has low
probability of being representable in $\poly(n)$ bits, in which case
there is an $s$-$t$ min cut $A_i$ whose value is not correctly represented with probability $1/2$.
The multiplicative error in the value of $A_i$ is $d/b = o(n^{1/3}/\log n)$.
This proves the first part of the theorem.

Similarly, any algorithm $\ALG$ performing $m = \poly(n)$ queries to
$f$ has low probability of determining whether $\cB = \emptyset$.
If $\cB = \emptyset$ then the minimum value of \eqref{eq:SSTMC} is $d=n^{1/3}$,
whereas if $\cB \neq \emptyset$ then the minimum value of \eqref{eq:SSTMC} is $b=O(h(n) \log n)$.
This proves the second part of the theorem.
\end{proofsketchof}

\subsection{Submodular Vertex Cover}

Let $G=(V,E)$ be a graph with $n = \card{V}$.
A set $C \subseteq V$ is a vertex cover if every edge has at least one endpoint in $C$.
Let $\cC \subset 2^V$ be the collection of vertex covers in the graph.
The \newterm{submodular vertex cover} problem \cite{GKTW09,IN09} is
\begin{equation}
\EquationName{SVC}
\min \setst{ f(S) }{ S \in \cC },
\end{equation}
where $f : 2^V \rightarrow \bR$ is a non-negative, submodular function.
An algorithm for this problem is said to have \newterm{approximation ratio} $\alpha$ if,
for any function $f$, it returns a set $S$ for which
$f(S) \leq \alpha \cdot \min \setst{ f(S) }{ S \in \cC }$.

\begin{theorem}[Goel et al.~\cite{GKTW09}, Iwata and Nagano \cite{IN09}]
There is an algorithm which performs $\poly(n)$ queries to $f$ and has approximation ratio $2$.
\end{theorem}

Goel et al.\ only state that their algorithm is applicable for monotone, submodular functions,
but the monotonicity restriction seems to be unnecessary.

\begin{theorem}[Goel et al.~\cite{GKTW09}]
For any constant $\epsilon>0$, any algorithm for the submodular vertex cover problem
with approximation ratio $2-\epsilon$ must perform exponentially many queries to $f$.
\end{theorem}

Modifying their result to incorporate our matroid construction in \Section{general-lower},
we obtain the following theorem.

\begin{theorem}
Let $G = (U \cup V, E)$ be a bipartite graph.
Assume that $f : 2^{U \cup V} \rightarrow \bR$ is a non-negative, monotone, submodular function.
Let $\epsilon \in (0,1/3)$ be a constant.
For any algorithm that performs any number of queries to $f$
and outputs a data structure of size $\poly(n)$,
that data structure cannot represent the minimizers of \eqref{eq:SVC}
to within an approximation factor better than $4/3-\epsilon$.
Moreover, any algorithm that performs $\poly(n)$ queries to $f$
cannot compute the minimum value of \eqref{eq:SSTMC} to within a $4/3-\epsilon$ factor.
\end{theorem}

\begin{proofsketch}
Let $G$ be a graph such that $|U|=|V|=|E|=n/2$,
and where the edges in $E$ form a matching between $U$ and $V$.
The minimal vertex covers are those that contain exactly one endpoint of each edge in $E$.
Set $k = 2^{\epsilon^2 n / 40}$.
Let $\cA = \set{A_1,\cdots,A_k}$ be a collection of independently and uniformly chosen minimal
vertex covers.  For any $i \neq j$, $\expect{ | A_i \cap A_j | } = n/4$
and a Chernoff bound shows that
$\prob{ | A_i \cap A_j | > (1+\epsilon) n/4 } \leq \exp( - \epsilon^2 n/12 )$.
A union bound shows that, with high probability,
$|A_i \cap A_j| \leq (1+\epsilon) n/4$ for all $i \neq j$.

We now apply \Lemma{pairwise} (in \Appendix{pairwise})
with each $b_i = b = (3+\epsilon)n/8$ and $d = n/2$.
We have
$$
\min_{i,j \in [k]} (b_i + b_j - \card{A_i \intersect A_j})
~\geq~ 2b - (1+\epsilon)n/4
~=~ 2(3+\epsilon)n/8 - (1+\epsilon)n/4
~=~ n/2,
$$
and therefore the hypotheses of \Lemma{pairwise} are satisfied.
It follows that, for any set $\cB \subseteq \cA$ the set
$$
\cI_\cB ~=~
    \setst{ I }{\: \card{I} \leq d \:\And\: \card{I \intersect A_j} \leq b ~~~\forall A_j \in \cB }
$$
is the family of independent sets of a matroid.
Let $f = \rankf_{\mat_\cB}$ be the rank function of this matroid.

Suppose $\ALG$ performs any number of queries to $f$. The set $\cB$ has low
probability of being representable in $\poly(n)$ bits, in which case
there is a minimal vertex cover $A_i$ whose value is not correctly represented with probability $1/2$.
The multiplicative error in the value of $A_i$ is
$$
\frac{d}{b} ~=~ \frac{n/2}{(3+\epsilon)n/8} ~>~ \frac{4}{3} - \epsilon.
$$
This proves the first part of the theorem.

Similarly, any algorithm $\ALG$ performing $m = \poly(n)$ queries to
$f$ has low probability of determining whether $\cB = \emptyset$.
If $\cB = \emptyset$ then the minimum value of \eqref{eq:SSTMC} is $d$,
whereas if $\cB \neq \emptyset$ then the minimum value of \eqref{eq:SSTMC} is $b$.
The multiplicative error is at least $d/b$, proving the second part of the theorem.
%
\end{proofsketch}

\section{Implications to Algorithmic Game Theory and Economics}
\label{sec:econ}

 An important consequence of our matroid construction in Section~\ref{section-new-extremal-matroid}
 is that matroid rank functions  do not have a ``sketch'',
i.e., a concise, approximate representation.
As matroid rank functions can be shown to satisfy the ``gross substitutes'' property~\cite{kazuo-book}, our work implies that gross substitute functions do not have a
concise, approximate representation. This provides a surprising answer to an 
open
question in economics~\cite{BingLM04PresentationaSoSV,liad-thesis,BN05}.
In this section we define gross substitutes functions, briefly describe their importance in economics,
and formally state the implications of our results for these functions.

Gross substitutes functions play an important role in algorithmic game theory and economics,
particularly through their use as valuation functions in combinatorial
auctions~\cite{CramtonSS06CombinatorialA,GulS99WalrasianEwGS,book07}.
Intuitively, in a gross substitutes valuation, increasing the price of certain items can
not reduce the demand for items whose price has not changed. Formally:


\begin{definition}
\label{def:GS}
For price vector $\vec{p} \in\reals^{n}$, the \newterm{demand correspondence}
$\demandSet{}_{f}(\vec{p})$ of valuation $f$ is the collection of preferred sets
at prices $\vec{p}$, i.e.,
$$
 \demandSet{}_{f}(\vec{p}) ~=~ \argmax_{S \subseteq \{1, \ldots, n\}}
 \set{ f(S) - \smallsum{j \in S}{}\, p_{j} }.
$$
A function $f$ is \newterm{gross substitutes} (GS) if for any price vector $\vec{q} \geq \vec{p}$
(i.e., for which $q_{i} \geq p_{i} ~\:\forall i \!\in\! [n]$), and any $A \in \demandSet{}_{f}(\vec{p})$ there exists $A' \in \demandSet{}_{f}(\vec{q})$ with $ A' \supseteq \{i\in A: p_{i} = q_{i} \}$.
\end{definition}

In other words, the gross substitutes property requires that all items $i$
in some preferred set $A$ at the old prices $\vec{p}$ and
for which the old and new prices are equal ($p_i = q_i$)
are simultaneously contained in some preferred set $A'$ at the new prices $\vec{q}$.

Gross substitutes valuations (introduced by Kelso and Crawford~\cite{KelsoC82JobMCFaGS})
enjoy several appealing structural properties whose implications been extensively studied by many
researchers~\cite{BingLM04PresentationaSoSV}.
For example, given bidders with gross substitutes valuations, simple item-price ascending auctions
can be used for determining the socially-efficient allocation.
As another example, the gross substitute condition is actually necessary for important economic
conclusions. For example, Gul and Stacchetti~\cite{GulS99WalrasianEwGS} and Milgrom~\cite{milg2000}
showed that given any valuation that is not gross substitutes, one can specify very simple
valuations for the other agents to create an economy in which no Walrasian equilibrium exists.

One important unsolved question
concerns the complexity of describing gross substitutes valuations.
Several researchers have asked whether there exist a ``succinct'' representation for such
valuations \cite{BingLM04PresentationaSoSV} \cite[Section 6.2.1]{liad-thesis}
\cite[Section 2.2]{BN05}.
In other words, can a bidder disclose the exact details of his valuation without conveying an
exceptionally large amount of information?
An implications of our work is that the answer to this question is ``no'',
in a very strong sense. Our work implies that
gross substitutes functions cannot be represented succinctly, even approximately,
and even with a large approximation factor.
Formally:

\begin{definition}
We say that $g:2^{[n]} \rightarrow \reals_+$ is an $\alpha$-sketch for $f:2^{[n]} \rightarrow \reals_+$ if $g$ can be represented in $\poly(n)$ space and for every set $S$ we have
that $f(S)/{\alpha} \leq g(S) \leq f(S)$.
\end{definition}

As matroid rank functions are known to satisfy the gross substitute property~\cite{kazuo-book}, our
work implies that  gross substitutes do not have a concise, approximate representation.
Specifically:

\begin{theorem}
\label{lowerboundGS}
Gross substitute functions do not admit  $o({n^{1/3}}/{\log n})$ sketches.
\end{theorem}


\section{Conclusions}

In this work we have used a learning theory perspective to uncover new structural properties of
submodular functions.
We have presented the first algorithms and lower bounds for learning submodular functions
in a distributional learning setting.
We also presented numerous implications of our work in algorithmic game theory,
economics, matroid theory and combinatorial optimization.

Regarding learnability, we presented polynomial upper and lower bounds on the approximation
factor achievable when using only a polynomial number of examples drawn
i.i.d.\ from an arbitrary distribution.
We also presented a simple algorithm achieving a constant-factor approximation under product
distributions.
These results show that, with respect to product distributions,
submodular functions behave in a fairly simple manner, whereas with
respect to general distributions, submodular functions
behave in a much more complex manner.

We constructed a new family of matroids with interesting technical properties
in order to prove our lower bound on PMAC-learnability.
The existence of these matroids also resolves an open question in economics:
an immediate corollary of our construction is that gross substitutes functions have no succinct,
approximate representation.
We also used these matroids to show that the optimal solutions of various submodular optimization
problems can have a very complicated structure.

The PMAC model provides a new approach for analyzing the learnability of real-valued functions.
This paper has analyzed submodular functions in the PMAC model.
We believe that it will be interesting to study PMAC-learnability of other classes of real-valued
functions.
Indeed, as discussed below, subsequent work has already studied subadditive and XOS
functions in the PMAC model.

One technical question left open by this work is determining the precise approximation factor
achievable for PMAC-learning submodular functions --- there is a gap between the
$O(n^{1/2})$ upper bound in \Theorem{n-monotone} and the $\tOmega(n^{1/3})$ lower bound
in~\Theorem{mainlb}. We suspect that the lower bound can be improved to
$\tilde{\Omega}(n^{1/2})$. If such an improved lower bound is possible, the matroids or
submodular functions used in its proof are likely to be very interesting.

\subsection{Subsequent Work}
\SectionName{subsequent}
Following our work, several authors have provided further results for learning submodular functions in a distributional learning setting.

 Balcan et al.~\cite{BCIW} and Badanidiyuru et al.\ \cite{BDFKNR} have provided further learnability results in the PMAC model for various classes of set functions commonly used in algorithmic game theory and economics. Building on our algorithmic technique, Balcan et al.~\cite{BCIW} give a computationally efficient algorithm for PMAC-learning subadditive functions to within a $\tO(\sqrt{n})$ factor. They also provide new target-dependent learnability result for \textsc{XOS} (or fractionally subadditive) functions. Their algorithms use the algorithmic technique that we develop in~\Section{general-upper}, together with new structural results for these classes of functions.
Badanidiyuru et al.\ \cite{BDFKNR} consider the problem of \newterm{sketching}
subadditive and submodular functions.
They show that the existence of such a sketch implies that PMAC-learning
to within a factor $\alpha$ is possible if computational efficiency is ignored.
As a consequence they obtain (computationally inefficient) algorithms for PMAC-learning
to within a $\tO(\sqrt{n})$ factor for subadditive functions,
and to within a $1+\epsilon$ factor for both coverage functions and OXS functions.

Regarding inapproximability, both Badanidiyuru et al.\ and Balcan et al.\
show that XOS (i.e., fractionally subadditive) functions do not have
sketches that approximate to within a factor $\tilde{o}(\sqrt{n})$.
Consequently, every algorithm for PMAC-learning XOS functions
must have approximation factor $\tOmega(\sqrt{n})$.
The construction used to prove this result is significantly simpler than our construction in
\Section{general-lower}, because XOS functions are a more expressive class than submodular
functions.

Motivated by problems in privacy preserving data analysis,
Gupta et al.\ \cite{GHRU}
considered how to perform statistical queries to a data set
in order to learn the answers to all statistical queries from a certain class.
They showed that this problem can be efficiently solved
when the queries are described by a submodular function.
One of the technical pieces in their work is an algorithm to learn submodular functions under a
product distribution. A main building block of their technique is the algorithm we provide in~\Section{special}
for learning under a product distribution, and their analysis is inspired by ours.
Their formal guarantee is incomparable to ours:
it is stronger in that they allow non-Lipschitz and non-monotone functions,
but it is weaker in that they require access to the submodular function via a value oracle,
and they guarantee only additive error (assuming the function is appropriately normalized).
Moreover, their running time is $n^{\poly{(1/\epsilon)}}$ whereas ours is  $\poly{(n,1/\epsilon)}$.

Cheraghchi et al.\ \cite{CKKL} study the noise stability of submodular functions.
As a consequence they obtain an algorithm for learning a submodular function
under product distributions.
Their algorithm also works for non-submodular and non-Lipschitz functions,
and only requires access to the submodular function via statistical queries, though the running time is $n^{\poly{(1/\epsilon)}}$.
Their algorithm is agnostic (meaning that they do not assume the target function is submodular),
and their performance guarantee proves that the $L_1$-loss of their hypothesis is
at most $\epsilon$ more than the best error achieved by any submodular function
(assuming the function is appropriately normalized).

Raskhodnikova and Yaroslavtsev~\cite{rask12} consider learnability of integer-valued, submodular functions and
 prove that any submodular function $f: \{0,1\}^n \rightarrow \{0,1,\ldots,k\}$ can be represented as a pseudo-Boolean $2k$-DNF formula.
They use this to provide an algorithm for learning such functions
using membership queries under the uniform distribution;
the algorithm runs in time polynomial in $\poly(n, k^{O(k \log k / \epsilon)}, 1/\epsilon)$.

\section*{Acknowledgments}
We thank Jan Vondr\'ak for simplifying our original proof of \Theorem{mainthm},
 Atri Rudra for explaining how our original proof of \Theorem{manymatroids}
was connected to expander graphs, and Florin Constantin for discussions about gross substitutes. 
We also thank Avrim Blum, Shahar Dobzinski, Steve Hanneke, Satoru Iwata, Lap Chi Lau, Noam Nisan, Alex Samorodnitsky, Mohit Singh, Santosh Vempala,
and Van Vu for helpful discussions.

\smallskip

This work was
supported in part by NSF grants CCF-0953192 and CCF-1101215, AFOSR grant FA9550-09-1-0538, a NSERC
Discovery Grant, and a Microsoft Research Faculty Fellowship.



\appendix

\section{Standard Facts}

\subsection{Submodular Functions}
\AppendixName{facts-submodular}

\begin{theorem}
\TheoremName{coverage-subm}
Given a finite universe $U$, let $S_1, S_2, \ldots, S_n$ be
subsets of $U$. Define $f : 2^{[n]} \rightarrow \bR_+$ by
$$f(A) =\left|\cup_{i\in A}
S_i\right|~~~\text{for}~~~~A \subseteq [n].$$
Then $f$ is monotone and submodular.
More generally, for any non-negative weight function $w : U \rightarrow \bR_+$,
the function $f$ defined by
$$f(A) =w\left(\cup_{i\in A} S_i\right)~~~\text{for}~~~~A \subseteq [n]$$
is monotone and submodular.
\end{theorem}

\begin{lemma}
\label{minim-subm} The minimizers of any submodular function are closed under union and
intersection.
\end{lemma}
\begin{proof}Assume that $J_1$ and $J_2$ are minimizers for $f$. By submodularity we have
$$f(J_1)+f(J_2) \geq f(J_1 \cap J_2) + f(J_1 \cup J_2).$$ We also have $$f(J_1 \cap J_2) + f(J_1
\cup J_2)  \geq f(J_1)+f(J_2),$$ so $f(J_1)=f(J_2)= f(J_1 \cap J_2) = f(J_1 \cup J_2)$,
as desired.
\end{proof}


\subsection{Sample Complexity Results}
\label{useful-lemmas}
We state here several known sample complexity bounds that were used for proving the results in \Section{general-upper}. 
See, e.g., \cite{DGL:book96,AB99}.

\begin{theorem}
\TheoremName{VCbound}
Let $\hclass$ be a set of functions from $\X$
to $\{-1,1\}$ with finite VC-dimension $\vcdim \geq 1$. Let
$D$ be an arbitrary, but fixed probability distribution
over $\X$ and let $c^*$ be an arbitrary target function. For any $\epsilon$, $\delta>0$, if we
draw a sample $\trainS$ from $D$ of size $$m (\epsilon, \delta, \vcdim)=\frac{1}{\epsilon} \left(4\vcdim \log {\left(
\frac{1}{\epsilon} \right)}+ 2 \log {\left( \frac{2}{\delta}
\right)} \right),$$ then with probability $1-\delta,$ all
hypotheses with error $\geq \epsilon$ are inconsistent with the
data; i.e., uniformly for all $h \in C$ with $\err(h) \geq \epsilon$, we have $\herr(h)>0$.
Here $\err(h)=\Pr_{x \sim D}{[h(x) \neq c^*(x)]}$ is the true error of $h$ and
 $\herr(h)=\Pr_{x \sim \trainS}{[h(x) \neq c^*(x)]}$ is the empirical error of $h$.
\end{theorem}

\smallskip

\begin{theorem}
\TheoremName{vc-normal}
Suppose that $\hclass$ is a set of functions from
$\X$ to $\{-1,1\}$ with finite VC-dimension $\vcdim \geq 1$. For any
distribution $D$ over $\X$, any target function (not necessarily in
$\hclass$), and any $\epsilon$, $\delta>0$, if we draw a
sample from $D$  of size $$m (\epsilon, \delta, \vcdim) =
\frac{64}{\epsilon^2} \left(2\vcdim \ln {\left( \frac{12}{\epsilon}
\right)}+ \ln {\left( \frac{4}{\delta} \right)} \right),$$ then with
probability at least $1-\delta$, we have $\left|\err(h) -
\herr(h)\right| \leq \epsilon$ for all
$\hyp \in \hclass$.
\end{theorem}


\section{Proofs for Concentration of Submodular Functions}
\AppendixName{appendix-special}

\subsection{Proof of~\Theorem{OURTALAGRAND}}

\repeatclaim{\Theorem{OURTALAGRAND}}{\OURTALAGRAND}

\begin{proof}
We begin by observing that the theorem is much easier to prove
in the special case\footnote
    {An initial draft of our paper proved only this easier case.
    After learning of the similar concentration inequality by Chekuri et al.~\cite{CVZ10},
    we extended our proof to handle functions $f$ that are not integer-valued.}
that $f$ is integer-valued.
Together with our other hypotheses on $f$,
this implies that $f$ must actually be a matroid rank function.
Whenever $f(S)$ is large, this fact can ``certified'' by any maximal independent subset of $S$.
The theorem then follows easily from a version of Talagrand's inequality
which leverages this certification property;
see, e.g., \cite[\S 7.7]{AlonSpencer} or \cite[\S 10.1]{MolloyReed}.

We now prove the theorem in its full generality.
We may assume that $t \leq \sqrt{b}$, otherwise the theorem is trivial, since $f(X)$ is non-negative.
%
Talagrand's inequality states: for any $\cA \subseteq \set{0,1}^n$ and $y \in \set{0,1}^n$
drawn from a product distribution,
\begin{equation}
\EquationName{talagrand}
\prob{ y \in \cA } \cdot \prob{ \rho(\cA,y) > t } ~\leq~ \exp(-t^2/4),
\end{equation}
where $\rho$ is a distance function defined by
$$
\rho(\cA,y)
    ~=~ \sup_{\substack{\alpha \in \bR^n \\ \norm{\alpha}_2=1}}
        \min_{z \in \cA} \sum_{i \::\: y_i \neq z_i} \alpha_i.
$$
We will apply this inequality to the set
$\cA \subseteq 2^V$ defined by $\cA = \setst{ X }{ f(X) < b - t \sqrt{b} }$.

\begin{claim}
For every $Y \subseteq V$, $f(Y) \geq b$ implies $\rho( \cA, Y ) > t$.
\end{claim}
\begin{subproof}
Suppose to the contrary that $\rho( \cA, Y ) \leq t$.
By relabeling, we can write $Y$ as $Y = \set{1,\ldots,k}$.
For $i \in \set{0,\ldots,k}$, let $E_i = \set{1,\ldots,i}$.
Define
$$
    \alpha_i ~=~
    \begin{cases}
    f( E_{i} ) - f( E_{i-1} ) &\qquad\text{(if $i \in Y$)} \\
    0 &\qquad\text{(otherwise).}
    \end{cases}
$$
Since $f$ is monotone and $1$-Lipschitz, we have $0 \leq \alpha_i \leq 1$.
Thus $\norm{\alpha}_2 \leq \sqrt{ \sum_i \alpha_i } \leq \sqrt{f(Y)}$,
by non-negativity of $f$.

The definition of $\rho$ and our supposition $\rho( \cA, Y ) \leq t$ imply that
there exists $Z \in \cA$ with
\begin{equation}
\EquationName{foundZ}
    \sum_{i \in (Y \setminus Z) \union (Z \setminus Y)} \alpha_i
    ~\leq~ \rho(\cA,Y) \cdot \norm{\alpha}_2
    ~\leq~ t \sqrt{f(Y)}.
\end{equation}
We may assume that $Z \subset Y$, since $Z \intersect Y$ also satisfies the desired conditions.
This follows since monotonicity of $f$ implies that $\alpha \geq 0$ and
that $\cA$ is downwards-closed.

We will obtain a contradiction by showing that $f(Y)-f(Z) \leq t \sqrt{f(Y)}$.
First let us order $Y \setminus Z$ as $(\phi(1), \ldots, \phi(m) )$,
where $\phi(i) < \phi(j)$ iff $i < j$.
Next, define $F_i = Z \union \set{ \phi(1), \ldots, \phi(i) } \subseteq Y$.
Note that $E_j \subseteq F_{\phi^{-1}(j)}$;
this follows from our choice of $\phi$,
since $Z \subseteq F_{\phi^{-1}(j)}$ but we might have $Z \not\subseteq E_j$.
Therefore
\begin{align*}
f(Y) - f(Z)
	&~=~ \sum_{i=1}^m  \big( f( F_i ) - f(F_{i-1}) \big) \\
	&~=~ \sum_{j \in Y \setminus Z}  \big( f( F_{\phi^{-1}(j)} ) - f( F_{\phi^{-1}(j)-1} ) \big) \\
	&~\leq~ \sum_{j \in Y \setminus Z}  \big( f( E_{j} ) - f( E_{j-1} ) \big)
        \qquad\text{(since $E_j \subseteq F_{\phi^{-1}(j)}$ and $f$ is submodular)} \\
	&~=~ \sum_{j \in Y \setminus Z} \alpha_j \\
	&~\leq~ t \sqrt{f(Y)} \qquad\text{(by \Equation{foundZ})}.
\end{align*}
So $f(Z) \geq f(Y) - t \sqrt{f(Y)} \geq b - t \sqrt{b}$,
since $f(Y) \geq b$ and $t \leq \sqrt{b}$.
This contradicts $Z \in \cA$.
\end{subproof}

This claim implies $\prob{ f(Y) \geq b } \leq \prob{ \rho(\cA,Y) > t }$,
so the theorem follows from \Equation{talagrand}.
\end{proof}

\subsection{Proof of \Corollary{rank-concentr}
}

\repeatclaim{\Corollary{rank-concentr}}{\rankconcentr}

\begin{proof}
Let $Y=f(X)$ and let $M$ be a median of $Y$.
The idea of the proof is simple:
\Theorem{OURTALAGRAND} shows tight concentration of $Y$ around $M$.
Since $Y$ is so tightly concentrated, we must have $\expect{Y} \approx M$.
This allows us to show tight concentration around $\expect{Y}$.
The remainder of the proof is simply a matter of detailed calculations.
Similar arguments can be found in \cite[\S 2.5]{Janson} and \cite[\S 20.2]{MolloyReed}.

\begin{claim}
\ClaimName{convenientTalagrandBound}
$$
\prob{ \abs{Y-M} \geq \lambda } ~\leq~
\begin{cases}
4 e^{-\lambda^2 / 8 M} &\qquad(0 \leq \lambda \leq M) \\
2 e^{-\lambda / 8} &\qquad(\lambda \geq M). \\
\end{cases}
$$
Also,
$$
\prob{ \abs{Y-M} \geq \lambda } ~\leq~
4 e^{-\lambda^2 / 24 M} \qquad(0 \leq \lambda \leq 5M).
$$
\end{claim}
\begin{subproof}
First, apply \Theorem{OURTALAGRAND} with $b = M$ and $t = \lambda / \sqrt{M}$.
Since $\prob{Y \geq b} \leq 1/2$, we get
\begin{equation}
\EquationName{convenientTalagrandMedian1}
\prob{ Y \leq M - \lambda } ~\leq~ 2 \exp( -t^2 / 4 M ).
\end{equation}
Next, apply \Theorem{OURTALAGRAND} with $b = M + \lambda$ and $t = \lambda / \sqrt{M+\lambda}$.
Since $\prob{Y \leq b - t \sqrt{b}} = \prob{Y \leq M} \leq 1/2$, we get
\begin{equation}
\EquationName{convenientTalagrandMedian2}
\prob{ Y \geq M + \lambda } ~\leq~ 2 \exp\big( -t^2 / 4 (M+\lambda) \big).
\end{equation}
Combining \eqref{eq:convenientTalagrandMedian1} and \eqref{eq:convenientTalagrandMedian2}
proves the claim.
\end{subproof}

\begin{claim}
\ClaimName{expectMedGap}
$\abs{ \expect{Y} - M } \leq 15 \sqrt{ \expect{Y} } + 16$.
Consequently, if $\expect{Y} \geq 256$ then
$\abs{ \expect{Y} - M } \leq 16 \sqrt{ \expect{Y} }$.
\end{claim}
\begin{subproof}
This is a standard calculation; see, e.g., \cite[\S 2.5]{Janson}.
Using \Claim{convenientTalagrandBound},
\begin{align*}
\abs{ \expect{Y} - M }
 &~\leq~ \expect{ \abs{Y - M} } \\
 &~=~ \int_{0}^\infty \prob{ \abs{ Y - M } \geq \lambda } \,d\lambda \\
 &~=~ \int_{0}^M 4 e^{-\lambda^2/8M} \,d\lambda
  ~+~ \int_{M}^\infty 2 e^{-\lambda/8} \,d\lambda \\
 &~\leq~ 4 \sqrt{2 \pi M} + 16 e^{-M/8}.
\end{align*}
Since $Y \geq 0$ we have $0 \leq M \leq 2 \expect{Y}$ (by Markov's inequality), so
$$
\abs{ \expect{Y} - M } ~\leq~ 15 \sqrt{ \expect{Y} } + 16.
$$
This quantity is at most $16 \sqrt{ \expect{Y} }$ if $\expect{Y} \geq 256$.
\end{subproof}

\noindent\textit{Case 1:}
$\expect{Y} \geq 584 / \alpha^2$.
Then
\begin{align}
\nonumber
&\prob{ \abs{Y-\expect{Y}} \geq (\sqrt{2}t+16) \sqrt{\expect{Y}} } \\\nonumber
~\leq~&
    \prob{ \abs{Y-M} \geq (\sqrt{2}t+16) \sqrt{\expect{Y}} - \abs{\expect{Y}-M} } \\\nonumber
~\leq~&
    \prob{ \abs{Y-M} \geq t \sqrt{2 \expect{Y}} }
    \qquad\text{(by \Claim{expectMedGap})}\\\nonumber
~\leq~&
    \prob{ \abs{Y-M} \geq t \sqrt{M} }
    \qquad\text{(since $\expect{Y} \geq M/2$)} \\\EquationName{TalagrandExpectation}
~\leq~&
    4 \exp( - t^2 / 24 ) \qquad\text{(if $t \leq 5 \sqrt{M}$)},
\end{align}
by \Claim{convenientTalagrandBound}.
Set $t=(\alpha \sqrt{\expect{Y}} - 16)/\sqrt{2}$.
One may check that
\begin{equation}
\EquationName{rootebound}
\frac{t^2}{24}
    ~=~ \frac{(\alpha \sqrt{\expect{Y}} - 16)^2}{2 \cdot 24}
    ~\geq~ \frac{\alpha^2 \expect{Y}}{422},
\end{equation}
since we assume $\expect{Y} \geq 584/\alpha^2$.
Furthermore,
\begin{align*}
\sqrt{\expect{Y}}
    &~\leq~ \sqrt{M} + \sqrt{\abs{\expect{Y}-M}} \\
    &~\leq~ \sqrt{M} + 4 \expect{Y}^{1/4} \qquad\text{(by \Claim{expectMedGap})} \\
    &~\leq~ \sqrt{M} + 0.82 \sqrt{\expect{Y}},
\end{align*}
since $\expect{Y}^{1/4} \geq 584^{1/4} > 4/0.82$.
Rearranging, $0.18 \cdot \sqrt{\expect{Y}} \leq \sqrt{M}$.
Therefore we have
$$
t   ~<~ \sqrt{\expect{Y}}/\sqrt{2}
    ~<~ 4 \cdot 0.18 \cdot \sqrt{\expect{Y}}
    ~\leq~ 4 \sqrt{M},
$$
so we may apply \eqref{eq:TalagrandExpectation} with this value of $t$.
\begin{align*}
\prob{ \abs{Y-\expect{Y}} \geq \alpha \expect{Y} }
    &~=~ \prob{ \abs{Y-\expect{Y}} \geq (\sqrt{2}t+16) \sqrt{\expect{Y}} } \\
    &~\leq~ 4 \exp( - t^2 / 24 ) \\
    &~\leq~ 4 \exp( - \alpha^2 \expect{Y} / 422 ),
\end{align*}
by \eqref{eq:rootebound}.

\vspace{3pt}

\noindent\textit{Case 2:}
$\expect{Y} < 584 / \alpha^2$.
Then $\alpha^2 \expect{Y} / 422 < \ln(4)$, so
$$
4 \exp(-\alpha^2 \expect{Y} / 422 ) ~>~ 1,
$$
and the claimed inequality is trivial.
\end{proof}

\section{Additional Proofs for Learning Submodular Functions}

\subsection{Learning Boolean Submodular Functions}
\AppendixName{boolean}

\begin{theorem}
\PropositionName{submodbinary}
The class of monotone, Boolean-valued, submodular functions is efficiently PMAC-learnable
with approximation factor $1$.
\end{theorem}
\begin{proof}
Let $f: 2^{[n]} \rightarrow \{0,1\}$ be an arbitrary monotone, boolean, submodular function.
We claim that $f$ is either constant or a monotone disjunction.
If $f(\emptyset)=1$ then this is trivial, so assume $f(\emptyset)=0$.


Since submodularity is equivalent to the property of decreasing marginal values,
and since $f(\emptyset)=0$, we get
$$
    f(T \cup \set{x}) - f(T) \leq  f( \set{x})
    \qquad\forall T \subseteq [n], x \in [n] \setminus T.
$$
If $f(\set{x})=0$ then this together with monotonicity implies that
$f(T \cup \set{x}) = f(T)$ for all $T$.
On the other hand, if $f(\set{x})=1$ then monotonicity implies that $f(T) = 1$ for all $T$ such that
$x \in T$.
Thus we have argued that $f$ is a disjunction:
$$
    f(S) ~=~ \begin{cases}
        1 &\text{(if $S \intersect X \neq \emptyset$)} \\
        0 &\text{(otherwise)}
    \end{cases},
$$
where $X = \setst{ x }{ f(\set{x}) = 1 }$.
This proves the claim.

It is well known that the class of disjunctions is easy to learn
in the supervised learning setting~\cite{KV:book94,Vapnik:book98}.
\end{proof}

Non-monotone, Boolean, submodular functions need not be disjunctions.
For example, consider the function
$f$ where $f(S)=0$ if $S \in \set{\emptyset,[n]}$ and $f(S)=1$ otherwise;
it is submodular, but not a disjunction.
However, it turns out that any submodular boolean function is a $2$-DNF.
This was already known~\cite{boolean-subm},
and it can be proven by case analysis as in \Proposition{submodbinary}.
It is well known that $2$-DNFs are efficiently PAC-learnable. 
We summarize this discussion as follows.

\begin{theorem}
The class of Boolean-valued,
submodular functions is efficiently PMAC-learnable with approximation factor $1$.
\end{theorem}

\subsection{Learning under Product Distributions}
\label{appendix-learn-product}

\repeatclaim{\Lemma{useful}}{\usefullemma}

\begin{proof}
(1): Let $\hat{f} : 2^{[n] \times [\ell]} \rightarrow \bR$ be defined by
$$
\hat{f}(S_1,\ldots,S_\ell) ~=~ \sum_{i=1}^\ell \targetf(S_i).
$$
It is easy to check that $\hat{f}$ is also non-negative, monotone, submodular and $1$-Lipschitz.
We will apply \Corollary{rank-concentr} to $\hat{f}$ with $\alpha = 1/10$.
Then
\begin{align}
\nonumber
\prob{ \mu \:<\: M \log(1/\epsilon) }
    &~=~ \prob{ \smallsum{i=1}{\ell} \targetf(S_i) \:<\: M \ell \log(1/\epsilon) }
        \\\nonumber
    &~\leq~ \prob{ \Big|\hat{f}(X) - \expect{\hat{f(X)}}\Big| \:>\: \expect{\hat{f}(X)}/10 }
        \qquad\text{(since $M \leq 0.9 \cdot H$)}
        \\\nonumber
    &~\leq~ 4 \exp\big( - \expect{\hat{f}(X)} / 42200 \big)
        \\\nonumber
    &~\leq~ 4 \exp\big( - \ell / 4 \big)
        \qquad\text{(since $H > 10550$)}
        \\\EquationName{usefulineq1}
    &~\leq~ 4 \delta^3 ~\leq~ \delta/4.
\end{align}

(2): Let $\hat{f}$ and $X$ be as above.
Then
\begin{align*}
\prob{ \smallfrac{5}{6} \expect{\targetf(S)} \:\leq\: \mu \leq \smallfrac{4}{3} \expect{\targetf(S)}}
    &~\leq~ \prob{ \big|\mu-\expect{\targetf(S)}\big| \:>\: \expect{\targetf(S)}/10 } \\
    &~=~ \prob{ \big|\hat{f}(X)-\expect{\hat{f}(X)}\big| \:>\: \expect{\hat{f}(X)}/10 } \\
    &~\leq~
    \delta/4,
\end{align*}
as in \eqref{eq:usefulineq1}, since $L \geq 10550$.

(3): Set $b = K \log(1/\epsilon)$ and $t = 4 \sqrt{\log(1/\epsilon)}$.
Since $K - 4 \sqrt{K} \geq 2H$ we have
$b - t \sqrt{b} \geq 2H \log(1/\epsilon) \geq 2 \expect{\targetf(S)}$, and so
$\prob{ \targetf(S) \leq b - t \sqrt{b}} \geq 1/2$ by Markov's inequality.
By \Theorem{OURTALAGRAND}, we have
$\prob{ \targetf(S) \geq b } \leq 2 \exp( - t^2 / 4 ) \leq \epsilon$
since $\epsilon \leq 1/2$.

(4): Set $b = M \log(1/\epsilon) \ell$ and $t=4 \sqrt{\log(1/\delta)}$.
Then
\begin{align*}
b - t \sqrt{b}
    &~=~ M \log(1/\epsilon) \ell
        - 4 \sqrt{\log(1/\delta)} \sqrt{ M \log(1/\epsilon) \ell } \\
    &~>~ (M - \sqrt{M}) \log(1/\epsilon) \ell
\end{align*}
since $4 \sqrt{\log(1/\delta) \ell} \leq \ell$ and $\epsilon \leq 1/5$.
Then, by Markov's inequality,
\begin{align*}
\prob{ \smallsum{i=1}{\ell} \targetf(S_i) \leq b - t \sqrt{b}}
&~\geq~ 1 - \frac{ \expect{\smallsum{i=1}{\ell} \targetf(S_i)} }
                 {(M - \sqrt{M}) \log(1/\epsilon) \ell} \\
&~\geq~ 1 - \frac{ L \log(1/\epsilon) \ell }{(M - \sqrt{M}) \log(1/\epsilon) \ell} \\
&~\geq~ 1/20
\end{align*}
since $L / (M - \sqrt{M}) \leq 0.95$.
Applying \Theorem{OURTALAGRAND} to the submodular function
$\hat{f}(S_1,\ldots,S_\ell) = \smallsum{i=1}{\ell} \targetf(S_i)$,
we have
$\prob{ \sum_{i=1}^\ell \targetf(S_i) \geq b }
\leq 20 \exp(-t^2/4) \leq 20 \cdot \delta^4 \leq \delta/4$.
\end{proof}

\subsection{Learning Lower Bounds}
\AppendixName{mainlb}

\repeatclaim{\Theorem{mainlbadaptive}}{\thmmainlbadaptive}

\smallskip

\begin{proof}
First, consider a fully-deterministic learning algorithm $\ALG$, i.e.,
an algorithm that doesn't even sample from $D$, though
it knows $D$ and can use it in deterministically choosing queries.  Say
this algorithm makes $q < n^c$ queries (which could be chosen adaptively).
Each query has at most $n$ possible answers, since the minimum rank of
any set is zero and the maximum rank is at most $n$. So the total number of
possible sequences of answers is at most $n^q$.

Now, since the algorithm is deterministic, the hypothesis it outputs
at the end is uniquely determined by this sequence of answers.  To be
specific, its choice of the second query is uniquely determined by the
answer given to the first query, its choice of the third query is
uniquely determined by the answers given to the first two queries, and
by induction, its choice of the $i$th query $q_i$ is uniquely determined
by the answers given to all queries $q_1$,...,$q_{i-1}$ so far. Its
final hypothesis is uniquely determined by all $q$ answers.
This then implies that $\ALG$ can output at most $n^q$ different hypotheses.

We will apply \Theorem{manymatroids} with $k = 2^t$ where $t=c\log(n) + \log(\ln n) + 14$ (so $k = n^c \cdot \ln(n) \cdot 2^{14} > 10000 \cdot q \cdot \ln(n)$).
Let $\cA$ and $\cM$ be the families constructed by \Theorem{manymatroids}.
Let the underlying distribution $D$ on $2^\ground$ be the uniform distribution on $\cA$.
(Note that $D$ is not a product distribution.)
Choose a matroid $\mat_\cB \in \cM$ uniformly at random and let the target function be
$\targetf = \rankf_{\mat_\cB}$.
Let us fix a hypotheses $h$ that $\ALG$ might output. By Hoeffding bounds, we have:

$$
\Pr_{\targetf, \cS}\Bigg[~
    \Pr_{A \sim D} \Big[\:
        \targetf(A) \not \in \big[h(A), \smallfrac{n^{1/3}}{16 t} h(A)\big]
        \leq 0.49 \:\Big]
    ~\Bigg]
    ~\leq~ e^{-2 (.01)^2 k}
    ~=~ e^{-2q \cdot \ln(n)}
    ~=~ n^{-2q},
$$ i.e., with probability at least $1- n^{-2q}$, $h$ has high approximation
error on over $49\%$ of the examples.

By a union bound over all over all the $n^q$ hypotheses $h$ that $\ALG$ might
output, we obtain that with probability at least $1/4$
(over the draw of the training samples) 
the hypothesis function  output by $\ALG$ does not approximate $\targetf$ within a $o({n^{1/3}}/{\log n})$ factor on at least $1/4$ fraction of the
examples  under $D$.

The above argument is a fixed randomized strategy for the adversary that works
against any deterministic $\ALG$ making at most $n^c$ queries.  By Yao's minimax principle, this means
that, for any randomized algorithm making at most $n^c$ queries, there exists $\mat_\cB$
which the algorithm does not learn well, even with arbitrary value queries.
\end{proof}

\repeatclaim{\Corollary{crypto}}{\cryptocor}


\begin{proofof}{\Corollary{crypto}}
The argument follows Kearns-Valiant \cite{KV94}.
We will apply \Theorem{manymatroids} with $k = 2^t$ where $t = n^\epsilon$.
There exists a family of pseudorandom Boolean functions $H_t = \setst{ h_y }{ y \in \set{0,1}^t }$,
where each function is of the form $h_y : \set{0,1}^t \rightarrow \set{0,1}$.
Choose an arbitrary bijection between $\set{0,1}^t$ and $\cA$.
Then each $h_y \in H_t$ corresponds to some subfamily $\cB \subseteq \cA$,
and hence to a matroid rank function $\rankf_{\mat_\cB}$.
Suppose there is a PMAC-learning algorithm for this family of functions
which achieves approximation ratio better than $n^{1/3} / 16t$ on a set of measure $1/2 + 1/\poly(n)$.
Then this algorithm must be predicting the function $h_y$ on a set of size
$1/2 + 1/\poly(n) = 1/2 + 1/\poly(t)$.
This is impossible, since the family $H_t$ is pseudorandom.
\end{proofof}

\section{Expander Construction}
\AppendixName{expander}

%


\repeatclaim{\Theorem{ourexpanders}}{
    Let $G=(U \union V, E)$ be a random multigraph where
    $\card{U}=k$, $\card{V}=n$, and every $u \in U$
    has exactly $d$ incident edges, each of which has an endpoint chosen
    uniformly and independently from all nodes in $V$.
    Suppose that $k \geq 4$, $d \geq \log(k)/\epsilon$ and $n \geq 16 L d/\epsilon$.
    Then, with probability at least $1-2/k$,
    \begin{equation}
    \EquationName{exp2}
    |\neigh(J)| ~\geq~ (1-\epsilon) \cdot d \cdot |J| \qquad\forall J \subseteq U,\, |J| \leq L.
    \end{equation}
    If it is desired that $|\neigh(\set{u})| = d$ for all $u \in U$
    then this can be achieved by replacing any parallel edges incident on $u$
    by new edges with distinct endpoints.
    This cannot decrease $|\neigh(J)|$ for any $J$, and so \eqref{eq:exp2} remains satisfied.
}


The proof is an variant of the argument in Vadhan's survey \cite[Theorem 4.4]{Vadhan}.

\begin{proof}
Fix $j \leq L$ and consider any set $J \subseteq U$ of size $\card{J} = j$.
The sampling process decides the neighbors $\neigh(J)$
by picking a sequence of $j d$ neighbors $v_1,\ldots,v_{j d} \in V$.
An element $v_i$ of that sequence is called a \newterm{repeat} if
$v_i \in \set{v_1,\ldots,v_{i-1}}$.
Conditioned on $v_1,\ldots,v_{i-1}$, the probability that $v_i$ is a repeat is at most $j d/n$.
The set $J$ violates \eqref{eq:exp2} only if there exist more than
$\epsilon j d$ repeats.
The probability of this is at most
$$
\binom{j d}{\epsilon j d} \Big( \frac{j d}{n} \Big)^{\epsilon j d}
    ~\leq~ \Big( \frac{e}{\epsilon} \Big)^{\epsilon j d}
           \Big( \frac{j d}{n} \Big)^{\epsilon j d}
    ~\leq~ (1/4)^{\epsilon j d}.
$$
The last inequality follows from $j \leq L$ and our hypothesis $n \geq 16 L d / \epsilon$.
So the probability that there exists a $J \subseteq U$ with $j = \card{J}$
that violates \eqref{eq:exp2} is at most
$$
    \binom{k}{j} (1/4)^{-\epsilon j d}
    ~\leq~ k^j 2^{-2 \epsilon j d}
    ~=~ 2^{-j(2 \epsilon d - \log k)}
    ~\leq~ k^{-j},
$$
since $d \geq \log(k)/\epsilon$.
Therefore the probability that any $J$ with $\card{J} \leq L$
violates \eqref{eq:exp2} is at most
$$
\sum_{j \geq 1} k^{-j} ~\leq~ 2/k.
$$
%
\end{proof}

\section{Special Cases of the Matroid Construction}
\AppendixName{specialcases}

The matroid constructions of \Theorem{mainthm} and \Theorem{modified}
have several interesting special cases.

\subsection{Partition Matroids}

We are given disjoint sets $A_1,\ldots,A_k$ and values $b_1,\ldots,b_k$.
We claim that the matroid $\cI$ defined in \Theorem{mainthm} is a partition matroid.
To see this, note that $g(J) = \sum_{j \in J} b_j$, since the $A_j$'s are disjoint,
so $g$ is a modular function.
Similarly, $\card{ I \intersect A(J) }$ is a modular function of $J$.
Thus, whenever $\card{J}>1$, the constraint $\card{ I \intersect A(J) } \leq g(J)$
is redundant --- it is implied by the constraints $\card{ I \intersect A_j } \leq b_j$
for $j \in J$.
So we have
$$\cI
 ~=~ \setst{ I }{ \card{I \intersect A(J)} \leq g(J) ~~\forall J \subseteq [k] }
 ~=~ \setst{ I }{ \card{I \intersect A_j} \leq b_j ~~\forall j \in [k] },
$$
which is the desired partition matroid.

\subsection{Pairwise Intersections}
\AppendixName{pairwise}

We are given sets $A_1,\ldots,A_k$ and values $b_1,\ldots,b_k$.
We now describe the special case of the matroid construction which only considers
the pairwise intersections of the $A_i$'s.

\begin{lemma}
\LemmaName{pairwise}
Let
$ d $ be a non-negative integer such that
$d \leq \min_{i,j \in [k]} (b_i + b_j - \card{A_i \intersect A_j}).
$
Then
$$
\cI
    ~=~ \setst{ I }{ \card{I} \leq d \And \card{I \intersect A_j} \leq b_j
            ~\forall j \in [k] }
$$
is the family of independent sets of a matroid.
\end{lemma}
\begin{proof}
Note that for any pair $J=\set{i,j}$,
we have $g(J)=b_i + b_j - \card{A_i \intersect A_j}$.
Then
$$
  d
  ~\leq~ \min_{i,j \in [k]} (b_i + b_j - \card{A_i \intersect A_j})
  ~=~ \min_{J \subseteq [k] ,\, \card{J}=2} g(J),
$$
so $g$ is $(d,2)$-large.
The lemma follows from \Theorem{modified}.
\end{proof}

\subsection{Paving Matroids}

A \newterm{paving matroid} is defined to be a matroid $\mat=(V,\cI)$ of rank $m$
such that every circuit has cardinality either $m$ or $m+1$.
We will show that every paving matroid can be derived from our matroid construction
(\Theorem{modified}).
First of all, we require a structural lemma about paving matroids.

\begin{lemma}
\LemmaName{pavstruct}
Let $\mat=(V,\cI)$ be a paving matroid of rank $m$.
There exists a family $\cA = \set{A_1, \ldots, A_k} \subset 2^V$
such that
\begin{subequations}
\begin{gather}
\EquationName{pav1}
\cI ~=~ \setst{ I }{ \card{I} \leq m ~\:\wedge\:~ \card{I \intersect A_i} \leq m-1 ~\:\forall i } \\
\EquationName{pav2}
\card{A_i \intersect A_j} ~\leq~ m-2 \quad\forall i \neq j
\end{gather}
\end{subequations}
\end{lemma}

Related results can be found in Theorem 5.3.5, Problem 5.3.7 and Exercise 5.3.8 of
Frank's book \cite{Frank}.

\begin{proof}
It is easy to see that there exists $\cA$ satisfying \Equation{pav1},
since we may simply take $\cA$ to be the family of circuits which have size $m$.
So let us choose a family $\cA$ that satisfies \Equation{pav1} and minimizes $\card{\cA}$.
We will show that this family must satisfy \Equation{pav2}.
Suppose otherwise, i.e., there exist
$i \neq j$ such that $\card{A_i \intersect A_j} \geq m-1$.

\textit{Case 1:}
$r(A_i \union A_j) \leq m-1$. Then
$
\cA \setminus \set{ A_i, A_j } \union \set{ A_i \union A_j }
$
also satisfies \Equation{pav1}, contradicting minimality of $\card{\cA}$.

\textit{Case 2:}
$r(A_i \union A_j) = m$.
Observe that $r(A_i \intersect A_j) \geq m-1$ since $|A_i \intersect A_j| \geq m-1$
and every set of size $m-1$ is independent.
So we have
$$
r(A_i \union A_j) + r(A_i \intersect A_j) ~\geq~
m + (m-1) ~>~ (m-1)+(m-1) ~\geq~ r(A_i) + r(A_j).
$$
This contradicts submodularity of the rank function.
\end{proof}

For any paving matroid, \Lemma{pavstruct} implies that its independent sets can be
written in the form
$$
\cI ~=~ \setst{ I }{ \card{I} \leq m ~\:\wedge\:~ \card{I \intersect A_i} \leq m-1 ~\:\forall i },
$$
where $\card{A_i \intersect A_j} \leq m-2$ for each $i \neq j$.
This is a special case of \Theorem{modified} since we
may apply \Lemma{pairwise} with each $b_i = m-1$ and $d = m$, since
$$
    \min_{i,j \in [k]} (b_i + b_j - \card{A_i \intersect A_j})
    ~\geq~ 2(m-1) - (m-2)
    ~=~ m.
$$

\end{document}